\documentclass[12pt]{article}

\input{style}

\usepackage{amsmath, amssymb, amstext, amsfonts, mathrsfs}
\usepackage{tikz}

\begin{document}

\title{\vspace*{-3.5cm} Optimal investment with transient price impact}
 
\author{ Peter Bank\footnote{Technische Universit{\"a}t Berlin,
    Institut f{\"u}r Mathematik, Stra{\ss}e des 17. Juni 136, 10623
    Berlin, Germany, email \texttt{bank@math.tu-berlin.de}.  Financial
    support by Einstein Foundation through project ``Game options and
    markets with frictions'' is gratefully acknowledged.}
  \hspace{2ex} Moritz Vo{\ss}\footnote{University of California Santa
    Barbara, Department of Statistics \& Applied Probability, Santa
    Barbara, CA 93106-3110, USA, email \texttt{voss@pstat.ucsb.edu}.}  }
\date{\today}

\maketitle
\begin{abstract}
  We introduce a price impact model which accounts for finite market
  depth, tightness and resilience. Its coupled bid- and ask-price
  dynamics induce convex liquidity costs. We provide existence of an
  optimal solution to the classical problem of maximizing expected
  utility from terminal liquidation wealth at a finite planning
  horizon. In the specific case when market uncertainty is generated
  by an arithmetic Brownian motion with drift and the investor
  exhibits constant absolute risk aversion, we show that the resulting
  singular optimal stochastic control problem readily reduces to a
  deterministic optimal tracking problem of the optimal frictionless
  constant Merton portfolio in the presence of convex costs. Rather
  than studying the associated Hamilton-Jacobi-Bellmann PDE, we
  exploit convex analytic and calculus of variations techniques
  allowing us to construct the solution explicitly and to describe the
  free boundaries of the action- and non-action regions in the
  underlying state space. As expected, it is optimal to trade towards
  the frictionless Merton position, taking into account the initial
  bid-ask spread as well as the optimal liquidation of the accrued
  position when approaching terminal time. It turns out that this
  leads to a surprisingly rich phenomenology of possible trajectories
  for the optimal share holdings.
\end{abstract}
 
\begin{description}
\item[Mathematical Subject Classification (2010):] 91G10, 91G80,
  91B06, \\ 49K21, 35R35
\item[JEL Classification:] G11, C61
\item[Keywords:] Utility maximization, transient price impact,
  singular control, convex analysis, calculus of variations, free
  boundary problem
\end{description}


\section{Introduction}

The classical Merton problem~\cite{Mert:71},~\cite{Mert:69} of
maximizing expected utility from terminal wealth by dynamically
trading a risky asset in a financial market has by now been
intensively studied and well understood in models with market
frictions like transaction costs. We refer to the recent survey
by~\citet{MuhleKarbeReppenSoner:17} for an overview. In contrast, less
is known about utility maximization problems in illiquid market models
where the friction is induced by \emph{price impact}: The investor
trades at bid- and ask-prices which are adversely affected by the
volume or speed of her current and past trades. Within these models,
the vast majority of the existing literature is primarily concerned
with the problem of optimally executing exogenously given orders; cf.,
e.g., the surveys by~\citet{GokayRochSoner:11}
and~\citet{GatheralSchied:13}. However, regarding more complex
optimization problems such as optimal portfolio choice, explicit
characterizations of optimal strategies seem to have been elusive so
far. This is notably the case for optimal investment problems on a
finite time horizon in the presence of a bid-ask spread and price
impact that, rather than being purely temporary or fully permanent, is
\emph{transient} in the sense that the impact of the investors current
and past trades on execution prices does not vanish instantaneously
but persists and decays over time at some finite resilience rate.

Most of the currently available work on optimal portfolio choice
problems in illiquid financial markets focuses on models with purely
temporary price impact, i.e., infinite resilience, zero bid-ask
spread, and restricts to long-term investors as, e.g.,
in~\citet{GuasoniWeb:17}, \cite{GuasoniWeb:15}, \cite{GuasoniWeb:16}
with constant relative risk aversion, in~\citet{FordeWeberZhang:16}
with constant absolut risk aversion or in~\citet{GarlPeder:13.1},
\cite{GarlPeder:13.2} with mean-variance preferences. In the latter
papers, the authors also take into account finite resilience. For
investors having a finite planning horizon but still solely facing
temporary price impact, asymptotic results have been obtained
by~\citet{MoreauMuhleKarbeSoner:17} and in a more general setup
in~\citet{CayeHerdegenMuhleKarbe:17};
cf. also~\citet{ChandraPapanicolaou:17} for a pertubation analysis.  The
results from~\cite{MoreauMuhleKarbeSoner:17} are also used as a
building block to describe asymptotically optimal trading strategies
under highly resilient price impact in~\citet{KallsenMuhleKarbe:14}, or
in~\citet{EkrenMuhleKarbe:17} in the setting
of~\cite{GarlPeder:13.1}. In all the above cited papers, trading
strategies are confined to be absolutely continuous.

In the present paper, we propose a price impact model which goes
beyond the block-shaped limit order book model of~\citet{ObiWang:13} by
allowing for both selling and buying stock.  Specifically, our model
determines bid- and ask-prices via a coupled system of controlled
diffusions, giving us the possibility to specify market depth,
tightness and resilience: the three dimensions of liquidity identified
in the seminal work by~\citet{Kyle:85}. The coupled bid- and ask-price
dynamics induce convex liquidity costs on the trading strategies which
are allowed to be singular and comprise non-infinitesimal block trades
as in~\cite{ObiWang:13}. In fact, our model is closely related to the
one proposed in~\citet{RochSoner:13} which is an extension of the
illiquid market model approach introduced by~\citet{CetinJarrProt:04}
in the sense that it additionally takes into account finite resilience
and a bid-ask spread. In contrast, our model captures recovery of the
bid- and ask-prices by a reversion to each other rather than towards
some auxiliary reference price process. Moreover, our illiquidity
parameters, i.e., market depth and resilience, are constant in order
to preserve tractability.

We provide existence of an optimal solution to the corresponding
classical problem of maximizing expected utility from terminal
liquidation wealth at some finite planning horizon. In its simplest
version, our price impact model is an illiquid variant of a Bachelier
model with convex liquidity costs which are levied on the agent's
trading activity. For an investor who exhibits constant absolute risk
aversion, it turns out that the resulting singular optimal stochastic
control problem readily reduces to a deterministic optimal tracking
problem of the optimal frictionless buy-and-hold Merton portfolio in
the presence of convex costs. Instead of the more common dynamic
programming methods which lead to the challenge of solving a
three-dimensional free boundary problem induced by a
Hamilton-Jacobi-Bellman partial differential equation, we exploit a
convex analytic approach. Deriving first order conditions in terms of
the (infinite dimensional) subgradients of the convex cost functional
allows us to construct explicitly the solution to the singular control
problem by calculus of variations. As a consequence, we are able to
describe analytically the free boundaries of the buying-, selling and
a no-trading region in the underlying three-dimensional state space
for the optimally controlled dynamics of the spread and the risky
asset holdings with respect to the remaining time to maturity.

Our explicit results
make transparent how the optimal strategy has to comprise several
aspects. As already expected by the work in~\citet{GuasoniWeb:17},
\cite{GuasoniWeb:15}, \cite{GuasoniWeb:16}, \citet{FordeWeberZhang:16},
and~\citet{GarlPeder:13.1}, \cite{GarlPeder:13.2}, it is indeed optimal
to trade towards the optimal frictionless portfolio while taking into
account the initial bid-ask spread as well as the available time
horizon. Specifically, since liquidation is costly in the present
setup, the optimizer also has to take care of optimally unwinding his
accrued position when approaching terminal time. It turns out that
already in this elementary illiquid Bachelier model the interaction of
market tightness, finite resilience, desired position targeting and
optimal liquidation at a finite time horizon permits a surprisingly
rich phenomenology of possible trajectories for the optimal share
holdings. In this regard, our optimization problem is substantially
different from the infinite horizon and zero spread frameworks
considered in the papers cited above. Our findings also complement and
extend the explicit results on the optimal order execution problem as
studied in~\citet{ObiWang:13} in a similar Bachelier-type setting.

The paper most closely related to ours is~\citet{SonerVukelja:16}.
Therein, the authors adopt the model from~\citet{RochSoner:13} without
bid-ask spread in a Black-Scholes framework with constant resilience
and stochastic market depth proportional to the risky asset
price. Using the dynamic programming principle and the notion of
viscosity solutions, the problem of maximizing expected utility from
terminal liquidation wealth for CRRA investors with finite planning
horizon is studied. Compared to our results, their more general
framework comes at the cost that a characterization of the optimal
strategy is only possible numerically via a discrete-time approximation
scheme.

The rest of the paper is organized as follows. In
Section~\ref{sec:model} we introduce a price impact
model. Section~\ref{sec:problem} outlines the problem of maximizing
expected utility from terminal liquidation wealth in our model and
provides existence of an optimal solution in a general setup. In the
specific case when market uncertainty is generated by an arithmetic
Brownian motion with drift and the investor exhibits constant absolute
risk aversion, we show that the optimal singular stochastic control
problem has a deterministic solution which we construct
explicitly. This is presented in Section~\ref{sec:example}. Technical
proofs are deferred to Section~\ref{sec:proofs}.


\section{A price impact model}
\label{sec:model}

We fix a filtered probability space
$(\Omega,\cF,(\mathcal{F}_t)_{t \geq 0},\PP)$ satisfying the usual
conditions of right continuity and completeness and consider an
investor whose trades in a risky asset affect its market prices in an
adverse manner. For our specification of her price impact, we propose
a variant of the block-shaped limit order book model introduced
by~\citet{ObiWang:13}. Specifically, the investor's trading strategy is
described by a pair $X = (X^{\uparrow},X^{\downarrow})$ of
predictable, nondecreasing, right-continuous processes where
$X^{\uparrow} = (X^{\uparrow}_t)_{t\geq 0}$,
$X^{\downarrow}=(X^{\downarrow}_t)_{t\geq 0}$ denote , respectively,
the cumulative purchases and sales of the risky asset until time
$t \geq 0$. We set
$X^{\uparrow}_{0-} \triangleq X^{\downarrow}_{0-} \triangleq
0$. Trading takes place via market orders in an idealized block-shaped
limit order book at the best bid- and ask-prices $B^X$ and
$A^X$. Their dynamics are specified as the solution to the following
coupled system of controlled diffusions
\begin{align}
  \begin{aligned} \label{eq:coupledBidAsk}
    dA^X_t & = dP_t + \eta dX^{\uparrow}_t - \frac{1}{2} \kappa ( A^X_{t-} -
    B^X_{t-} ) dt, \\
    dB^X_t & = dP_t - \eta dX^{\downarrow}_t + \frac{1}{2} \kappa ( A^X_{t-} -
    B^X_{t-} ) dt
  \end{aligned}
             \qquad (t \geq 0),
\end{align}
with given parameters $\eta > 0$, $\kappa > 0$,
$A^X_{0-} \triangleq A_0 > 0$ and $B^X_{0-} \triangleq B_0 > 0$. The
interpretation of the bid- and ask-price dynamics in
\eqref{eq:coupledBidAsk} is the following: Both processes $A^X$ and
$B^X$ are driven by some common exogenous fundamental random shock
$dP_t$ modeled by a continuous semimartingale $(P_t)_{t \geq 0}$ with
initial value $P_{0-} \triangleq (A_0 + B_0)/2$. The process
$(P_t)_{t \geq 0}$ can also be regarded as the unaffected price
process. Due to finite market depth $1/\eta \in (0, \infty)$ which can
be interpreted as the height of a block-shaped limit order book, a buy
order~$dX^\uparrow_t$ incurs an impact and increases the best
ask-price $A^X$ by the amount $\eta dX^{\uparrow}_t$ whereas the best
bid-price $B^X$ is not directly affected. After completion of each buy
trade, ask- and bid-prices revert to each other at some resilience
rate $\kappa > 0$.  The effects of sell orders $dX^\downarrow_t$ on
the best bid-price $B^X$ in \eqref{eq:coupledBidAsk} are
analogous. Note that price impact is \emph{transient} and does not
vanish instantaneously but persists and decays over time at a finite
exponential rate $\kappa$. We will assume for simplicity that both
illiquidity parameters, i.e., the instantaneous price impact factor
$\eta$ as well as the resilience rate $\kappa$, are
constant. According to the bid- and ask-price dynamics in
\eqref{eq:coupledBidAsk}, the controlled evolution of the bid-ask
spread $\zeta^X_t \triangleq A_t^X-B_t^X$ is described by
\begin{equation}
  \label{eq:spreaddynamics}
  d\zeta^X_t = \eta (dX^{\uparrow}_t
  + dX^{\downarrow}_t) - \kappa \zeta^X_{t-} dt
  \qquad (t \geq 0)
\end{equation}
with initial value $ \zeta^X_{0-} \triangleq \zeta_{0} \geq 0$ and
right-continuous solution
\begin{equation}
  \label{eq:spreadsolution}
  \zeta^X_t = e^{-\kappa (t-s)} \bigg( \zeta^X_{s-} + \eta \int_{[s,t]}
  e^{\kappa (u-s)} (dX^{\uparrow}_u + dX^{\downarrow}_u) \bigg)
  \quad (0 \leq s \leq t).
\end{equation}
Let us now derive the investor's wealth dynamics corresponding to a
trading strategy $X=(X^{\uparrow},X^{\downarrow})$. First, we
associate to $X$ the self-financing portfolio process
$(\xi_t^X,\varphi_t^{X})_{t \geq 0}$ with some given initial values
$(\xi^X_{0-},\varphi^X_{0-}) \in \RR^2$ where $\xi_t^X$ denotes the
amount of cash and
$\varphi^X_t \triangleq \varphi^X_{0-} + X^{\uparrow}_t -
X^{\downarrow}_t$ the number of shares of the risky asset held at time
$t \geq 0$.  Assuming zero interest rates, the self-financing
condition dictates that the cash balance $\xi^{X}$ changes only due to
trading activity $X$, i.e., we postulate that
\begin{equation*}
  d\xi^X_t = - \left( A^X_{t-} + \frac{\eta}{2} \Delta X^{\uparrow}_t
  \right) dX^{\uparrow}_t 
  + \left( B^X_{t-} - \frac{\eta}{2} \Delta X^{\downarrow}_t \right)
  dX^{\downarrow}_t \quad (t \geq 0)
\end{equation*}
with
$\Delta X^{\uparrow,\downarrow}_t \triangleq
X^{\uparrow,\downarrow}_{t} - X^{\uparrow,\downarrow}_{t-}$,
respectively. Observe that the effective execution price to, e.g., buy
a not necessarily infinitesimal quantity of $d X^{\uparrow}_t$ shares at time~$t$ is
given by $A^X_{t-} + \eta \Delta X^{\uparrow}_t/2$
where~$\eta \Delta X^{\uparrow}_t/2$ accounts for the impact a
non-infinitesimal order incurs; cf., e.g., also~\citet{AlfFruSch:10}
or~\citet{PredShaShr:11}. Analogous considerations apply for sell orders. The investor's
total wealth at any time is now expressed in terms of the
\emph{liquidation value} of her current portfolio. That is, we define
the investor's liquidation wealth process $(V_t(X))_{t \geq 0}$
associated to her portfolio process $(\xi^X,\varphi^{X})$ with trading
strategy $X=(X^{\uparrow},X^{\downarrow})$ and initial endowment
$(\xi^X_{0-},\varphi^X_{0-}) \in \RR^2$ as
\begin{equation}
  \label{eq:defliquidwealth}
  V_t(X) \triangleq \, \xi^X_t + \frac{1}{2} ( A^X_t+B^X_t) \varphi^X_t - 
  \left( \frac{1}{2} \zeta^X_t \vert \varphi^X_t \vert  
    + \frac{\eta}{2} (\varphi^X_t)^2 \right)
  \quad (t \geq 0).
\end{equation}
We set the initial value to
$V_{0-}(X) \triangleq \xi^X_{0-} + \varphi^X_{0-}(A_0+B_0)/2 -(\zeta_0
\vert \varphi^X_{0-} \vert + \eta (\varphi^X_{0-})^2)/2$. Note that
the liquidation value $V_t(X)$ in \eqref{eq:defliquidwealth}
decomposes into two parts: The first part represents the portfolio's
book value $\xi^X_t + \varphi^X_t(A^X_t+B^X_t)/2$, where the value of
the position $\varphi^X_t$ in the risky asset is measured in terms of
the mid-quote price $(A^X_t+B^X_t)/2$. The second part
$\zeta^X_t \vert \varphi^X_t \vert/2 + \eta (\varphi^X_t)^2/2$
accounts for the corresponding liquidation costs which are incurred by
the bid-ask spread $\zeta^X_t$ as well as the instantaneous price
impact $\eta$ when unwinding in one single block trade the
$\varphi^X_t$ shares. Following lemma shows that the dynamics of the
liquidation wealth process $(V_t(X))_{t \geq 0}$ in
\eqref{eq:defliquidwealth} conveniently separate into the common
frictionless wealth and a nonnegative, convex cost functional.

\begin{Lemma}\label{lem:propliquidwealth}
  The liquidation wealth process $(V_t(X))_{t \geq 0}$ of a strategy
  $X=(X^{\uparrow},X^{\downarrow})$ defined in
  \eqref{eq:defliquidwealth} allows for the decomposition
  \begin{equation} \label{eq:liquidwealth}
    V_t(X) = V_{0-}(X) + L_{0-}(X) + \int_0^t \varphi^X_s dP_s - L_t(X) \quad (t \geq 0)
  \end{equation}
  where $(L_t(X))_{t \geq 0}$ denotes the liquidity costs
  defined as
  \begin{equation} \label{eq:liquiditycosts}
    \begin{aligned}
      L_t(X) \triangleq & ~\frac{1}{4 \eta} \Big( \eta \vert \varphi^X_t
      \vert + (\zeta^X_t - e^{-\kappa t} \zeta_0) \Big)^2 +
      \frac{1}{2} \vert \varphi^X_t \vert e^{-\kappa t} \zeta_0
      + \frac{\eta}{4}(\varphi^X_{0-})^2  \\
      & + \frac{1}{2} \int_{[0,t]} e^{-\kappa s} \zeta_0
      (dX^{\uparrow}_s + dX^{\downarrow}_s) + \frac{\kappa}{2\eta}
      \int_0^t (\zeta^X_{s-} - e^{-\kappa s} \zeta_0)^2
      ds 
    \end{aligned}
  \end{equation}
  with initial value
  $L_{0-}(X) \triangleq \zeta_0 \vert \varphi^X_{0-} \vert/2 + \eta
  (\varphi^X_{0-})^2/2$.
  In particular, for all $t \geq 0$ the functional $L_t(X)$ is convex
  in $X$ and satisfies
  \begin{equation}
    L_t(X)\geq \frac{\eta}{4} e^{-2 \kappa t} (X_t^{\uparrow} +
    X_t^{\downarrow})^2
    + \frac{\kappa \eta}{2} \int_0^t e^{-2 \kappa s} (X_s^{\uparrow} +
    X_s^{\downarrow})^2 ds \geq
    0. \label{eq:liquiditycostsbound}
  \end{equation}
\end{Lemma}    

Observe that the quantity
$V_{0-}(X) + L_{0-}(X) = \xi^X_{0-} + \varphi^X_{0-} P_{0-}$ in
\eqref{eq:liquidwealth} represents by definition the initial wealth's
book value or initial frictionless wealth of strategy $X$ with initial
endowment $(\xi^X_{0-},\varphi^X_{0-})$.

\begin{Remark}
\begin{enumerate}
\item Compared to other price impact models which are used in the
  literature in the context of optimal portfolio choice, our price
  impact in~\eqref{eq:coupledBidAsk} depends on the trading volume of
  the investor in the spirit of~\citet{ObiWang:13} and not on the
  trading rate as, e.g., in~\citet{GarlPeder:13.1},
  \cite{GarlPeder:13.2} or~\citet{FordeWeberZhang:16}. These papers
  adopt purely temporary price impact as proposed
  by~\citet{AlmgChr:01}. In~\citet{GuasoniWeb:17}, \cite{GuasoniWeb:15},
  \cite{GuasoniWeb:16} temporary price impact is not only induced by
  the trading rate but also depends on the investor's total
  wealth. Our model captures transient price impact which decays only
  gradually over time. As a consequence, trading strategies are no
  longer restricted to be absolutely continuous but also comprise
  non-infinitesimal block trades. In fact, our modeling approach is
  similar to the one proposed in~\citet{RochSoner:13} where the authors
  allow for more general stochastic dynamics for the market depth and
  the resilience rate. Another difference is that our bid- and
  ask-prices in~\eqref{eq:coupledBidAsk} revert to each other and not
  to some reference price as in~\cite{RochSoner:13}.
\item Recall that proportional transaction costs as considered, e.g.,
  in~\citet{DavisNorm:90}, are linear in the risky asset
  holdings. Temporary price impact which is linear in the trading rate
  of absolutely continuous strategies as considered
  in~\citet{GarlPeder:13.1}, \cite{GarlPeder:13.2} or
  \citet{GuasoniWeb:17}, \cite{GuasoniWeb:16} induces quadratic
  liquidity costs on the latter. The authors
  in~\citet{FordeWeberZhang:16}, \citet{GuasoniWeb:15}
  and~\citet{CayeHerdegenMuhleKarbe:17} allow for nonlinear price
  impact which introduces a dependence of the incurred trading costs
  on a fractional power of the turnover rates. In our model above,
  price impact in~\eqref{eq:coupledBidAsk} is still linear in the
  trading strategy $X=(X^\uparrow,X^\downarrow)$ but the induced
  liquidity costs in~\eqref{eq:liquiditycosts} are convex in $X$ rather than
  purely quadratic because of the emergence of the absolute value function.
\end{enumerate}
\end{Remark}


\section{Optimal investment problem}
\label{sec:problem}

We consider an investor who aims to trade optimally in the price
impact model introduced in Section~\ref{sec:model}. The investor's
preferences are described by a utility function
$u: \RR \rightarrow \RR$ in $C^1(\mathbb{R})$ which is strictly
concave, increasing and bounded from above. She wants to maximize
expected utility from her terminal liquidation wealth $V_T(X)$ at some
finite planning horizon $T>0$ as defined in~\eqref{eq:defliquidwealth}
by following a trading strategy $X=(X^{\uparrow},X^{\downarrow})$ with
given initial endowment $\xi^X_{0-} \set \xi_0 \in \mathbb{R}$ in cash
and $\varphi^X_{0-} \set \varphi_0 \in \mathbb{R}$ shares of the risky
asset. Her corresponding initial wealth and the associated liquidation
costs are denoted by $V_0 \set V_{0-}(X)$ and $L_0 \set L_{0-}(X)$ for
some given initial bid-ask spread $\zeta^X_{0-} = \zeta_0 \geq 0$. In
other words, in view of Lemma~\ref{lem:propliquidwealth}, the
agent's aim is to solve the optimization problem
\begin{equation}
  \label{eq:utilityproblem}
  \EE u(V_T(X)) = \EE u\left( V_0 + L_0 +
    \int_0^T \varphi^X_t dP_t - L_T(X) \right)
  \rightarrow \max_{X=(X^{\uparrow},X^{\downarrow}) \in \cX}
\end{equation}
over all admissible trading policies
\begin{align*}
  \cX \set & \; \Big\{ (X_t)_{t \geq 0} = (X^{\uparrow}_t,X^{\downarrow}_t)_{t \geq 0} :
             X^{\uparrow}, X^{\downarrow} \textrm{ right-continuous}, \bigg. \\ 
           & \bigg. \hspace{15pt} \textrm{predictable, nondecreasing
             processes with } X^{\uparrow}_{0-} \set X^{\downarrow}_{0-} \set 0 \Big\}.
\end{align*}
The main tool which allows us to provide existence of an optimal
strategy to the maximization problem in~\eqref{eq:utilityproblem} is
given by the following convex compactness result for processes of
finite variation.

\begin{Lemma}[\citet{Guasoni:02}, Lemma~3.4] \label{lem:compactness}
  Consider a sequence of strategies $(X^n)_{n \geq 1} \subset \cX$
  such that
  $\conv(\{ X^{\uparrow,n}_T + X^{\downarrow,n}_T : n \geq 1 \})$ is
  bounded in $L^0(\Omega,\cF,\PP)$.  Then there exists a strategy
  $X \in \cX$ and a sequence $(\tilde{X}^n)_{n \geq 1} \subset \cX$ of
  cofinal convex combinations, i.e.,
  $\tilde{X}^n \in \conv(X^n, X^{n+1}, \ldots)$ for all $n \geq 1$,
  converging to $X$ weakly on $[0,T]$:
  \begin{equation} \label{lem:compactness:conv} 
    \lim_{n \rightarrow
      \infty} \tilde{X}^{\uparrow,\downarrow,n}_t(\omega) =
    X^{\uparrow,\downarrow}_t(\omega) \quad \text{for all } t \in
    \{\Delta X^{\uparrow,\downarrow}(\omega)=0\} \cup \{T\}, \; \omega \in \Omega.
  \end{equation} 
\end{Lemma}


Another important ingredient is provided by the continuity of the
liquidation wealth $V_T(X)$ in $X \in \cX$ given
in~\eqref{eq:liquidwealth}.

\begin{Lemma} \label{lem:semicontinuity} 
  Let $T >0$ and let $(X^n)_{n \geq 1} \subset \cX$ be a sequence of
  strategies with the same initial endowment
  $(\xi^X_{0-},\varphi^X_{0-}) = (\xi_0,\varphi_0)$ such that
  $X^n \rightarrow X \in \cX$ weakly on $[0,T]$ on all of $\Omega$. Then it holds that
  \begin{equation*}
  \lim_{n \rightarrow \infty} V_T(X^n) = V_T(X) \quad
  \textrm{pointwise for all } \omega \in \Omega.
  \end{equation*}
\end{Lemma}

As a consequence, due to convexity of the liquidity cost functional
$L_T(X)$ in $X \in \cX$ by virtue of Lemma \ref{lem:propliquidwealth},
we obtain the following existence and uniqueness result for the
optimization problem in \eqref{eq:utilityproblem}.

\begin{Theorem} \label{thm:existence} 
  There exists a unique strategy
  $\hat{X}=(\hat{X}^{\uparrow},\hat{X}^{\downarrow}) \in \cX$ such
  that $\EE u(V_T(\hat{X})) \geq \EE u(V_T(X))$ for all strategies
  $X=(X^{\uparrow},X^{\downarrow}) \in \cX$.
\end{Theorem}

\begin{proof}
  Consider a maximizing sequence $(X^n)_{n \geq 1} \subset \cX$ such
  that
  \begin{equation*}
    u^* \set \sup_{X \in \cX} \EE u(V_T(X))
    = \lim_{n \rightarrow \infty} \EE u(V_T(X^n)) \in (-\infty,u(\infty)).
  \end{equation*}
  We can assume without loss of generality that the sequence
  $(X^n)_{n \geq 1}$ belongs to the level-set
  $\cL_0 := \{ X \in \cX : \EE u(V_T(X)) \geq \EE u(V_T(0)) = u(V_0 +
  L_0) \}$. Moreover, due to Lemma \ref{lem:levelsets} below, it holds
  that $\conv(\{ X^{\uparrow}_T + X^{\downarrow}_T : X \in \cL_0 \})$ is
  $L^0(\Omega,\cF,\PP)$-bounded. Hence, by virtue of the compactness
  result in Lemma~\ref{lem:compactness}, there exists a strategy
  $\hat{X} \in \cX$ and a sequence
  $(\tilde{X}^n)_{n \geq 1} \subset \cX$ of convex combinations
  $\tilde{X}^n \in \conv(X^n, X^{n+1}, \ldots)$ such that a.s.
  $\tilde{X}^n \rightarrow \hat{X}$ weakly on $[0,T]$ for $n \uparrow \infty$.
  We claim that $\hat{X}$ is the optimal solution to problem
  \eqref{eq:utilityproblem}. Indeed, since the liquidity costs are
  convex, $(\tilde{X}^n)_{n \geq 1}$ is again a maximizing
  sequence. Specifically, given a finite number of strictly positive weights
  $(\lambda^n_m)_{m \geq n}$ of $\tilde{X}^n$, we have
  \begin{equation*}
    u(V_T(\tilde{X}^n)) \geq \sum_{m \geq n} \lambda_m^n u(V_T(X^m))
  \end{equation*}
  where we also used monotonicity and concavity of $u$. Taking
  expectations and passing to the limit in the above inequality yields
  $\lim_{n \rightarrow \infty} \EE u(V_T(\tilde{X}^n)) \geq
  u^*$. Moreover, by continuity of the liquidation wealth provided
  in~Lemma~\ref{lem:semicontinuity} and Fatou's Lemma we obtain
  \begin{equation*}
    u^* \geq \EE u(V_T(\hat{X})) \geq 
    \EE \limsup_{n \rightarrow \infty} u(V_T(\tilde{X}^n))
    \geq \limsup_{n \rightarrow \infty} \EE u(V_T(\tilde{X}^n)) \geq u^*.
  \end{equation*}
  Uniqueness of the optimizer $\hat{X}$ follows from strict concavity
  of the utility function $u$ and again convexity of the liquidity
  costs.
\end{proof}


\section{Illiquid Bachelier model with exponential utility}
\label{sec:example}

Let us investigate the utility maximization problem from
terminal liquidation wealth as formulated in \eqref{eq:utilityproblem}
in the specific case when market uncertainty $dP_t$ in our price
impact model~\eqref{eq:coupledBidAsk} is generated by a Brownian
motion with drift $\mu > 0$ and volatility $\sigma > 0$. That is, we
assume that the unaffected price process $(P_t)_{t \geq 0}$ is given
by
\begin{equation} \label{eq:BachelierP}
  P_{0-} = \frac{1}{2}\left(A_0+B_0\right), \quad dP_t = \mu dt + \sigma dW_t \quad (t \geq 0)
\end{equation}
where $(W_t)_{t \geq 0}$ denotes a standard Brownian motion on the
given filtered probability space $(\Omega,\cF,(\cF_t)_{t \geq 0},\PP)$.
In addition, we assume that the inverstor's preferences are prescribed
by an exponential utility function 
\begin{equation*}
u(x)=-e^{-\alpha x} \quad (x \in \mathbb{R})
\end{equation*}
with constant absolute risk aversion parameter $\alpha > 0$. In this
setup, the optimization problem in~\eqref{eq:utilityproblem} becomes
\begin{equation}
  \label{eq:optprob}
  \EE \bigg[ - \exp\Big\{ -\alpha \Big( \mu \int_0^T \varphi^X_t dt 
  + \sigma \int_0^T \varphi^X_t dW_t - L_T(X) \Big) \Big\} \bigg] 
  \rightarrow \max_{X \in \cX}.
\end{equation}
Note that for exponential utility, the
optimal strategy in~\eqref{eq:optprob} does not depend on the
investor's initial frictionless wealth $V_0+L_0$. By virtue of
Theorem~\ref{thm:existence}, there exists a unique optimal solution to
the maximization problem in \eqref{eq:optprob} for any time horizon
$T>0$, initial position $\varphi_0 \in \mathbb{R}$ in the risky asset
and any initial bid-ask spread $\zeta_0 \geq 0$

\begin{Remark}[Frictionless case]
  It is well known in the literature that in the frictionless case
  with $\eta = \zeta_0 = 0$, i.e., $A^X = B^X = P$
  in~\eqref{eq:coupledBidAsk} and $L_T(X) = 0$
  in~\eqref{eq:liquiditycosts} for any $X \in \cX$, the optimal
  strategy $\hat{X}^0=(\hat{X}^{0,\uparrow}, \hat{X}^{0,\downarrow})$
  to problem~\eqref{eq:optprob} (with initial position $\varphi_0=0$)
  is simply a deterministic buy-and-hold-strategy given by
  \begin{equation*}
    d\hat{X}_t^{0,\uparrow}=\frac{\mu}{\alpha \sigma^2} \delta_0(dt) 
    \quad \textrm{and} \quad
    d\hat{X}_t^{0, \downarrow}=\frac{\mu}{\alpha \sigma^2} \delta_T(dt) 
    \;\; \text{on } [0,T].
  \end{equation*}
  Here, $\delta_0$ and $\delta_T$ denote the Dirac measure in $0$ and
  $T$, respectively. Put differently, the optimal frictionless share
  holdings $\varphi^0$ in the risky asset are constant and given by
  the so-called \emph{Merton portfolio}
  \begin{equation} \label{eq:Merton} \varphi_t^0 \set
    \frac{\mu}{\alpha\sigma^2} \quad (0 \leq t \leq T)
  \end{equation}
  which is acquired at time $0$ and unwound at time $T$ with, respectively, an
  initial and a final block trade.
\end{Remark}

When taking into account illiquidity frictions as in our setup, that
is, price impact induced by finite market depth as well as market
tightness imposed by the bid-ask spread, it is intuitively sensible to
expect the following: Instead of directly implementing the desired
frictionless Merton position in~\eqref{eq:Merton}, the optimal
frictional portfolio for problem~\eqref{eq:optprob} will gradually
trade towards the latter. In fact, in the presence of price impact
$\eta>0$, it turns out that problem~\eqref{eq:optprob} readily
translates into a deterministic \emph{optimal tracking problem} of the
frictionless optimal portfolio position $\varphi^0$.

\begin{Proposition} \label{prop:optproblem} 
  For given time horizon $T > 0$, initial position $\varphi_0 \in \RR$
  and initial spread $\zeta_0 \geq 0$, the optimal investment strategy
  of the maximization problem in~\eqref{eq:optprob} is deterministic
  and coincides with the minimizer of the convex cost functional
  \begin{equation}
    \label{eq:trackingproblem}
    J_T(X) \triangleq L_T(X) + \frac{\alpha \sigma^2}{2} \int_0^T
    \left( \varphi^X_t - \frac{\mu}{\alpha\sigma^2} \right)^2 dt 
    \rightarrow \min_{X \in \cX^d} 
    \end{equation}
    with
    $\cX^d \set \{ X \in \cX : X=(X^\uparrow,X^\downarrow) \textrm{
      deterministic} \}$. 
\end{Proposition}
\begin{proof}
  We give an argument similar to \citet{SchiedSchonTehranchi:10}, but
  extend it to also cover unbounded strategies. For notational convenience, let
  us define the cost functional
  \begin{equation*}
    \tilde{J}_T(X) \set L_T(X) - \mu
    \int_0^T \varphi^X_t dt + \frac{\alpha \sigma^2}{2} \int_0^T
    (\varphi^X_t)^2 dt = J_T(X) -
    \frac{\mu^2}{2\alpha\sigma^2}T 
  \end{equation*}
  for all $X \in \cX$ and let us set
  $\tilde{J}_T^* \set \inf_{X \in \cX^d} \tilde{J}_T(X)$. Next, let
  $X \in \cX$ be such that $\EE u(V_T)>-\infty$. We will
  argue below that for such $X$ the density
\begin{equation} \label{eq:densityPX}
  \begin{aligned}
    \frac{d\mathbb{P}^X}{d\mathbb{P}} & \set 
    \mathcal{E}\left( -\alpha \sigma \int_0^\cdot \varphi^X_t dW_t
    \right)_T \\ & =
    \exp \left( -\alpha \sigma \int_0^T \varphi^X_t dW_t -
      \frac{\alpha^2\sigma^2}{2} \int_0^T (\varphi^X_t)^2 dt \right)
  \end{aligned}
\end{equation}
  induces a probability measure on
  $(\Omega,\mathcal{F}_T)$. Then we can write
  \begin{equation} \label{p:prop:optproblem:3}
    \begin{aligned}
      & \EE[u(V_T(X))] \\
      &= \EE \left[ -\exp\left( -\alpha \int_0^T \varphi^X_t dP_t +
          \alpha
          L_T(X) \right) \right]  \\
      & = \EE_{\mathbb{P}^X} \left[ -\exp\left(\alpha L_T(X) -\alpha
          \mu \int_0^T \varphi^X_t dt + \frac{\alpha^2 \sigma^2}{2}
          \int_0^T
          (\varphi^X_t)^2 dt\right) \right]  \\
      & = \EE_{\mathbb{P}^X} \left[ -e^{\alpha \tilde{J}_T(X)} \right]
      \leq -e^{\alpha \tilde{ J}_T^*},
    \end{aligned}
  \end{equation}
  with equality holding true for the unique deterministic minimizer
  $X \in \cX^d$ of $\tilde{J}_T$. Thus, the maximizer of the
  right-hand side in~\eqref{p:prop:optproblem:3} over all admissible
  strategies $\cX$ which corresponds to our original problem
  in~\eqref{eq:optprob} is actually given by the deterministic
  strategy attaining the value~$\tilde{J}^*_T$. 

  It remains to verify that~\eqref{eq:densityPX} indeed defines a
  probability measure $\PP^X$ for $X \in \cX$ with $\EE u(V_T(X))>-\infty$,
  i.e., such that
  \begin{equation}\label{eq:Xintegrable}
   \EE \left[\exp\left(\alpha (L_T(X) - \int_0^T \varphi^X_t dP_t)\right)\right]<\infty.
  \end{equation}
  This will be accomplished by verifying Kazamaki's criterion for the process
  $M \set - \alpha \int_0^. \varphi^X_t \sigma dW_t$. To this end,
  observe first that we can assume without loss of generality that
  $\varphi_{0-}=\varphi_T=0$ and so, with
  $\|X\|_T \set X^{\uparrow}_T+X^{\downarrow}_T$ and $P^*_T \set
  \sup_{t \in [0,T]} |P_t|$, we can use~\eqref{eq:liquiditycostsbound}
  to estimate
  \begin{align*}
    L_T(X) -\int_0^T \varphi^X_t dP_t & = L_T(X) + \int_0^T P_t \,
             d\varphi^X_t \\
  & \geq c \|X\|_T^2 - P^*_T \|X\|_T \geq \frac{c}{2} \|X\|^2_T \text{
    on } {\{P^*_T\leq c\|X\|_T/2\}}
  \end{align*}
  for $c \set \eta e^{-2 \kappa T}/4$. With~\eqref{eq:Xintegrable} and
  the fact that $P^*_T \in L^2(\PP)$ it thus follows
  that $\|X\|_T \in L^2(\PP)$ which guarantees uniform integrability of
  $M$. Moreover, we have
  \begin{align*}
    \EE& \left[\exp\left(\frac{1}{2} M_T\right)\right]\\ 
       & = \EE \left[\exp\left(\frac{\alpha}{2}
         (L_T(X)-\int_0^T \varphi^X_t
         dP_t)\right)
         \exp\left(-\frac{\alpha}{2}(L_T(X)-\int_0^T
         \varphi^X_t \mu
         dt)\right)\right]\\
       &\leq \EE \left[\exp\left(\alpha
         (L_T(X)-\int_0^T \varphi^X_t
         dP_t)\right)\right]^{1/2} \\
& \hspace{15pt} \cdot \EE \left[\exp\left(-\alpha
         (L_T(X)-\int_0^T
         \varphi^X_t \mu
         dt)\right)\right]^{1/2} < \infty,
\end{align*}
which is finite because of~\eqref{eq:Xintegrable}
and~\eqref{eq:liquiditycostsbound}. It follows that $M$ indeed
satisfies Kazamaki's criterion.
\end{proof}

\begin{Remark} \label{rem:opttracking}
$\phantom{}$
\vspace{-.5em}
  \begin{enumerate}
  \item For deterministic strategies $X \in \cX^d$ the liquidation
    wealth $V_T(X)$ in \eqref{eq:liquidwealth} in the present illiquid
    Bachelier model is normally distributed. Hence, the maximization
    problem in \eqref{eq:optprob} and thus the minimization problem in
    \eqref{eq:trackingproblem} is equivalent to the problem of
    maximizing a mean-variance criterion given by
    \begin{align*}
      & \mathbb{E}[V_T(X)] -\frac{\alpha}{2} \text{var}(V_T(X)) \\
& = V_0 + L_0 + \mu
        \int_0^T \varphi^X_t dt - L_T(X) - \frac{\alpha \sigma^2}{2} \int_0^T
        (\varphi^X_t)^2 dt,
    \end{align*}
    cf. also the discussion in \citet{SchiedSchonTehranchi:10}.
  \item The minimization problem in \eqref{eq:trackingproblem} can be
    regarded as a deterministic optimal tracking problem of the
    frictionless Merton portfolio
    $\varphi^0 \equiv \mu/(\alpha\sigma^2)$ in the presence of trading
    costs measured by $L_T(\cdot)$. That is, the optimal
    strategy~$\hat{X}$ seeks to minimize both the squared deviation of
    its share holdings $\varphi^{\hat{X}}$ from the preferred constant
    position $\varphi^0$ of~\eqref{eq:Merton} as well as the
    incurred liquidity costs $L_T(\hat{X})$ which are levied on its
    trading activity $\hat{X}=(\hat{X}^\uparrow,\hat{X}^\downarrow)$
    due to market tightness and finite market depth. In addition,
    liquidation is costly in the current setup. Therefore, besides
    trading towards~$\varphi^0$, the optimizer also has to take into
    account unwinding the accrued position in the risky asset in an
    optimal manner when approaching terminal time $T$.
  \item The deterministic optimal tracking problem
    in~\eqref{eq:trackingproblem} is similiar to the stochastic
    tracking problem studied in~\citet{BankSonerVoss:17}
    (cf. also~\citet{BankVoss:18} for a more general
    framework). Therein, the authors investigate the problem of
    minimizing the $L^2(\mathbb{P} \otimes dt)$-distance of a
    portfolio process $\varphi^X$ from a given predictable stochastic
    target process $(\xi_t)_{0 \leq t \leq T}$ in the presence of
    temporary price impact as in~\citet{AlmgChr:01}. This means that
    investment strategies $\varphi^X$ are restricted to be absolutely
    continuous and quadratic costs are levied on the respective
    trading rates~$\dot\varphi^X$.  The process
    $(\xi_t)_{0 \leq t \leq T}$ represents, e.g., an optimal
    investment or hedging strategy adopted from a frictionless
    setting. In the current setup in~\eqref{eq:trackingproblem},
    liquidity costs $L_T(\cdot)$ are induced by market tightness and
    transient price impact \`a la~\citet{ObiWang:13} and strategies
    are allowed to be singular.
  \end{enumerate}
\end{Remark}


\subsection{First order optimality conditions}
\label{subsec:foc}

Since the objective functional $J_T(\cdot)$ of the minimization
problem in Proposition~\ref{prop:optproblem} is convex, tools from
convex analysis and calculus of variations can be employed to derive
a characterization of the optimal solution in terms of sufficient
first order conditions. Specifically, let us note that the convex
functional $J_T(\cdot)$ is supported on $\cX^d$ by the
infinite-dimensional \emph{buy-} and \emph{sell-subgradients} defined
as
\begin{align}
  {}^{\varrho}\nabla_t^{\uparrow} J_T(X) \set 
  & \int_t^T \left( \kappa e^{-\kappa (u-t)} \zeta_u^X 
    + \alpha \sigma^2 \left( \varphi^X_u - \frac{\mu}{\alpha\sigma^2}
    \right) \right) du \nonumber \\
  &  + \frac{1}{2} \left( \zeta_T^X + \eta \vert \varphi^X_T \vert
    \right) e^{-\kappa (T-t)} \label{eq:buysubgradient}  \\
  & + \frac{\eta}{2} \varphi^X_T + \frac{1}{2} \sign_\varrho(\varphi^X_T) \zeta^X_T
  \qquad (0 \leq t \leq T) \nonumber\\
  \intertext{and}
  {}^{\varrho}\nabla_t^{\downarrow} J_T(X) \set 
  & \int_t^T \left( \kappa e^{-\kappa (u-t)} \zeta_u^X 
    + \alpha \sigma^2 \left( \frac{\mu}{\alpha\sigma^2} -
    \varphi^X_u\right) \right) du \nonumber \\
  &  + \frac{1}{2} \left( \zeta_T^X + \eta \vert \varphi^X_T \vert
    \right) e^{-\kappa (T-t)} \label{eq:sellsubgradient}  \\
  & - \frac{\eta}{2} \varphi^X_T - \frac{1}{2} \sign_\varrho(\varphi^X_T) \zeta^X_T
  \qquad (0 \leq t \leq T) \nonumber
\end{align}
in the sense of Lemma \ref{lem:subgradients} below. 

\begin{Remark} \label{rem:sign} The map $x \mapsto \sign_\varrho(x)$
  appearing in the definition of the buy- and sell-subgradients in
  \eqref{eq:buysubgradient} and \eqref{eq:sellsubgradient} 
  represents the subgradient of the absolute value function
  $x \mapsto \vert x \vert$ (cf. proof of Lemma~\ref{lem:subgradients}
  in Section~\ref{sec:proofs}) and therefore allows for an arbitrary
  value $\sign_\varrho(0) \set \varrho \in [-1,1]$ when
  $\varphi^X_T = 0$. In this case the subgradients are actually
  set-valued. The dependence on the value $\varrho$ is indicated by
  the left-hand superscript in the operator symbols
  ${}^{\varrho}\nabla^{\uparrow}$
  and~${}^{\varrho}\nabla^{\downarrow}$. To alleviate notation, we
  will simply write $\nabla^{\uparrow}$, $\nabla^{\downarrow}$ and
  $\sign(\cdot)$ most of the time unless a specification of the
  value~$\varrho$ becomes necessary.
\end{Remark}

\begin{Lemma} \label{lem:subgradients}
  For any two strategies $X, Y \in \mathcal{X}^d$ with the same initial
  position $\varphi^{Y}_{0-}=\varphi^{X}_{0-}$ and initial spread
  $\zeta_0 \geq 0$ and for any $\varrho \in [-1,1]$, we have
  \begin{equation*}
    J_T(Y) - J_T(X) \geq \int_{[0,T]} {}^{\varrho}\nabla_t^{\uparrow}
    J_T(X) (dY_t^{\uparrow} - dX_t^{\uparrow}) 
    + \int_{[0,T]} {}^{\varrho}\nabla_t^{\downarrow} J_T(X)
    (dY_t^{\downarrow} - dX_t^{\downarrow})
  \end{equation*}
  with ${}^{\varrho}\nabla^{\uparrow} J_T(X)$ and
  ${}^{\varrho}\nabla^{\downarrow} J_T(X)$ as defined
  in~\eqref{eq:buysubgradient} and~\eqref{eq:sellsubgradient},
  respectively.
\end{Lemma} 

For any nondecreasing, right-continuous process $Z$ with
$Z_{0-} \set 0$, let us further define the set
\begin{equation} \label{def:actionset}
  \{ dZ > 0 \} \set \{t \in [0,T]: Z_{t-} < Z_u \text{ for all } u > t \}
\end{equation}
and observe that for any continuous $G=(G_t)_{0 \leq t \leq T}$ we
have
\begin{equation*}
\int_0^T G_t \, dZ_t = \int_{\{ dZ > 0 \}} G_t \, dZ_t.
\end{equation*}
Having at hand the subgradients in~\eqref{eq:buysubgradient}
and~\eqref{eq:sellsubgradient}, we can now formulate sufficient first
order optimality conditions for the minimization problem stated in
Proposition~\ref{prop:optproblem}.

\begin{Proposition}[First order conditions] \label{prop:foc} 
  The strategy $\hat{X}=(\hat{X}^{\uparrow},\hat{X}^{\downarrow})$ in
  $\cX^d$ solves the optimization problem
  in~\eqref{eq:trackingproblem} if the following conditions hold true:
  \begin{itemize}
  \item[(i)] $\nabla_t^{\uparrow} J_T(\hat{X}) \geq 0$ for all
    $t \in [0,T]$ with `=' on the set
    $\{ d\hat{X}^{\uparrow} > 0 \}$,
  \item[(ii)] $\nabla_t^{\downarrow} J_T(\hat{X}) \geq 0$ for all
    $t \in [0,T]$ with `=' on the set
    $\{ d\hat{X}^{\downarrow} > 0 \}$.
  \end{itemize}
  In case $\varphi^{\hat{X}}_T = 0$, the conditions in (i) and (ii)
  are meant to hold for ${}^{\varrho}\nabla^{\uparrow}$ and
  ${}^{\varrho}\nabla^{\downarrow}$ with some $\varrho \in [-1,1]$.
\end{Proposition}

\begin{proof}
  Assume that $\hat{X}=(\hat{X}^{\uparrow},\hat{X}^{\downarrow})$
  satisfies conditions (i) and (ii) (for some suitable
  $\varrho \in [-1,1]$ in case $\varphi^{\hat{X}}_T = 0$) and let
  $Y \in \cX^d$ be an arbitrary competing strategy with the same initial
  endowment $\varphi^Y_{0-} = \varphi^{\hat{X}}_{0-}$. Then, by virtue
  of Lemma~\ref{lem:subgradients} above, it holds that
  \begin{align*}
    J_T(Y) - J_T(\hat{X})  \geq 
    & \int_{[0,T]} {}^\varrho\nabla_t^{\uparrow} J_T(\hat{X}) dY_t^{\uparrow} 
      + \int_{[0,T]} {}^\varrho\nabla_t^{\downarrow} J_T(\hat{X}) dY_t^{\downarrow} \\
    & -\int_{[0,T]} {}^\varrho\nabla_t^{\uparrow} J_T(\hat{X}) d\hat{X}_t^{\uparrow} 
      - \int_{[0,T]} {}^\varrho\nabla_t^{\downarrow} J_T(\hat{X})
      d\hat{X}_t^{\downarrow} .
  \end{align*}
  By our assumptions (i) and (ii) the right-hand side is nonnegative
  which implies $J_T(Y) \geq J_T(\hat{X})$.
\end{proof}

\begin{Remark}
  In view of Lemma~\ref{lem:subgradients}, the quantities
  ${}^\varrho\nabla_t^{\uparrow} J_T(X)$ and
  ${}^\varrho\nabla_t^{\downarrow} J_T(X)$
  in~\eqref{eq:buysubgradient} and~\eqref{eq:sellsubgradient} can be
  regarded as (lower bounds for) the marginal costs which are incurred
  by an additional infinitesimal buy order and sell order at time~$t$,
  respectively, otherwise following strategy $X$. Hence, an optimal
  strategy $\hat{X}$ which satisfies the first order conditions in
  Proposition~\ref{prop:foc} acts so as to keep these additional
  marginal costs from intervention always nonnegative and only
  intervenes, i.e., buys or sells the risky asset, when the
  corresponding marginal costs
  ${}^\varrho\nabla_t^{\uparrow} J_T(\hat{X})$ or
  ${}^\varrho\nabla_t^{\downarrow} J_T(\hat{X})$ vanish.  In this
  regard, observe that, loosely speaking, the subgradients
  in~\eqref{eq:buysubgradient} and~\eqref{eq:sellsubgradient} at time
  $t$ can be interpreted as assessing for the future period $[t,T]$
  the trade-off between deviating from the target
  $\mu/(\alpha\sigma^2)$, the incurred spread $\zeta^X$ as well as the
  magnitude of the final position $\varphi^X_T$.
\end{Remark}

Due to market tightness, it is intuitively sensible to expect that an
optimal strategy satisfying the first order conditions in
Proposition~\ref{prop:foc} will never purchase and sell the risky
asset at the same time. In fact, this holds true in our setting and is
a direct consequence of the structure of the subgradients.

\begin{Lemma} \label{lem:nobuysell} For any strategy $X \in \cX^d $,
  $X \neq (0,0)$, we have
  $$
   \{ \nabla_.^{\uparrow} J_T(X) = 0\} \subset
   \{\nabla_.^{\downarrow} J_T(X) > 0\}
  \; \text{and} \;
   \{ \nabla_.^{\downarrow} J_T(X) = 0\} \subset
   \{\nabla_.^{\uparrow} J_T(X) > 0\}.
  $$
\end{Lemma}

\begin{Remark}[Dynamic programming principle] \label{rem:dynamicprog}
  Note that for any strategy
  $X = (X^{\uparrow},X^{\downarrow}) \in \cX^d$ the subgradients of
  our functional $J_T(\cdot)$ in~\eqref{eq:buysubgradient}
  and~\eqref{eq:sellsubgradient} at time $t \in [0,T]$ only depend on
  the values $\varphi^X_{t-}$, $X^{\uparrow}_{t-}$,
  $X^{\downarrow}_{t-}$, $\zeta^X_{t-}$, the remaining time to
  maturity $T-t$ and the future evolution of the strategy
  $(X_u)_{t \leq u \leq T}$. This property together with the
  uniqueness of the optimal solution to
  problem~\eqref{eq:trackingproblem} implies that the dynamic
  programming principle (or so-called Bellman optimality) holds true
  in our setting. Specifically, let $\hat{X} \in \cX^d$ denote the
  unique optimal strategy for problem~\eqref{eq:trackingproblem} with
  time horizon $T>0$, initial position
  $\varphi^{\hat{X}}_{0-} = \varphi \in \RR$ and initial spread
  $\zeta^{\hat{X}}_{0-} = \zeta \geq 0$ which satisfies the first
  order conditions in Proposition~\ref{prop:foc}. From now on, we will
  use the notation
  $\hat{X}^{T,\zeta,\varphi}=(\hat{X}^{T,\zeta,\varphi,\uparrow},
  \hat{X}^{T,\zeta,\varphi,\downarrow})$ to emphasize the dependence
  of the optimal control on the \emph{problem data}
  $(T,\zeta,\varphi)$. Then for any $0 \leq t < T$ we have that the
  strategy
  \begin{equation*}
    \hat{X}^{T-t, \zeta^{\hat{X}}_{t-},\varphi^{\hat{X}}_{t-}}_s \set \hat{X}^{T,\zeta,\varphi}_{t+s}
    - \hat{X}^{T,\zeta,\varphi}_{t-} \qquad (0 \leq s \leq T-t)
  \end{equation*}
  is optimal for problem~\eqref{eq:trackingproblem} with problem data
  $(T-t, \zeta^{\hat{X}}_{t-},\varphi^{\hat{X}}_{t-})$, i.e., time
  horizon $T-t>0$, initial spread $\zeta^{\hat{X}}_{t-} \geq 0$ and
  initial position $\varphi^{\hat{X}}_{t-} \in \mathbb{R}$. Indeed,
  observe that
  \begin{equation*} 
    \nabla_s^{\uparrow,\downarrow} J_{T-t}(\hat{X}^{T-t,
      \zeta^{\hat{X}}_{t-},\varphi^{\hat{X}}_{t-}}) =
    \nabla_{t+s}^{\uparrow,\downarrow}
    J_{T}(\hat{X}^{T,\zeta,\varphi}) \quad (0 \leq s \leq T-t)
  \end{equation*}
  holds true which implies that
  $\hat{X}^{T-t, \zeta^{\hat{X}}_{t-},\varphi^{\hat{X}}_{t-}}$
  satisfies the first order conditions in Proposition~\ref{prop:foc}
  and is thus optimal.
\end{Remark}


\subsection{The state space}
\label{subsec:construction}

We want to solve the optimization problem formulated in
\eqref{eq:trackingproblem} for any given problem data
$(T,\zeta_0,\varphi_0)$, i.e., for any time horizon $T$, initial
spread $\zeta_0$ and initial position $\varphi_0$ in the risky
asset. For this purpose, let us introduce the three-dimensional
\emph{state space}
\begin{equation} \label{def:statespace}
  \cS \set \{ (\tau,\zeta,\varphi) : \tau \geq 0, \zeta \geq 0, \varphi \in \RR\} \subset \mathbb{R}^3
\end{equation}
with time to maturity $\tau$, spread $\zeta$ and number of shares
$\varphi$. For any triplet or problem data $(\tau,\zeta,\varphi)$ in
the state space $\cS$ we want to identify the corresponding unique
optimal strategy $\hat{X}^{\tau,\zeta,\varphi}$ with
$\varphi^{\hat{X}^{\tau,\zeta,\varphi}}_{0-} = \varphi$ and
$\zeta^{\hat{X}^{\tau,\zeta,\varphi}}_{0-} = \zeta$ which minimizes
the functional $J_{\tau}(\cdot)$ in~\eqref{eq:trackingproblem} for
time horizon~$\tau$ (cf. Remark \ref{rem:tauzero} below for our
convention in the special case $\tau = 0$). More precisely, we want to
describe the evolution of the optimally controlled system
$(\tau-t,\zeta^{\hat{X}^{\tau,\zeta,\varphi}}_{t},
\varphi^{\hat{X}^{\tau,\zeta,\varphi}}_{t})_{0 \leq t \leq \tau}$ in
the state space $\cS$. Intuitively, the first order optimality
conditions formulated in Proposition~\ref{prop:foc} suggest a
separation of the state space $\cS$ into two action regions -- a
\emph{buying}- and a \emph{selling}-region -- as well as a non-action
or \emph{waiting}-region for the optimizer
$\hat{X}^{\tau,\zeta,\varphi}$. Loosely speaking, depending on whether
the optimally controlled triplet
$(\tau-t,\zeta^{\hat{X}^{\tau,\zeta,\varphi}}_{t},\varphi^{\hat{X}^{\tau,\zeta,\varphi}}_{t})$
at time $t \in [0,\tau]$ is located in the buying-, selling- or
waiting-region, the corresponding optimal strategy
$\hat{X}^{\tau,\zeta,\varphi}$ buys, sells or does not do anything,
respectively, at this time instant $t$. In fact,
Proposition~\ref{prop:foc}, Lemma \ref{lem:nobuysell} as well as
Remark \ref{rem:dynamicprog} motivate the following definition of the
buying-, selling- and waiting-region.

\begin{Definition}[Buying-, selling-,
  waiting-region] \label{def:region}
$\phantom{}$
\vspace{-.5em}
  \begin{enumerate}
  \item We define the \emph{buying-region} as
    \begin{equation}
      \begin{aligned} \label{def:buyregion}
        \hspace{-20pt}\mathcal{R}_{\mathrm{buy}} \set \Big\{ (\tau,\zeta,\varphi)
        \in
        \cS : & \; \text{the optimal strategy } \hat{X}^{\tau,\zeta,\varphi} \in \cX^d \text{
          satisfies } \Big. \\
        & \; {}^\varrho\nabla_0^{\uparrow} J_\tau(\hat{X}^{\tau,\zeta,\varphi}) =
        0 \text{ for some } \varrho \\
        & \Big. \text{ and } \hat{X}^{\tau,\zeta,\varphi,\uparrow}_0 > 0 \Big\}
      \end{aligned}
    \end{equation}
    and the boundary of the buying-region as
    \begin{equation}
      \begin{aligned} \label{def:buybound}
        \hspace{-25pt}\partial \mathcal{R}_{\mathrm{buy}} \set 
        \Big\{ (\tau,\zeta,\varphi)
        \in
        \cS : & \; \text{the optimal strategy } \hat{X}^{\tau,\zeta,\varphi} \in \cX^d \text{
          satisfies} \Big. \\
        & \; {}^\varrho\nabla_0^{\uparrow} J_\tau(\hat{X}^{\tau,\zeta,\varphi}) =
        0 \text{ for some } \varrho \\
        & \Big. \text{ and } \hat{X}^{\tau,\zeta,\varphi,\uparrow}_0 = 0 \Big\}.
      \end{aligned}
    \end{equation}
  \item We define the \emph{selling-region} as
    \begin{equation}
      \begin{aligned} \label{def:sellregion}
        \hspace{-20pt}\mathcal{R}_{\mathrm{sell}} \set \Big\{ (\tau,\zeta,\varphi)
        \in
        \cS : & \; \text{the optimal strategy } \hat{X}^{\tau,\zeta,\varphi} \in \cX^d \text{
          satisfies} \Big. \\
        & \; {}^\varrho\nabla_{0}^{\downarrow} J_\tau(\hat{X}^{\tau,\zeta,\varphi}) =
        0 \text{ for some } \varrho \\
        & \Big. \text{ and } \hat{X}^{\tau,\zeta,\varphi,\downarrow}_0 > 0 \Big\}
      \end{aligned}
    \end{equation}
    and the boundary of the selling-region as
    \begin{equation}
      \begin{aligned} \label{def:sellbound}
        \hspace{-22pt}\partial \mathcal{R}_{\mathrm{sell}} \set \Big\{
        (\tau,\zeta,\varphi) \in
        \cS : & \; \text{the optimal strategy } \hat{X}^{\tau,\zeta,\varphi} \in \cX^d \text{
          satisfies} \Big. \\
        & \; {}^\varrho\nabla_0^{\downarrow}
        J_\tau(\hat{X}^{\tau,\zeta,\varphi}) = 0 \text{ for some }
        \varrho \\
        & \text{ and }
        \hat{X}^{\tau,\zeta,\varphi,\downarrow}_0 = 0 \Big\}.
      \end{aligned}
    \end{equation}
  \item We define the \emph{waiting-region} as
    \begin{equation} \label{def:waitregion}
      \mathcal{R}_{\mathrm{wait}} \set \cS \backslash
      (\bar{\mathcal{R}}_{\mathrm{buy}} \cup \bar{\mathcal{R}}_{\mathrm{sell}})
  \end{equation}
where $\bar{\mathcal{R}}_{\mathrm{buy}/\mathrm{sell}} \set
\mathcal{R}_{\mathrm{buy}/\mathrm{sell}} \cup \partial
\mathcal{R}_{\mathrm{buy}/\mathrm{sell}}$, respectively.
\end{enumerate}
\end{Definition}

\begin{Remark} \label{rem:tauzero}
$\phantom{}$
\vspace{-.5em}
  \begin{enumerate}
  \item Lemma~\ref{lem:nobuysell} implies 
    $\partial \mathcal{R}_{\mathrm{buy}} \cap \partial
    \mathcal{R}_{\mathrm{sell}} = \varnothing$. Moreover, it will  become
    clear in Theorem~\ref{thm:main} that these boundaries coincide with
    the topological ones.
  \item By definition in~\eqref{def:statespace} problem data or
    triplets $(0,\zeta,\varphi)$ with $\tau = 0$ also belong to the
    state space $\cS$. Hence, we have to find a convention for how to
    specify the associated optimal strategies
    $\hat{X}^{0,\zeta,\varphi}$ with
    $\varphi_{0-}^{\hat{X}^{0,\zeta,\varphi}} = \varphi$ and
    $\zeta^{\hat{X}^{0,\zeta,\varphi}}_{0-} = \zeta$. In view of the
    subgradients in \eqref{eq:buysubgradient} and
    \eqref{eq:sellsubgradient} we have
    \begin{equation*}
      {}^{\varrho}\nabla_0^{\uparrow,\downarrow}
      J_0(\hat{X}^{0,\zeta,\varphi}) 
      = \frac{1}{2}
      \left( \eta \vert \varphi \vert + \zeta
      \right) \pm \frac{\eta}{2} \varphi \pm \frac{1}{2} \sign_\varrho(\varphi)
      \zeta.
    \end{equation*}
    Thus, in case $\varphi > 0$ it holds that
    $\nabla_0^{\uparrow} J_0(\hat{X}^{0,\zeta,\varphi}) > 0$ and
    $\nabla_0^{\downarrow} J_0(\hat{X}^{0,\zeta,\varphi}) =
    0$. Therefore, we stipulate that the associated optimal strategy
    $\hat{X}^{0,\zeta,\varphi}$ is given by
    $\hat{X}_0^{0,\zeta,\varphi,\uparrow} \set 0$ and
    $\hat{X}_0^{0,\zeta,\varphi,\downarrow} \set \varphi > 0$, i.e.,
    it unwinds with a single block sell order the position
    $\varphi$. Analogously, in case $\varphi < 0$ we have
    $\nabla_0^{\uparrow} J_0(\hat{X}^{0,\zeta,\varphi}) = 0$ and
    $\nabla_0^{\downarrow} J_0(\hat{X}^{0,\zeta,\varphi}) > 0$ and
    thus we set $\hat{X}_0^{0,\zeta,\varphi,\downarrow} \set 0$ as
    well as $\hat{X}^{0,\zeta,\varphi,\uparrow}_0 = -\varphi>0$, i.e.,
    the optimal strategy clears out its short position by executing a
    single block buy order. In case $\varphi = 0$, we have
    \begin{equation*}
      {}^{\varrho}\nabla_0^{\uparrow,\downarrow}
      J_0(\hat{X}^{0,\zeta,0}) = \frac{1}{2} \zeta \pm
      \frac{1}{2} \varrho \zeta \geq 0
    \end{equation*}
    for all $\varrho \in [-1,1]$. We make the convention that the
    associated optimal strategy is simply defined as
    $\hat{X}^{0,\zeta,0}_0 \set (0,0)$.
  \item Note that our convention in 2.) together with the dynamic
    programming principle from Remark~\ref{rem:dynamicprog} entails
    that any optimal strategy $\hat{X}^{\tau,\zeta,\varphi}$ with a
    final position
    $\varphi^{\hat{X}^{\tau,\zeta,\varphi}}_{\tau} \neq 0$ in the
    risky asset in fact unwinds its remaining shares with a single
    block order.
\end{enumerate}
\end{Remark}


\subsection{Main result}
\label{subsec:mainresult}

Our main result is an explicit description of the buying- and
selling-region $\cR_{\mathrm{buy}}$ and $\cR_{\mathrm{sell}}$ in the
state space $\cS$ defined in
Definition~\ref{def:region}. Specifically, it turns out that the free
boundaries $\partial \cR_{\mathrm{buy}}$ and
$\partial \cR_{\mathrm{sell}}$ can be described analytically as the
graph of two \emph{free boundary functions}
$(\tau,\zeta) \mapsto \phi_{\mathrm{buy}}(\tau,\zeta)$ and
$(\tau,\zeta) \mapsto \phi_{\mathrm{sell}}(\tau,\zeta)$ defined on the
time-to-maturity and spread domain $[0,+\infty)^2$. All the results in
this section will be proved in Section~\ref{sec:proofs}.

\begin{Theorem}\label{thm:main}
  For the two functions
  \begin{equation} \label{eq:phisellphibuy}
    \phi_{\mathrm{buy}}(\tau,\zeta) < \phi_{\mathrm{sell}}(\tau,\zeta)
  \end{equation}
  defined in~\eqref{def:phisell} and~\eqref{def:phibuy:eq1} --
  \eqref{def:phibuy:eq3c} below, we have
  \begin{align}
    \mathcal{R}_{\mathrm{sell}} 
    & = \{ (\tau,\zeta,\varphi) \in
      \cS : \varphi > \phi_{\mathrm{sell}}(\tau,\zeta) \}, \label{main:sellreg} \\
    \partial \mathcal{R}_{\mathrm{sell}} 
    & = \{ (\tau,\zeta,\varphi)
      \in \cS : \varphi = \phi_{\mathrm{sell}}(\tau,\zeta) \} \label{main:sellbound}
  \end{align}
  as well as
  \begin{align}
    \mathcal{R}_{\mathrm{buy}} 
    & = \{ (\tau,\zeta,\varphi) \in
      \cS : \varphi < \phi_{\mathrm{buy}}(\tau,\zeta)
      \}, \label{main:buyreg} \\
    \partial \mathcal{R}_{\mathrm{buy}} 
    & = \{ (\tau,\zeta,\varphi)
      \in \cS : \varphi =   \phi_{\mathrm{buy}}(\tau,\zeta) \}. \label{main:buybound}
  \end{align}
  In particular, it holds that
  \begin{align}
    \mathcal{R}_{\mathrm{wait}} 
    & = \{ (\tau,\zeta,\varphi) \in
      \cS :
      \phi_{\mathrm{buy}}(\tau,\zeta) <
      \varphi <
      \phi_{\mathrm{sell}}(\tau,\zeta) \}, \label{main:waitreg} \\
      \partial \mathcal{R}_{\mathrm{wait}} 
      & = \partial \mathcal{R}_{\mathrm{buy}}
        \cup \partial \mathcal{R}_{\mathrm{sell}}. \label{main:waitbound}
  \end{align}
\end{Theorem}

In fact, the behaviour of optimal strategies with initial problem data
$(\tau,\zeta,\varphi)$ in $\mathcal{R}_{\mathrm{buy}}$,
$\mathcal{R}_{\mathrm{sell}}$, or $\mathcal{R}_{\mathrm{wait}}$ can be
readily deduced from the definition of the buying-, selling- and
waiting-region in~\eqref{def:buyregion}, \eqref{def:sellregion}
and~\eqref{def:waitregion}, together with the dynamic programming
principle from Remark~\ref{rem:dynamicprog}.

\begin{Remark} \label{rem:main1}
$\phantom{}$
\vspace{-.5em}
  \begin{enumerate}
  \item For each problem data
    $(\tau,\zeta,\varphi) \in \mathcal{R}_{\mathrm{sell}}$, i.e.,
    $\varphi > \phi_{\mathrm{sell}}(\tau,\zeta)$ in view
    of~\eqref{main:sellreg} in Theorem~\ref{thm:main}, it follows from
    the definition of $\mathcal{R}_{\mathrm{sell}}$ and
    $\partial \mathcal{R}_{\mathrm{sell}}$ in~\eqref{def:sellregion}
    and~\eqref{def:sellbound} that the optimal strategy
    $\hat{X}^{\tau,\zeta,\varphi}$ will actually ``jump'' with an
    initial impulse block sell order of size
    $\hat{X}_0^{\tau,\zeta,\varphi,\downarrow} = x^{\downarrow} > 0$
    satisfying the equation
    \begin{equation} \label{eq:impulsesell} \varphi - x^{\downarrow} =
      \phi_{\mathrm{sell}}(\tau,\zeta + \eta x^{\downarrow})
    \end{equation}
    to the triplet
    $(\tau,\zeta + \eta x^{\downarrow}, \varphi - x^{\downarrow})$
    which belongs to $\partial \mathcal{R}_{\mathrm{sell}}$ by virtue
    of~\eqref{main:sellbound}. Thereafter, it coincides with the
    corresponding optimal strategy
    $\hat{X}^{\tau, \zeta + \eta x^{\downarrow}, \varphi -
      x^{\downarrow}}$ which does satisfy
    $\hat{X}_0^{\tau, \zeta + \eta x^{\downarrow}, \varphi -
      x^{\downarrow},\downarrow} =0$ in line with the definition of
    $\partial \mathcal{R}_{\mathrm{sell}}$ in~\eqref{def:sellbound};
    cf. proof of Theorem~\ref{thm:main} below. Similarly, for each
    problem data
    $(\tau,\zeta,\varphi) \in \mathcal{R}_{\mathrm{buy}}$, i.e.,
    $\varphi < \phi_{\mathrm{buy}}(\tau,\zeta)$ in view
    of~\eqref{main:buyreg} in Theorem~\ref{thm:main}, the optimal
    strategy $\hat{X}^{\tau,\zeta,\varphi}$ will ``jump'' with an
    initial impulse block buy order of size
    $\hat{X}^{\tau,\zeta,\varphi,\uparrow} = x^{\uparrow} > 0$
    satisfying the equation
    \begin{equation} \label{eq:impulsebuy} \varphi + x^{\uparrow} =
      \phi_{\mathrm{buy}}(\tau,\zeta + \eta x^{\uparrow})
    \end{equation}
    to the triplet
    $(\tau,\zeta + \eta x^{\uparrow}, \varphi +x^{\uparrow})$ in
    $\partial \mathcal{R}_{\mathrm{buy}}$ by virtue
    of~\eqref{main:buybound} and then will coincide with the
    corresponding optimal strategy
    $\hat{X}^{\tau, \zeta + \eta x^{\uparrow}, \varphi +
      x^{\uparrow}}$. Again it will hold that
    $\hat{X}_0^{\tau, \zeta + \eta x^{\uparrow}, \varphi +
      x^{\uparrow},\uparrow} =0$ in line with the definition of
    $\partial \mathcal{R}_{\mathrm{buy}}$ in~\eqref{def:buybound}.
  \item For any problem data
    $(\tau,\zeta,\varphi) \in \mathcal{R}_{\mathrm{wait}}$, i.e.,
    $\phi_{\mathrm{buy}}(\tau,\zeta) < \varphi <
    \phi_{\mathrm{sell}}(\tau,\zeta)$ in view of~\eqref{main:waitreg}
    in Theorem~\ref{thm:main}, the optimal strategy
    $\hat{X}^{\tau,\zeta,\varphi}$ will remain inactive until the
    first time $t \in (0,\tau]$ that either
    \begin{equation} \label{eq:waitsell} \varphi =
      \phi_{\mathrm{sell}}(\tau - t,\zeta e^{-\kappa t})
    \end{equation}
    or
    \begin{equation} \label{eq:waitbuy} \varphi =
      \phi_{\mathrm{buy}}(\tau - t,\zeta e^{-\kappa t})
    \end{equation}
    holds true. That is, the triplet
    $(\tau-t,\zeta e^{-\kappa t}, \varphi)$ belongs to
    $\partial \mathcal{R}_{\mathrm{sell}}$ or
    $\partial \mathcal{R}_{\mathrm{buy}}$ due
    to~\eqref{main:sellbound} and~\eqref{main:buybound},
    respectively. On the remaining time interval $[\tau - t,\tau]$ the
    optimal strategy then coincides with the corresponding optimal
    strategy $\hat{X}^{\tau-t,\zeta e^{-\kappa t},\varphi}$. Note that
    $\hat{X}_0^{\tau-t,\zeta e^{-\kappa t},\varphi,\downarrow}
    =\hat{X}_0^{\tau-t,\zeta e^{-\kappa t},\varphi,\uparrow}=0$ will
    hold true due to the definition of the boundaries
    in~\eqref{def:sellbound} and~\eqref{def:buybound},
    respectively. For the case where neither~\eqref{eq:waitsell}
    nor~\eqref{eq:waitbuy} allows for a solution $t \in (0,\tau]$, it
    will be optimal to remain inactive all along $[0,\tau]$. Recall
    from our convention in Remark~\ref{rem:tauzero}, 2.) that any non-zero
    final position in the risky asset will be unwound with a single
    block trade.
  \end{enumerate}
\end{Remark}

As a consequence of Remark~\ref{rem:main1}, it suffices to
characterize all optimal strategies
$\hat{X}^{\tau,\zeta,\varphi}=(\hat{X}^{\tau,\zeta,\varphi,\uparrow},$
$\hat{X}^{\tau,\zeta,\varphi,\downarrow})$ with initial problem data
$(\tau,\zeta,\varphi)$ which belong to the boundaries
$\partial \cR_{\mathrm{sell}}$ or $\partial \cR_{\mathrm{buy}}$. The
next two corollaries summarize how these strategies can be computed
explicitly.

\begin{Corollary}[Selling boundary] \label{cor:main1} Let
  $(\tau,\zeta,\varphi) \in \partial
  \mathcal{R}_{\mathrm{sell}}$. Then we have
  $\{ d\hat{X}^{\tau,\zeta,\varphi,\downarrow}>0\}=[0,\tau]$. The
  optimal share holdings $\varphi^{\hat{X}^{\tau,\zeta,\varphi}}$ and
  spread dynamics $\zeta^{\hat{X}^{\tau,\zeta,\varphi}}$ satisfy
  \begin{equation} \label{eq:slide1}
    \varphi_t^{\hat{X}^{\tau,\zeta,\varphi}} =
    \phi_{\mathrm{sell}}(\tau-t,\zeta_t^{\hat{X}^{\tau,\zeta,\varphi}})
    \quad (0 \leq t \leq \tau).
  \end{equation}
  In particular, $\varphi^{\hat{X}^{\tau,\zeta,\varphi}}$ solves
  the second order ODE  \begin{equation} \label{eq:ODE1}
    \ddot\varphi_t^{\hat{X}^{\tau,\zeta,\varphi}} = \beta^2 \left(
      \varphi_t^{\hat{X}^{\tau,\zeta,\varphi}} -
      \frac{\mu}{\alpha\sigma^2} \right)
  \end{equation}
  on $(0,\tau)$ with initial conditions 
  \begin{equation} \label{eq:ODEinitcond1}
    \varphi_0^{\hat{X}^{\tau,\zeta,\varphi}} = \varphi, \quad
    \dot\varphi_0^{\hat{X}^{\tau,\zeta,\varphi}} = \beta \left(
      c_-(\tau,\zeta,\varphi) - c_+(\tau,\zeta,\varphi) \right),
  \end{equation}
  where $\beta \set \kappa\lambda/\sqrt{\lambda^2+\kappa\eta}$ and
  $c_{\pm}(\tau,\zeta,\varphi)$ are given as in~\eqref{def:cpm}.
\end{Corollary}

\begin{Remark} \label{rem:OW} Note that in case $\zeta = \mu = 0$, the
  optimal strategy described in Remark~\ref{rem:main1}~1.) together with
  Corollary~\ref{cor:main1} and the convention from
  Remark~\ref{rem:tauzero} 2.) coincides for any $\varphi > 0$ with
  the optimal liquidation strategy computed in~\citet{ObiWang:13},
  Proposition 4.
\end{Remark}

For the buying boundary, the description of optimal strategies becomes a bit
more involved as one has to distinguish three cases depending on the
size of the initial spread:

\begin{Corollary}[Buying boundary] \label{cor:main2} Let
  $(\tau,\zeta,\varphi) \in \partial \mathcal{R}_{\mathrm{buy}}$ and
  let $\bar\zeta$, $\bar\varphi$, $\hat{\zeta}^{\mathrm{buy}}$ be
  given as in~\eqref{def:barzeta}, \eqref{def:barphi},
  \eqref{def:zetahatbuy}, respectively, as well as
  $\tau_{\mathrm{buy}}$, $\tau_{\mathrm{wait}}$ as defined in
  Lemmas~\ref{lem:taubuy} and \ref{lem:tauwait}.
  \begin{enumerate}
  \item If $\zeta > \hat{\zeta}^{\mathrm{buy}}(\tau,2\mu/\kappa,0,0)$,
    then we have
    $\{ d\hat{X}^{\tau,\zeta,\varphi,\uparrow}>0 \} = [0,\tau]$. The
    optimal share holdings $\varphi^{\hat{X}^{\tau,\zeta,\varphi}}$
    and spread dynamics $\zeta^{\hat{X}^{\tau,\zeta,\varphi}}$ satisfy
    \begin{equation} \label{eq:slide2}
      \begin{aligned}
        \varphi_t^{\hat{X}^{\tau,\zeta,\varphi}} =
        \phi_{\mathrm{buy}}(\tau-t,\zeta_t^{\hat{X}^{\tau,\zeta,\varphi}})
        \quad (0 \leq t \leq \tau).
      \end{aligned}
    \end{equation}
    In other words, $\varphi^{\hat{X}^{\tau,\zeta,\varphi}}$ solves the
    second order ODE
    \begin{equation} \label{eq:ODE2}
      \ddot\varphi_t^{\hat{X}^{\tau,\zeta,\varphi}} = \beta^2 \left(
        \varphi_t^{\hat{X}^{\tau,\zeta,\varphi}} -
        \frac{\mu}{\alpha\sigma^2} \right)
    \end{equation}
    on $(0, \tau)$ with initial conditions
    \begin{equation} \label{eq:ODEinitcond2}
      \varphi_0^{\hat{X}^{\tau,\zeta,\varphi}} = \varphi, \quad
      \dot\varphi_0^{\hat{X}^{\tau,\zeta,\varphi}} = \beta \left(
        c_-(\tau,-\zeta,\varphi) - c_+(\tau,-\zeta,\varphi) \right).
    \end{equation}
  \item If
    $\bar\zeta(\tau) < \zeta \leq
    \hat{\zeta}^{\mathrm{buy}}(\tau,2\mu/\kappa,0,0)$, then we have
    $\{d\hat{X}^{\tau,\zeta,\varphi,\uparrow}>0 \} = [0,
    \tau_{\mathrm{buy}}(\tau,\zeta)]$ with
    $\tau_{\mathrm{buy}}(\tau,\zeta) \in (0, \tau]$. The optimal share
    hodings $\varphi^{\hat{X}^{\tau,\zeta,\varphi}}$ and spread
    dynamics $\zeta^{\hat{X}^{\tau,\zeta,\varphi}}$ satisfy
    \begin{equation} \label{eq:slide3}
      \varphi_t^{\hat{X}^{\tau,\zeta,\varphi}} =
      \phi_{\mathrm{buy}}(\tau-t,\zeta_t^{\hat{X}^{\tau,\zeta,\varphi}})
      \quad \left( 0 \leq t \leq \tau_{\mathrm{buy}}(\tau,\zeta)
      \right).
    \end{equation}
    In this case, the ODE dynamics in~\eqref{eq:ODE2} are satisfied by
    $\varphi^{\hat{X}^{\tau,\zeta,\varphi}}$ on
    $(0, \tau_{\mathrm{buy}}(\tau,\zeta))$ with terminal conditions
    \begin{equation} \label{eq:ODEfinalcond}
      \begin{aligned}
        \varphi_{\tau_{\mathrm{buy}}(\tau,\zeta)}^{\hat{X}^{\tau,\zeta,\varphi}}
        = & \,\, \bar\varphi\left(\tau-\tau_{\mathrm{buy}}(\tau,\zeta)\right), \\
        \dot\varphi_{\tau_{\mathrm{buy}}(\tau,\zeta)}^{\hat{X}^{\tau,\zeta,\varphi}}
        = & \,\, \bar\zeta(\tau-\tau_{\mathrm{buy}}(\tau,\zeta))
        \beta^2/\lambda^2 \\
        & + \left( \bar\varphi(\tau-\tau_{\mathrm{buy}}(\tau,\zeta)) -
          \mu/\lambda^2 \right) \beta^2/\kappa.
      \end{aligned}
    \end{equation}
  \item If $0 \leq \zeta \leq \bar\zeta(\tau)$, then
    $\hat{X}^{\tau,\zeta,\varphi,\uparrow} \equiv 0$ on $[0,\tau]$.
  \end{enumerate}
  Moreover, in both cases 2.) and 3.), if
  \begin{equation} \label{def:tausell}
    \tau_{\mathrm{sell}}(\tau,\zeta)\set
    \tau-\tau_{\mathrm{buy}}(\tau,\zeta)-\tau_{\mathrm{wait}}(\tau,\zeta)
    > 0,
  \end{equation}
  then it holds that
  \begin{equation} \label{eq:waitandsell} \left(
      \tau_{\mathrm{sell}}(\tau, \zeta),
      \zeta^{\hat{X}^{\tau,\zeta,\varphi}}_{\tau_{\mathrm{buy}}(\tau,\zeta)+\tau_{\mathrm{wait}}(\tau,\zeta)},
      \varphi^{\hat{X}^{\tau,\zeta,\varphi}}_{\tau_{\mathrm{buy}}(\tau,\zeta)+\tau_{\mathrm{wait}}(\tau,\zeta)}
    \right) \in \partial \mathcal{R}_{\mathrm{sell}}.
    \end{equation}
\end{Corollary}

\begin{Remark} \label{rem:main2} Notice that except for a possible
  initial and final singular block trade (recall
  Remarks~\ref{rem:tauzero}, 2.) and~\ref{rem:main1}), share holdings
  $\varphi^{\hat{X}}$ of optimal strategies turn out to be absolutely
  continuous. This is in line with the optimal execution strategies
  computed in~\citet{ObiWang:13}. During these periods of steadily
  buying or selling, the dynamics of the optimal share holdings are
  prescribed by the same second order ODE in~\eqref{eq:ODE1}
  and~\eqref{eq:ODE2}. In fact, satisfying the ODE with the
  corresponding boundary conditions forces, respectively, the
  buy-subgradient or sell-subgradient to vanish which is in line with
  the first order optimality conditions in
  Proposition~\ref{prop:foc}. Also note that while the optimal
  strategy is continuously buying- or selling, the optimally
  controlled triplet evolves along the boundary of the buying- or
  selling-region in the state space~$\cS$; cf.~\eqref{eq:slide1},
  \eqref{eq:slide2}, and~\eqref{eq:slide3} together with
  Theorem~\ref{thm:main}.
\end{Remark}


\subsection{Illustration}
\label{subsec:illustrations}

Let us illustrate with a numerical example the separation of the
three-dimensional state space~$\cS$ into a buying-, waiting- and
selling-region as characterized in Theorem~\ref{thm:main} along with
trajectories of optimal strategies as described in
Corollaries~\ref{cor:main1} and~\ref{cor:main2} together with
Remark~\ref{rem:main1}. All explicit representations of the free
boundaries $\partial \mathcal{R}_{\text{buy}}$ and
$\partial \mathcal{R}_{\text{sell}}$ as well as of the illustrated
optimal strategies
$\hat{X}^{\tau,\zeta,\varphi} =
(\hat{X}^{\tau,\zeta,\varphi,\uparrow},
\hat{X}^{\tau,\zeta,\varphi,\downarrow})$ can be found in
Section~\ref{subsec:proofs3}. As for the model parameters, we simply
choose
\begin{equation*} 
  \kappa = 1, \quad \eta = 2, \quad \mu = 10, \quad
  \sigma = 1, \quad \alpha = 1.
\end{equation*}

\begin{figure}[h]
\centering
\includegraphics[scale=.6]{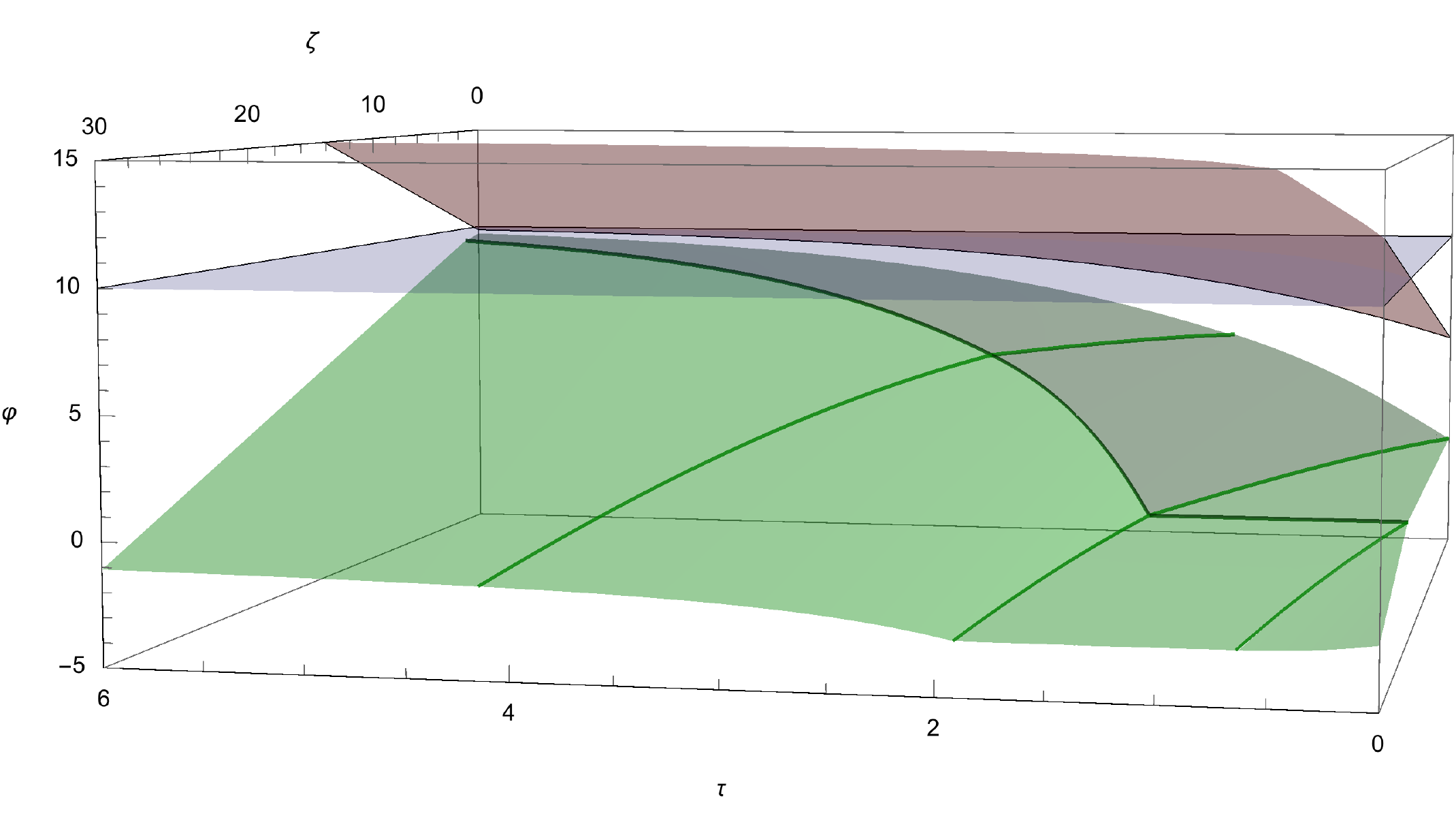}
\caption{The three-dimensional state space $\cS$ with time to maturity
  $\tau$, spread $\zeta$ and number of shares $\varphi$. The blue
  plane represents the Merton plane at level
  $\varphi^0=\mu/(\alpha\sigma^2)=10$. The boundary of the
  buying-region $\partial \mathcal{R}_{\text{buy}}$ is colored in
  green and the boundary of the selling-region
  $\partial \mathcal{R}_{\text{sell}}$ is colored in red.}
\label{fig:stateSpace}
\end{figure}

Figure~\ref{fig:stateSpace} shows the three-dimensional state space
$\cS$ with time to maturity~$\tau$, spread~$\zeta$ and number of
shares~$\varphi$. The blue plane represents the constant optimal
frictionless Merton position at level
$\varphi^0 = \mu/(\alpha\sigma^2) = 10$, henceforth referred to as
\emph{Merton plane}. The upper red surface is the free boundary of the
selling region $\partial \mathcal{R}_{\text{sell}}$ as characterized
in Theorem~\ref{thm:main}, i.e.,
\begin{equation*}
  \partial \mathcal{R}_{\text{sell}} 
  = \left\{ (\tau,\zeta,\varphi) \in \cS : 
    \varphi = \phi_{\mathrm{sell}}(\tau,\zeta)\right\}
\end{equation*}
with $\phi_{\mathrm{sell}}$ defined in~\eqref{def:phisell}.  The lower
green surface depicts the free boundary of the buying region
$\partial \mathcal{R}_{\text{buy}}$, that is,
\begin{equation*}
  \partial \mathcal{R}_{\text{buy}} 
  = \left\{ (\tau,\zeta,\varphi) \in \cS : 
    \varphi = \phi_{\mathrm{buy}}(\tau,\zeta)\right\}
\end{equation*}
with $\phi_{\mathrm{buy}}$ defined in~\eqref{def:phibuy:eq1} --
\eqref{def:phibuy:eq3c}. Observe that $\partial\mathcal{R}_{\text{buy}}$
actually decomposes into seven parts
(cf. Section~\ref{subsubsec:proofs:aux} for more details). As expected
by the formulation of the optimization problem in
Proposition~\ref{prop:optproblem} as an optimal trading problem
towards the constant Merton portfolio, one can observe in
Figure~\ref{fig:stateSpace} that the green boundary of the buying
region $\partial \mathcal{R}_{\text{buy}}$ is always below the Merton
plane. Moreover, at least for large maturities~$\tau$ and large
initial spread values~$\zeta$, the red boundary of the selling region
$\partial \mathcal{R}_{\text{sell}}$ is above the Merton
position. However, notice that it falls below the latter for small
maturities and small spread values $\zeta$. That is, even though the
position in the risky asset is below the target portfolio $\varphi^0$,
the short time horizon forces the optimizer to start liquidating the
share holdings right away. Recall from Remark~\ref{rem:opttracking}
2.) that this results from the fact that liquidation is costly. Hence,
the optimal control also has to take into account unwinding the
accrued position when terminal time comes close. The same
interpretation also applies for the ``plateau'' of the buying
boundary~$\partial \mathcal{R}_{\text{buy}}$ at level $\varphi=0$ for
small maturities and small spread values. In other words, starting
with a short position in the risky asset and facing a short time
horizon, it is optimal to simply clear out the short position even before
the time horizon is reached. Mathematically, the presence of this plateau is due to
the dependence of the subgradients in~\eqref{eq:buysubgradient}
and~\eqref{eq:sellsubgradient} on $\varrho \in [-1,1]$ in case where
the optimal terminal
position~$\varphi^{\hat{X}^{\tau,\zeta,\varphi}}_{\tau}$ is equal to
zero (again cf. Section~\ref{subsubsec:proofs:aux} for the details).

Figure~\ref{fig:policies} depicts the evolutions of some optimal share
holdings $\varphi^{\hat{X}^{\tau,\zeta,\varphi}}$ for different
problem data $(\tau,\zeta,\varphi) \in \cS$ as functions in time to
maturity $\tau - t$ with $0 \leq t \leq \tau$. The corresponding
spread dynamics $\zeta^{\hat{X}^{\tau,\zeta,\varphi}}$ are presented
in Figure~\ref{fig:spread}. The trajectories of the associated
optimally controlled state processes
$(\tau-t,\zeta^{\hat{X}^{\tau,\zeta,\varphi}}_t,\varphi^{\hat{X}^{\tau,\zeta,\varphi}}_t)_{0
  \leq t \leq \tau}$ embedded in the three-dimensional state space
$\cS$ are illustrated in Figure~\ref{fig:stateSpacePolicies}.

\begin{figure}[h]
  \centering
  \includegraphics[scale=.7]{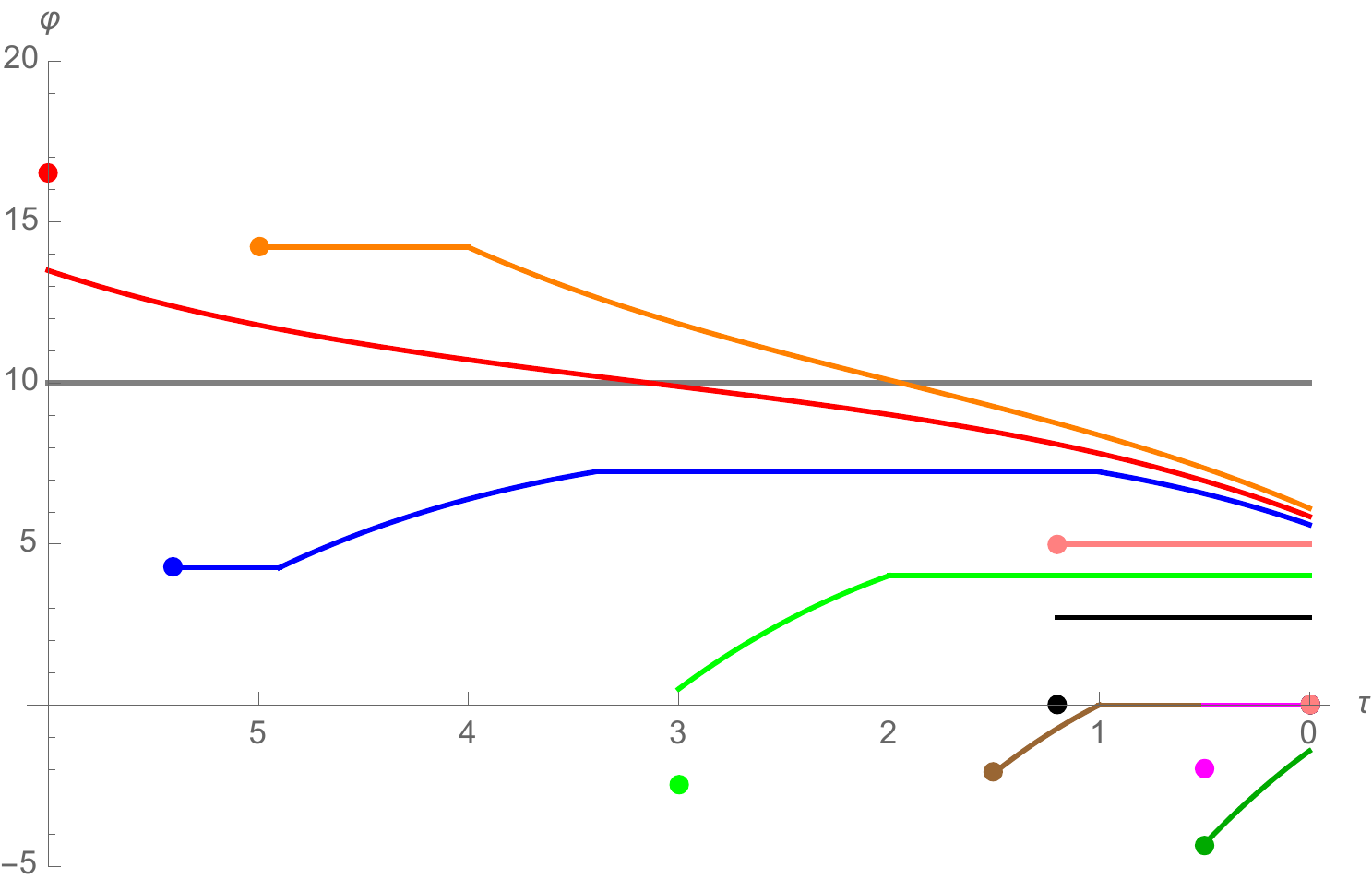}
  \caption{Evolution of optimal share holdings for different initial
    problem data $(\tau,\zeta,\varphi) \in \cS$ as functions in time
    to maturity $\tau - t$ with $0 \leq t \leq \tau$. The dots
    represent the initial position in the risky asset. By our
    convention from Remark~\ref{rem:tauzero}, 2.), all strategies unwind
    non-zero positions in the end with an impulse trade. The grey line
    depicts the Merton position $\varphi^0=10$.}
  \label{fig:policies}
\end{figure}

\begin{figure}[h]
  \centering
  \includegraphics[scale=.7]{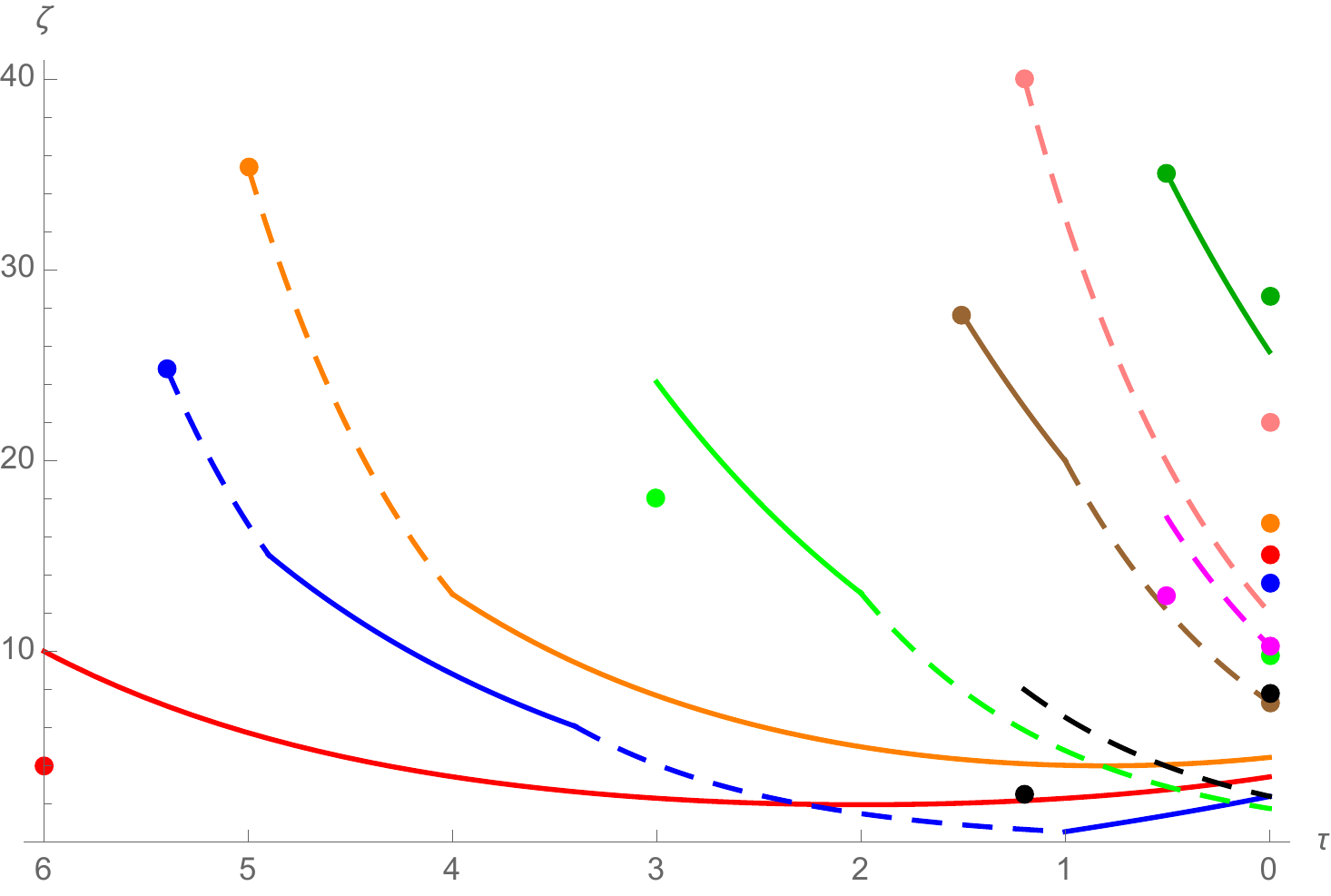}
  \caption{Evolution of the corresponding spread dynamics of the
    optimal share holdings from Figure~\ref{fig:policies}, again as
    functions in time to maturity $\tau - t$ with
    $0 \leq t \leq \tau$. The dots represent the initial and final
    spread values. Periods where the optimal strategy is inactive are
    indicated by dashed lines.}
  \label{fig:spread}
\end{figure}

\begin{figure}
  \centering
  \includegraphics[scale=.6]{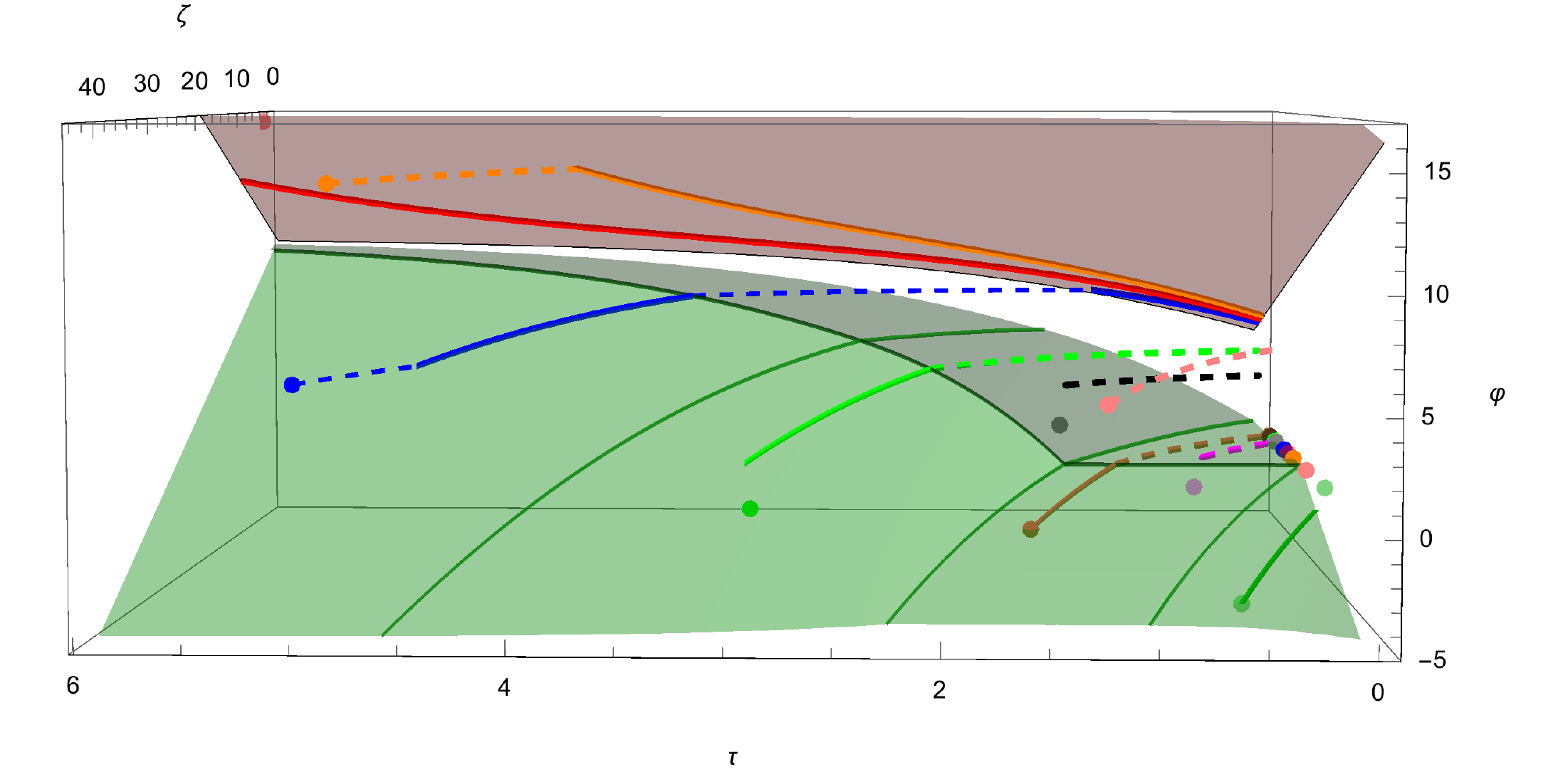}
  \caption{Evolution of optimally controlled state processes embedded
    in the three-dimensional state space $\cS$ corresponding to the
    optimal share holdings from Figure~\ref{fig:policies} with spread
    dynamics depicted in Figure~\ref{fig:spread}. Dashed lines
    indicate waiting parts of the strategies and the big dots
    represent the corresponding initial and final triplets.}
  \label{fig:stateSpacePolicies}
\end{figure}

The red policy is similar to the optimal liquidation strategies
computed in~\citet{ObiWang:13} for a risk-averse investor, though with
a non-zero but small initial spread (recall Remark~\ref{rem:OW}).
Observe that the trajectory starts in the selling
region~$\mathcal{R}_{\text{sell}}$ with an initial position in the
risky asset above the Merton portfolio. Thus, as described in
Remark~\ref{rem:main1} 1.), the policy jumps with an initial block
sell order to the boundary of the selling
region~$\partial \mathcal{R}_{\text{sell}}$ and then continues
steadily trading towards the Merton level $\varphi^0$ by selling the risky
asset as described in Corollary~\ref{cor:main1}. In particular, note
that the strategy steadily sells until maturity even after reaching
the targeted Merton level $\varphi^0=10$. As characterized
in~\eqref{eq:slide1} (recall also Remark~\ref{rem:main2}), the
optimally controlled trajectory evolves along the
boundary~$\partial \mathcal{R}_{\text{sell}}$. At the end, following
our convention from Remark~\ref{rem:tauzero}~2.), the remaining shares
are liquidated with a single block sell order. Similarly to the red
policy, the orange policy also has an initial position above Merton
but it comes along with a large initial spread. As a consequence, the
corresponding problem data $(\tau,\zeta,\varphi)$ belongs to the
waiting region~$\mathcal{R}_{\text{wait}}$. In this case it is optimal
to exploit the resilience effect first and to be inactive until the
value of the spread is sufficiently small so that it becomes optimal
to trade towards~$\varphi^0$. This happens when the trajectory hits
the boundary of the selling
region~$\partial \mathcal{R}_{\text{sell}}$ as described in
Remark~\ref{rem:main1} 2.). The blue policy is an optimal strategy
which decomposes into a waiting-, buying-, waiting- and selling
part. The initial share holdings are below the desired Merton position
but similar to the orange policy the initial spread value is too large
to intervene immediately. Again the optimal strategy is inactive until
the spread is sufficiently small so that it becomes optimal to trade
towards $\varphi^0$ and to buy shares according to the description in
Corollary~\ref{cor:main2}~2.), once the trajectory hits the boundary
of the buying region $\partial \mathcal{R}_{\text{buy}}$. Note that
the optimizer is exploiting the resilience effect also while it is
purchasing the risky asset since the value of the spread continues to
decrease; see Figure~\ref{fig:spread}. Thereafter, when the position
in the risky asset is close enough to Merton $\varphi^0$ with respect
to the remaining time, the optimizer becomes inactive again. During
this waiting period the spread continues to decay until the trajectory
hits the boundary of the selling
region~$\partial \mathcal{R}_{\text{sell}}$ as characterized
in~\eqref{eq:waitandsell}. The optimal control then starts to
continuously unwind its accrued position until terminal time. The
black policy is of buy-and-hold type with initial and final block
trades, remarkably similar to the frictionless optimizer. This is due
to the fact that the time horizon is very small together with a small
initial spread which makes it optimal to execute a single initial
block buy order from~$\mathcal{R}_{\text{sell}}$
to~$\partial \mathcal{R}_{\text{sell}}$ towards the Merton position
$\varphi^0$ but without reaching it. Thereafter, the optimizer follows
the characterization in Corollary~\ref{cor:main2} 3.)
where~\eqref{eq:waitandsell} does not occur. The pink policy does not
trade at all and unwinds at the end. The brown policy starts with a
short position in the risky asset. Again, since time horizon is
relatively short, it merely clears out its short position in the risky asset
as described in Corollary~\ref{cor:main2} 2.) even before the end and
then remains at level $\varphi =0$ until the time horizon is
reached. Similarly for the magenta policy but instead with a single
initial block buy order. The dark green policy continuously liquidates
an initial short position until the end and correspond to the case
described in Corollary~\ref{cor:main2} 1.).

To sum up, the numerical example illustrates how the optimal
strategies which maximize expected utility from terminal liquidation
wealth in our illiquid Bachelier model exhibit for different time
horizons $\tau$, initial spread values $\zeta$ and initial endowments
$\varphi$ a rich phenomenology of possible trajectories.


\section{Proofs}
\label{sec:proofs}

\subsection{Proofs for Sections \ref{sec:model} and \ref{sec:problem}}
\label{subsec:proofs1}

We start with the computation of the dynamics of the liquidation
wealth process $(V_t(X))_{t \geq 0}$ defined in
\eqref{eq:defliquidwealth} and the associated liquidity costs
$(L_t(X))_{t \geq 0}$ stated in Lemma
\ref{lem:propliquidwealth}.\medskip

\noindent{{\emph{Proof of Lemma
      \ref{lem:propliquidwealth}.}}}
To alleviate the notation, let us introduce the mid-quote price
process $M^X_t \set (A^X_t + B^X_t)/2$ for all $t \geq 0$ with initial
value $M^X_{0-} \set (A_{0} + B_{0})/2 = P_{0-}$. Applying integration
by parts in \eqref{eq:defliquidwealth} as in, e.g.,
\citet{JacodShir:02}, Definition I.4.45, yields
\begin{equation} \label{eq:p:wealthdyn:1}
  \begin{aligned}
    dV_t(X) = & \; -\frac{1}{2} ( \zeta^X_{t-} + \eta \Delta
    X^{\uparrow}_t) dX^{\uparrow}_t -\frac{1}{2} ( \zeta^X_{t-} + \eta
    \Delta X^{\downarrow}_t ) dX^{\downarrow}_t + \varphi^X_{t-}
    dM^X_t
    \\
    & \; -\eta \varphi^X_{t-} d\varphi^X_t - \frac{1}{2} \Big(
    \zeta_{t-}^X d\vert \varphi^X_t\vert + \vert \varphi^X_{t-}\vert
    d\zeta^X_t + d[\vert \varphi^X \vert, \zeta^X]_t \Big),
  \end{aligned}
\end{equation}
where we used the fact that
$[\varphi^X,M^X] = \eta [\varphi^X,\varphi^X]/2$ by virtue of
\cite{JacodShir:02}, Theorem I.4.52. Moreover, note that Proposition
I.4.49 a) in \cite{JacodShir:02} implies
$[\vert \varphi^X \vert, \zeta^X]_t = \int_{[0,t]} \Delta \zeta^X_s
d\vert \varphi^X_s \vert$ for all $t \geq 0$, because
$\vert \varphi^X \vert$ is predictable and $\zeta^X$ is of finite
variation. Inserting this, the spread dynamics
\eqref{eq:spreaddynamics} as well as the dynamics of the mid-quote
$dM^X_t = dP_t + \frac{\eta}{2} dX^{\uparrow}_t - \frac{\eta}{2}
dX^{\downarrow}_t$ in \eqref{eq:p:wealthdyn:1} above yields
\begin{align} \label{eq:p:wealthdyn:3}
  \begin{aligned}
    dV_t(X) = & ~ \varphi^X_{t-} dP_t - \frac{1}{2} \zeta^X_t d\vert
    \varphi^X_t \vert
    + \frac{1}{2} \kappa \vert \varphi_{t-}^X \vert \zeta^X_{t-} dt  \\
    & - \frac{1}{2} \Big( \zeta^X_{t-} + \eta \Delta X^{\uparrow}_t +
    \eta \varphi^X_{t-}
    + \eta \vert \varphi^X_{t-} \vert \Big) dX^{\uparrow}_t \\
    & - \frac{1}{2} \Big( \zeta^X_{t-} + \eta \Delta X^{\downarrow}_t
    - \eta \varphi^X_{t-} + \eta \vert \varphi^X_{t-} \vert \Big)
    dX^{\downarrow}_t \quad (t \geq 0).
  \end{aligned}
\end{align}
This motivates to define the liquidation cost functional $L_t(X)$ as
\begin{align} \label{eq:p:wealthdyn:4}
  \begin{aligned}
    L_t(X) \set & \; L_{0-}(X) + \frac{1}{2} \int_{[0,t]} \zeta^X_s
    d\vert \varphi^X_s \vert - \frac{1}{2} \kappa \int_{[0,t]} \vert
    \varphi^X_{s-}\vert
    \zeta^X_{s-} ds \\
    & \; + \frac{1}{2} \int_{[0,t]} \left(\zeta^{X}_{s-} + \eta \Delta
      X^{\uparrow}_s + \eta \varphi^X_{s-}
      + \eta \vert \varphi^X_{s-} \vert \right) dX^{\uparrow}_s \\
    & \; + \frac{1}{2} \int_{[0,t]} \left( \zeta^X_{s-} + \eta \Delta
      X^{\downarrow}_s - \eta \varphi^X_{s-} + \eta \vert
      \varphi^X_{s-} \vert \right) dX^{\downarrow}_s \quad (t \geq 0)
  \end{aligned}
\end{align}
with
$L _{0-} (X) \set \zeta_0 \vert \varphi^X_{0-} \vert/2 + \eta
(\varphi^X_{0-})^2/2$. Using once more the spread dynamics in
\eqref{eq:spreaddynamics} we can write
$-\frac{\kappa}{2} \vert \varphi^X_{t-} \vert \zeta^X_{t-} dt =
\frac{1}{2} \vert \varphi^X_{t-} \vert d\zeta^X_t - \frac{\eta}{2}
\vert \varphi^X_t \vert (dX^{\uparrow}_t + dX^{\downarrow}_t)$.
Inserting this expression in \eqref{eq:p:wealthdyn:4} gives us
\begin{align} \label{eq:p:wealthdyn:5}
  \begin{aligned}
    L_t(X) = & \; L_{0-}(X) + \frac{1}{2} \int_{[0,t]} (\zeta^X_{s-} +
    \eta \Delta X^{\uparrow}_s )
    dX^{\uparrow}_s \\
    & + \frac{1}{2} \int_{[0,t]} (\zeta^X_{s-} + \eta \Delta
    X^{\downarrow}_s )
    dX^{\downarrow}_s +\frac{1}{2} \int_{[0,t]} \zeta^X_s d\vert \varphi^X_s \vert \\
    & +\frac{1}{2} \int_{[0,t]} \vert \varphi^X_{s-} \vert d\zeta^X_s
    +\frac{1}{2} \eta \int_{[0,t]} \varphi^X_{s-} d \varphi^X_s \quad
    (t \geq 0).
  \end{aligned}
\end{align}
Again, integration by parts as in \cite{JacodShir:02}, Definition
I.4.45, allows us to write
\begin{align}
  \frac{1}{2} \int_{[0,t]} \vert \varphi^X_{s-} \vert d\zeta^X_s
  = & \;
      \frac{1}{2}
      \vert \varphi^X_t \vert \zeta^X_t - \frac{1}{2} \vert \varphi^X_{0-} \vert
      \zeta^X_{0-} - \frac{1}{2} \int_{[0,t]} \zeta^X_s d\vert
      \varphi^X_s \vert, \label{eq:p:wealthdyn:6} \\ 
  \frac{1}{2} \eta \int_{[0,t]} \varphi^X_{s-} d \varphi^X_s 
  = & \;
      \frac{1}{4}
      \eta
      \left( (\varphi^X_t)^2 - (\varphi^X_{0-})^2 -
      [\varphi^X,\varphi^X]_t \right). \label{eq:p:wealthdyn:7}
\end{align}
Plugging back \eqref{eq:p:wealthdyn:6} and \eqref{eq:p:wealthdyn:7}
into \eqref{eq:p:wealthdyn:5}, using the definition of $L_{0-}(X)$ as
well as the fact that
$[X^{\uparrow,\downarrow},X^{\uparrow,\downarrow}]_t = \int_{[0,t]}
\Delta X^{\uparrow,\downarrow}_s \, dX^{\uparrow,\downarrow}_s$ for
all $t\geq 0$ (cf. Proposition I.4.49 a) in \cite{JacodShir:02})
finally yields
\begin{align} \label{eq:p:wealthdyn:8}
  \begin{aligned}
    L_t(X) = & ~\frac{1}{2} |\varphi^X_t| \zeta^X_t + \frac{\eta}{4}
    \big((\varphi^X_t)^2 + (\varphi^X_{0-})^2 \big) + \frac{1}{2}
    \int_{[0,t]} \zeta^X_{s-} (dX_s^{\uparrow} + dX_s^{\downarrow})\\
    & + \frac{\eta}{4} ([X^{\uparrow},X^{\uparrow}]_t +
    [X^{\downarrow},X^{\downarrow}]_t +
    2[X^{\uparrow},X^{\downarrow}]_t) \quad (t \geq 0).
  \end{aligned}
\end{align}
Next, by using the explicit representation of the spread $\zeta^X$ in
\eqref{eq:spreadsolution} and introducing the process
$Y_t \set \int_{[0,t]} e^{\kappa s} (dX^{\uparrow}_s +
dX^{\downarrow}_s)$ for all $t \geq 0$ we obtain
\begin{equation} \label{eq:p:wealthdyn:9}
  \begin{aligned}
    \int_{[0,t]} \zeta^X_{s-} (dX_s^{\uparrow} + dX_s^{\downarrow}) &
    = \int_{[0,t]} e^{-\kappa s} \zeta_{0} (dX_s^{\uparrow} +
    dX_s^{\downarrow}) + \frac{\eta}{2} e^{-2 \kappa t} Y_t^2 \\
    & \hspace{10pt} + \kappa \eta \int_0^t e^{-2\kappa s} Y^2_{s-} ds
    - \frac{\eta}{2} \int_{[0,t]} e^{-2\kappa s} d[Y,Y]_s.
  \end{aligned}
\end{equation}
Once more due to \cite{JacodShir:02}, Proposition I.4.49, observe that
we have
\begin{equation}
  d[Y,Y]_s = \;  e^{2 \kappa s} \left( d[X^{\uparrow},X^{\uparrow}]_s 
    + d[X^{\downarrow},X^{\downarrow}]_s +
    2d[X^{\uparrow},X^{\downarrow}]_s \right). \label{eq:p:wealthdyn:10}
\end{equation}
In addition, it holds that
$e^{-\kappa t} Y_t = ( \zeta_t^X - e^{-\kappa t} \zeta_0)/\eta$ for
all $t \geq 0$. Thus, using this representation as well
as~\eqref{eq:p:wealthdyn:10} in \eqref{eq:p:wealthdyn:9}, and plugging
the resulting term back into \eqref{eq:p:wealthdyn:8} yields the
desired form of the liquidity cost functional in
\eqref{eq:liquiditycosts}. Finally, one can easily observe that the
functional $L_t(X)$ in~\eqref{eq:liquiditycosts} is convex in $X$ for
each $t \geq 0$. Moreover, using the lower estimate
$\zeta^X_t - e^{-\kappa t} \zeta_0 \geq \eta e^{-\kappa t}
(X^{\uparrow}_t + X^{\downarrow}_t)$ for all $t \geq 0$, we obtain the
lower bound of $L_t(X)$ as claimed
in~\eqref{eq:liquiditycostsbound}. \qed \medskip

In order to apply Lemma~\ref{lem:compactness} in our setting in the
proof of Theorem~\ref{thm:existence}, we need the following lemma.

\begin{Lemma} \label{lem:levelsets} For the level-set
  $\cL_0 \set \{ X \in \cX : \EE u(V_T(X)) \geq \EE u(V_T(0))\}$,
  $\conv(\{ X^{\uparrow}_T + X^{\downarrow}_T : X \in \cL_0 \})$ is
  $L^0(\Omega,\cF,\PP)$-bounded.
\end{Lemma}

\begin{proof}
  First, observe that due to convexity of the liquidity cost
  functional $L_T(X)$ in $X \in \cX$ by virtue of Lemma
  \ref{lem:propliquidwealth} as well as concavity and monotonicity of
  the utility function~$u$, the level-set $\cL_0$ is a convex set.  As
  a consequence, it holds that
  $\conv(\{ X^{\uparrow}_T + X^{\downarrow}_T : X \in \cL_0 \}) = \{
  X^{\uparrow}_T + X^{\downarrow}_T : X \in \cL_0 \}$.  Next, note
  that for any $X \in \cX$ the liquidation wealth $V_T(X)$ as given in
  \eqref{eq:liquidwealth} can be bounded from above by
  \begin{align}
    V_T(X)
    & = V_{0-}(X) + L_{0-}(X) + \int_0^T \varphi^X_t dP_t - L_T(X) \nonumber \\
    & \leq \xi^X_{0-} + 2 ( \varphi^X_{0-} + X^{\uparrow}_T + X^{\downarrow}_T) P_T^* - c
      (X^{\uparrow}_T + X^{\downarrow}_T)^2 \nonumber \\
    & = \xi^X_{0-} + \frac{1}{c} (P_T^*)^2 -
      \Big( \sqrt{c} (X^{\uparrow}_T + X^{\downarrow}_T) 
      - \frac{1}{\sqrt{c}} P_T^* \Big)^2 + 2 \varphi^X_{0-} P_T^* 
      \label{p:levelsets:bound1}
  \end{align}
  with $P^*_T\set\max_{0 \leq s\leq T} \vert P_s \vert$, where we used
  integration by parts, the fact that the semimartingale
  $(P_t)_{t \geq 0}$ is continuous and the lower bound
  $L_T(X)\geq c ( X^{\uparrow}_T + X^{\downarrow}_T )^2$ from Lemma
  \ref{lem:propliquidwealth} for some constant $c>0$. Henceforth, to
  alleviate the presentation, let us assume without loss of generality
  that $\xi^X_{0-} = \varphi^X_{0-} = 0$ as well as $u(0)=0$. Due to
  the upper bound in~\eqref{p:levelsets:bound1}, we obtain for all
  $X \in \cL_0$ the estimate
  \begin{equation*}
    \EE[u(V_T(0))]
    \leq \EE\left[ u \left( \frac{1}{c} (P_T^*)^2 -
        \left( \sqrt{c} (X^{\uparrow}_T + X^{\downarrow}_T) 
          - \frac{1}{\sqrt{c}} P_T^* \right)^2 \right) \right]. 
  \end{equation*}
  Hence, together with the fact that $u$ is bounded from above, it
  must hold for the negative part that
  \begin{equation*}
    \sup_{X \in \cL_0} \EE\left[ u \left( \frac{1}{c} (P_T^*)^2 -
        \left( \sqrt{c} (X^{\uparrow}_T + X^{\downarrow}_T) 
          - \frac{1}{\sqrt{c}} P_T^* \right)^2 \right)^{\!\! -} \; \right] < \infty.
  \end{equation*}
  Moreover, since $u \in C^1(\mathbb{R})$ is strictly concave and
  increasing which yields $u(z) \leq u(0) + u'(0) z = u'(0) z$ and
  thus $u(z)^{-} \geq u'(0) (-z)^+$ for all $z \in \RR$, we obtain
  \begin{equation} \label{p:levelsets:bound4} 
    \sup_{X \in \cL_0}
    \EE\left[\left( \left( \sqrt{c} (X^{\uparrow}_T + X^{\downarrow}_T) -
          \frac{1}{\sqrt{c}} P_T^* \right)^2 - \frac{1}{c} (P_T^*)^2
      \right)^{\!\! +} \; \right] < \infty.
  \end{equation}
  Finally, observe that the
  $L^1(\Omega,\mathcal{F},\mathbb{P})$-boundedness
  in~\eqref{p:levelsets:bound4} implies that the set
  $\{ X^{\uparrow}_T + X^{\downarrow}_T : X \in \cL_0 \}$ is bounded in
  $L^0(\Omega,\mathcal{F},\mathbb{P})$.
\end{proof}

The last ingredient for the proof of Theorem~\ref{thm:existence} is
the continuity of the liquidation wealth $V_T(X)$ in $X$.
\medskip

\noindent{{\emph{Proof of
      Lemma~\ref{lem:semicontinuity}.}}}
We fix $\omega \in \Omega$. By weak convergence of $X^n$ to $X$ on
$[0,T]$, we obtain that
$\zeta^{X^n}_t(\omega) \rightarrow \zeta^X_t(\omega)$ for $t=T$ and
all $t \in [0,T)$ such that
$\Delta X^{\uparrow}_t(\omega)=\Delta X^{\downarrow}_t(\omega)=0$;
cf. the representation of the spread in~\eqref{eq:spreadsolution}. In
particular, it holds that
$\zeta^{X^n}_\cdot(\omega) \rightarrow \zeta^X_\cdot(\omega)$
$dt$-a.e. on $[0,T]$ because the number of jumps of
$X^{\uparrow}(\omega)$, $X^{\downarrow}(\omega)$ is countable. At the
time $\zeta^{X^n}_s(\omega)$ is uniformly bounded in $n$ and $s$ since
so is $X^n_s(\omega)$. Thus,
by dominated convergence, we get for any $\omega \in \Omega$ that
$\lim_{n \rightarrow \infty} \int_0^t \left(\zeta^{X^n}_{s}(\omega)
  - e^{-\kappa s} \zeta_0 \right)^2 ds = \int_0^t
\left(\zeta^X_{s}(\omega) - e^{-\kappa s} \zeta_0 \right)^2 ds.$
Moreover, we obviously have that
$\varphi_T^{X^n}(\omega) \rightarrow \varphi_T^X(\omega)$. Hence,
referring to the representation of the liquidity costs $L_T(X^n)$
in~\eqref{eq:liquiditycosts}, we can conclude that
$\lim_{n \rightarrow \infty} L_T(X^n(\omega)) =
L_T(X(\omega))$. Next, concerning the stochastic integral of
$\varphi^{X^n}$ with respect to the continuous semimartingale~$P$ in
the liquidation wealth $V_T(X^n)$ in~\eqref{eq:liquidwealth}, we
obtain, after applying integration by parts, the expression
  \begin{align*}
   & \int_0^T \varphi^X_t dP_t 
    = \varphi^X_T P_T - \varphi^X_{0-} P_{0-} 
      - \int_{[0,T]} P_s (dX^{\uparrow}_s -dX^{\downarrow}_s) \\
    & = \lim_{n \rightarrow \infty} \Big( \varphi^{X^n}_T P_T - \varphi^{X^n}_{0-} P_{0-} 
      - \int_{[0,T]} P_s (dX^{\uparrow,n}_s -dX^{\downarrow,n}_s) \Big) \\
    & = \lim_{n \rightarrow \infty}  \int_0^T \varphi^{X^n}_t dP_t 
      \quad \textrm{for all} \,\, \omega \in \Omega,
  \end{align*}
  where we again used weak convergence of
  $X^n(\omega) \xrightarrow[]{w} X(\omega)$ on $[0,T]$ for all
  $\omega \in \Omega$ and the continuity of $P$. In summary, we
  obtain $\lim_{n \rightarrow \infty} V_T(X^n) = V_T(X)$
  pointwise for all $\omega \in \Omega$ as desired.
\qed


\subsection{Proofs of Lemma \ref{lem:subgradients} and Lemma
  \ref{lem:nobuysell}}
\label{subsec:proofs2}

Next, let us compute the infinite-dimensional subgradients
in~\eqref{eq:buysubgradient} and~\eqref{eq:sellsubgradient} of the
convex cost functional $J_T(\cdot)$ on $\cX^d$ given
in~\eqref{eq:trackingproblem}.  \medskip

\noindent{{\emph{Proof of Lemma
      \ref{lem:subgradients}.}}}
Let us define the deviation functional
\begin{equation}
  D_T(X) \set \frac{\alpha\sigma^2}{2} \int_0^T \left(\varphi^X_t -
    \frac{\mu}{\alpha\sigma^2} \right)^2 dt \label{p:lem:sub:D1}
\end{equation}
on $\cX^d$. Then the convex cost functional $J_T(\cdot)$ in
\eqref{eq:trackingproblem} can be written as
$ J_T(X) = L_T(X) + D_T(X)$.  We will proceed in three steps.

\emph{Step 1:} Let us start with the computation of the subgradients
of the liquidity cost functional $L_T(\cdot)$ given in
\eqref{eq:liquiditycosts}. Observe that for any $X, Y \in \cX^d$ with
$\varphi^{Y}_{0-}=\varphi^{X}_{0-}$ and any $\epsilon \in (0, 1]$ we
obtain
\begin{equation} \label{p:lem:sub:L1}
  \begin{aligned}
    & \frac{L_T(\epsilon Y + (1-\epsilon) X) - L_T(X)}{\epsilon} =
    \frac{\kappa}{\eta} \int_0^T (\zeta^X_t - e^{-\kappa
      t} \zeta_0) (\zeta^Y_t - \zeta^X_t) dt \\
    & + \frac{\eta}{4} \frac{| \epsilon \varphi^Y_T + (1-\epsilon)
      \varphi^X_T |^2 - |\varphi^X_T|^2}{\epsilon} + \frac{1}{2}
    \zeta^X_T \frac{| \epsilon \varphi^Y_T + (1-\epsilon)
      \varphi^X_T | - |\varphi^X_T|}{\epsilon}  \\
    & + \frac{1}{2} \left( \zeta^Y_T - \zeta^X_T \right) \left(
      |\epsilon \varphi^Y_T + (1-\epsilon) \varphi^X_T| +
      \frac{1}{\eta} (\zeta^X_T - e^{-\kappa T} \zeta_0)
    \right)  \\
    & + \frac{1}{2} \int_{[0,T]} e^{-\kappa t} \zeta_0
    (dY^{\uparrow}_t + dY^{\downarrow}_t - dX^{\uparrow}_t
    - dX^{\downarrow}_t)  \\
    & + \epsilon \left(\frac{1}{4 \eta} (\zeta^Y_T - \zeta^X_T)^2 +
      \frac{\kappa}{2 \eta} \int_0^T (\zeta^Y_t - \zeta^X_t)^2 dt
    \right).
  \end{aligned}
\end{equation}
Note that we have the lower bound
$| \epsilon \varphi^Y_T + (1-\epsilon) \varphi^X_T |^2 -
|\varphi^X_T|^2\geq 2\varepsilon \varphi^X_T (\varphi^Y_T -
\varphi^X_T)$ and
$| \epsilon \varphi^Y_T + (1-\epsilon) \varphi^X_T | -
|\varphi^X_T|\geq \epsilon \sign_\varrho(\varphi^X_T) (\varphi^Y_T -
\varphi^X_T)$, where we denote by $x \mapsto \sign_\varrho(x)$ the
subgradient of the function $x \mapsto |x|$ with
$\sign_\varrho(0) = \varrho \in [-1,1]$;
cf. Remark~\ref{rem:sign}. Plugging back 
these lower bounds into \eqref{p:lem:sub:L1} and passing to the limit
$\epsilon \downarrow 0$ yields
\begin{equation} \label{p:lem:sub:L2}
  \begin{aligned}
    & \lim_{\epsilon \downarrow 0} \frac{L_T(\epsilon Y + (1-\epsilon)
      X) - L_T(X)}{\epsilon} \geq \frac{\kappa}{\eta} \int_0^T
    (\zeta^X_t - e^{-\kappa
      t} \zeta_0) (\zeta^Y_t - \zeta^X_t) dt \\
    & + \frac{1}{2} \left( \eta \varphi^X_T +
      \sign_\varrho(\varphi^X_T) \zeta^X_T \right) (\varphi^Y_T -
    \varphi^X_T) + \frac{1}{2} \!\left( |\varphi^X_T| + \frac{1}{\eta}
      (\zeta^X_T - e^{-\kappa T} \zeta_0) \!
    \right)  \\
    & \cdot \left( \zeta^Y_T - \zeta^X_T \right) + \frac{1}{2}
    \int_{[0,T]} \zeta_0 e^{-\kappa t} (dY_t^{\uparrow} -
    dX_t^{\uparrow} + dY_t^{\downarrow} - dX^{\downarrow}_t).
  \end{aligned}
\end{equation}
Next, let us express every term in \eqref{p:lem:sub:L2} as an integral
with respect to either $Y^{\uparrow} - X^{\uparrow}$ or
$Y^{\downarrow} - X^{\downarrow}$. Using~\eqref{eq:spreadsolution} for
the spreads $\zeta^Y$ and $\zeta^X$ as well as Fubini's Theorem, we
can rewrite the first term in~\eqref{p:lem:sub:L2} as
\begin{equation*}
  \begin{aligned}
    & \frac{\kappa}{\eta} \int_0^T (\zeta^X_t - e^{-\kappa
      t} \zeta_0) (\zeta^Y_t - \zeta^X_t) dt \\
    & = \kappa \int_{[0,T]} \left( \int_s^T  (\zeta^X_t - e^{-\kappa
      t} \zeta_0) e^{-\kappa (t-s)} dt \right) (dY_s^{\uparrow} -
      dX^{\uparrow}_s)  \\
    & \phantom{=} + \kappa \int_{[0,T]} \left( \int_s^T  (\zeta^X_t - e^{-\kappa
      t} \zeta_0) e^{-\kappa (t-s)} dt \right) (dY_s^{\downarrow} -
      dX^{\downarrow}_s). 
  \end{aligned}
\end{equation*}
Moreover, using that
$\varphi^Y_T - \varphi^X_T = \int_{[0,T]} (dY_s^{\uparrow} -
dX^{\uparrow}_s) - \int_{[0,T]} (dY_s^{\downarrow} -
dX_s^{\downarrow})$ as well as
$\zeta^Y_T - \zeta^X_T = \int_{[0,T]} \eta e^{-\kappa(T-s)}
(dY_s^{\uparrow} - dX^{\uparrow}_s) + \int_{[0,T]} \eta
e^{-\kappa(T-s)} (dY_s^{\downarrow} - dX_s^{\downarrow})$ allows us to
finally write \eqref{p:lem:sub:L2} as
\begin{equation} \label{p:lem:sub:L3}
  \begin{aligned}
    & \lim_{\epsilon \downarrow 0} \frac{L_T(\epsilon Y + (1-\epsilon)
      X) - L_T(X)}{\epsilon}
    \\
    & \geq \int_{[0,T]} {}^{\varrho}\nabla_s^{\uparrow} L_T(X)
    (dY_s^{\uparrow} - dX_s^{\uparrow}) + \int_{[0,T]}
    {}^{\varrho}\nabla_s^{\downarrow} L_T(X) (dY_s^{\downarrow} -
    dX_s^{\downarrow}),
  \end{aligned}
\end{equation}
where we set
\begin{align*}
  {}^{\varrho}\nabla_s^{\uparrow, \downarrow} L_T(X) \set 
  & \, \kappa \int_s^T
    e^{-\kappa (t-s)} \zeta^X_t dt +
    \frac{1}{2} (\eta |\varphi^X_T| + \zeta^X_T) e^{-\kappa (T-s)} \\
  & \pm
    \frac{\eta}{2} \varphi^X_T \pm \frac{1}{2} \sign_\varrho(\varphi^X_T)
    \zeta^X_T \qquad (0 \leq s \leq T). 
\end{align*}

\emph{Step 2:} Let us now compute the subgradients of the deviation
functional $D_T(\cdot)$ defined in \eqref{p:lem:sub:D1}. Again, for
any $X, Y \in \cX^d$ with $\varphi^{Y}_{0-}=\varphi^{X}_{0-}$ and any
$\epsilon \in (0, 1]$ we obtain
\begin{align*}
  & \frac{D_T(\epsilon Y + (1-\epsilon) X) - D_T(X)}{\epsilon} \\
  &   = \alpha \sigma^2 \int_0^T \left( \varphi^X_t
    - \frac{\mu}{\alpha \sigma^2} \right) (\varphi^Y_t - \varphi^X_t) dt
    + \epsilon \frac{\alpha \sigma^2}{2} \int_0^T (\varphi^Y_t - \varphi^X_t)^2 dt
\end{align*}
and hence, together with 
Fubini's Theorem, we arrive at
\begin{align*}
  & \lim_{\epsilon \downarrow 0} \frac{D_T(\epsilon Y + (1-\epsilon)
    X) 
    - D_T(X)}{\epsilon} \\
  & = \alpha \sigma^2 \int_{[0,T]} \left( \int_s^T \left(\varphi^X_t -
    \frac{\mu}{\alpha \sigma^2} \right) dt \right) (dY^{\uparrow}_s
    - dX^{\uparrow}_s)  \\
  & \phantom{=} + \alpha \sigma^2 \int_{[0,T]} \left( \int_s^T \left(
    \frac{\mu}{\alpha \sigma^2} - \varphi^X_t \right) dt \right) (dY^{\downarrow}_s
    - dX^{\downarrow}_s).
\end{align*}
Consequently, we can write
\begin{align}
  & \lim_{\epsilon \downarrow 0} \frac{D_T(\epsilon Y + (1-\epsilon)
    X) 
    - D_T(X)}{\epsilon} \nonumber \\
  & = \int_{[0,T]} \nabla_s^{\uparrow} D_T(X) (dY_s^{\uparrow} - dX_s^{\uparrow}) 
    + \int_{[0,T]} \nabla_s^{\downarrow} D_T(X) (dY_s^{\downarrow} -
    dX_s^{\downarrow}), \label{p:lem:sub:D2}
\end{align}
where we set
\begin{align}
  \nabla_s^{\uparrow, \downarrow} D_T(X) \set \pm \alpha \sigma^2
  \int_s^T \left(\varphi^X_t -
  \frac{\mu}{\alpha \sigma^2} \right) dt \qquad (0 \leq s \leq T). 
  \label{p:lem:sub:D3}
\end{align}

\emph{Step 3:} Finally, for the convex cost functional $J_T(\cdot)$ we
obtain for any $X, Y \in \cX^d$ with
$\varphi^{Y}_{0-}=\varphi^{X}_{0-}$ and any $\epsilon \in (0,1]$ the
lower bound
\begin{align*}
  &  J_T(Y) - J_T(X) \\
  & \geq\frac{ J_T(\epsilon Y+(1-\epsilon)X)-J_T(X)}{\epsilon} \\ 
  & = \frac{ L_T(\epsilon Y + (1-\epsilon) X) -
    L_T(X)}{\epsilon} + \frac{ D_T(\epsilon Y + (1-\epsilon) X) -
    D_T(X)}{\epsilon}.
\end{align*}
Passing to the limit $\epsilon \downarrow 0$ yields together with
\eqref{p:lem:sub:L3} and \eqref{p:lem:sub:D2}
\begin{align*}
  & \frac{J_T(\epsilon Y + (1-\epsilon) X) - J_T(X)}{\epsilon} \\
  & \geq \int_{[0,T]} ({}^{\varrho}\nabla_s^{\uparrow} L_T(X) +
    \nabla_s^{\uparrow} D_T(X)) (dY_s^{\uparrow} - dX_s^{\uparrow}) 
  \\
  & \phantom{\geq} + \int_{[0,T]} ({}^{\varrho}\nabla_s^{\downarrow} L_T(X) +
    \nabla_s^{\downarrow} D_T(X)) (dY_s^{\downarrow} -
    dX_s^{\downarrow}),
\end{align*}
where we note that
${}^{\varrho}\nabla_s^{\uparrow,\downarrow} J_T(X) =
{}^{\varrho}\nabla_s^{\uparrow,\downarrow} L_T(X) +
\nabla_s^{\uparrow,\downarrow} D_T(X)$ for all $s \in [0,T]$ as
desired.  \qed \medskip

\noindent{{\emph{Proof of Lemma~\ref{lem:nobuysell}.}}}
For a stragegy $X=(X^{\uparrow},X^{\downarrow}) \in \cX^d$,
$X \neq (0,0)$, let $t \in [0,T]$ be such that
$\nabla_t^{\uparrow} J_T(X) = 0$. Using the definition of
$\nabla^{\uparrow}_t J_T(X)$ in \eqref{eq:buysubgradient} this amounts to
the identity
\begin{align*}
  &
    - \frac{\eta}{2} \varphi^X_T - \frac{1}{2} \sign(\varphi^X_T)
    \zeta^X_T \\
  & = \int_t^T \left( \kappa e^{-\kappa (u-t)} \zeta_u^X 
    + \alpha \sigma^2 \left( \varphi^X_u - \frac{\mu}{\alpha\sigma^2}
    \right) \right) du 
    + \frac{1}{2} \left( \eta \vert \varphi^X_T \vert + \zeta_T^X
    \right) e^{-\kappa (T-t)}.
\end{align*}
Plugging this in the definition of $\nabla^{\downarrow}_t J_T(X)$ in
\eqref{eq:sellsubgradient} yields
\begin{equation*}
  \nabla_t^{\downarrow} J_T(X) = 2 \int_t^T \kappa e^{-\kappa (u-t)} \zeta_u^X du + \left( \eta \vert \varphi^X_T \vert + \zeta_T^X
  \right) e^{-\kappa (T-t)} > 0
\end{equation*}
because $X \neq (0,0)$. The same reasoning applies when the roles of
$\uparrow$ and $\downarrow$ are interchanged.  \qed

\subsection{Proofs of Section~\ref{subsec:mainresult}}
\label{subsec:proofs3}

In this section we prove our main Theorem~\ref{thm:main} together with
Corollaries~\ref{cor:main1} and~\ref{cor:main2}. We start with
introducing the two key objects, that is, the free boundary functions
$\phi_{\mathrm{buy}}(\tau,\zeta)$ and
$\phi_{\mathrm{sell}}(\tau,\zeta)$ on the domain $[0,+\infty)^2$.

\subsubsection{The free boundary functions}
\label{subsubsec:proofs:aux}

Introducing the function $\phi_{\mathrm{sell}}$ is
straightforward. Recall that $\lambda = \sqrt{\alpha} \sigma$ and
$\beta = \kappa\lambda/\sqrt{\lambda^2+\kappa\eta}$. We set
\begin{equation} \label{eq:def:lambdagamma} 
  \gamma_{\pm} \set \lambda
  \pm \sqrt{\kappa\eta + \lambda^2}
\end{equation}
and denote
\begin{equation}
  C(\tau) \set \frac{e^{-\beta\tau}\gamma_- + 
    e^{\beta\tau}\gamma_+}{e^{-\beta\tau}\gamma^2_- + 
    e^{\beta\tau}\gamma^2_+}, \quad 
  D(\tau) \set 1- \frac{2\kappa\eta}{e^{-\beta\tau}\gamma^2_- +
    e^{\beta\tau}\gamma^2_+} \label{eq:def:CandD} \quad (\tau \geq 0).
\end{equation}
On the domain $[0,+\infty)^2$, the free boundary
function $(\tau,\zeta) \mapsto \phi_{\mathrm{sell}}(\tau,\zeta)$ will
then be defined as
\begin{equation} \label{def:phisell} 
  \phi_{\mathrm{sell}}(\tau,\zeta)
  \set \frac{\mu}{\lambda^2} D(\tau) + \zeta \frac{\kappa}{\lambda}
  C(\tau) \quad (\tau \geq 0, \zeta \geq 0).
\end{equation}
Let us note the following property which can be easily checked:

\begin{Lemma} \label{lem:phisell} 
  We have $\phi_{\mathrm{sell}}(\tau,\zeta) > 0$ for all
  $\tau \geq 0$, $\zeta \geq 0$, since $1 > D(\tau) > 0$ and
  $C(\tau) > 0$ for all $\tau \geq 0$. \qed
\end{Lemma}

\begin{Remark} \label{rem:negphisell} 
  By a slight abuse of the definition of the function
  $\phi_{\mathrm{sell}}$ in~\eqref{def:phisell} which is only confined to
  the positive half-plane $[0,+\infty)^2$, we will
  also use for $\zeta > 0$ the notation
  $\phi_{\mathrm{sell}}(\tau,-\zeta)$ with the obvious meaning
  $\phi_{\mathrm{sell}}(\tau,-\zeta) \set \mu D(\tau)/\lambda^2 - \zeta
  \kappa C(\tau) /\lambda$.
\end{Remark}

In contrast to $\phi_{\mathrm{sell}}$ in~\eqref{def:phisell},
introducing the free boundary function
$(\tau,\zeta) \mapsto \phi_{\mathrm{buy}}(\tau,\zeta)$ on the domain
$[0,+\infty)^2$ is much more intricate and
necessitates several auxiliary constants and functions. First, let
$\bar{\theta} > 0$ denote the unique strictly positive solution to the
equation
\begin{equation} \label{eq:besslich1} 
  e^{\kappa \bar{\theta}} (2-\kappa
  \bar{\theta}) + 2 + \kappa \bar{\theta} = 0
\end{equation} 
and let $\underline{\theta} \in (0,\bar\theta)$ denote the unique
solution to the equation
\begin{equation} \label{eq:besslich2} 
  e^{\kappa \underline\theta}
  (\kappa \underline{\theta} - 1) = 1.
\end{equation} 
Next, we introduce the mapping
$\tau \mapsto (\bar\zeta(\tau),\bar\varphi(\tau))$ for all
$\tau \geq 0$ via
\begin{equation} \label{def:barzeta} 
  \bar{\zeta}(\tau) \set
  \begin{cases}
    s_1(\tau-\bar\theta,\bar\theta), & \tau > \bar\theta, \\
    s_2(\tau), &\underline\theta < \tau \leq \bar\theta, \\
    2\mu/\kappa, & 0 \leq \tau \leq \underline\theta,
  \end{cases}
\end{equation}
with
\begin{align}
  s_1(\tau,\theta) & \set 
                     \frac{\mu(1-D(\tau)) e^{\kappa\theta}}{\lambda\kappa
                     C(\tau)+\frac{\kappa}{2} e^{\kappa\theta}
                     \frac{(1+e^{-\kappa\theta})^2}{\kappa\theta-1-e^{-\kappa\theta}}}
                     \quad (\tau \geq 0, \theta \geq 0), \label{def:s1} \\
  s_2(\tau) & \set \mu\eta
              \frac{\kappa\tau e^{-\kappa\tau}+1+e^{-\kappa\tau}}{\lambda^2
              \kappa\tau +
              \frac{\kappa\eta}{2}(1+e^{-\kappa\tau})^2-\lambda^2(1+e^{-\kappa\tau})}
              \quad (\underline\theta \leq \tau \leq \bar\theta), \label{def:s2}
\end{align}
and
\begin{equation} \label{def:barphi}
  \bar{\varphi}(\tau) \set
  \begin{cases}
    \phi_{\textrm{sell}}(\tau-\bar\theta,\bar\zeta(\tau)
    e^{-\kappa\bar\theta}), & \tau > \bar\theta, \\
    \phi_2(\tau,\bar\zeta(\tau)), &\underline\theta < \tau \leq \bar\theta, \\
    0, & 0 \leq \tau \leq \underline\theta,
  \end{cases}
\end{equation}
where
\begin{equation} \label{def:phi2}
  \phi_2(\tau,\zeta)  \set 
  \frac{\mu\tau -\frac{1}{2} \zeta (1+e^{-\kappa\tau})}{\lambda^2
    \tau +
    \frac{\eta}{2}(1+e^{-\kappa\tau})} \quad (\tau \geq 0,
  \zeta \geq 0). 
\end{equation}
We further set
\begin{equation} \label{def:s3}
  s_3(\tau) \set \frac{2\mu\tau}{1+e^{-\kappa\tau}} \quad (0 \leq \tau
  \leq \underline\theta).
\end{equation}
Let us mention that since $\bar\theta$ satisfies~\eqref{eq:besslich1}
it holds in~\eqref{def:s1} that
\begin{equation} \label{eq:barzeta1}
  s_1(\tau-\bar\theta,\bar\theta) =
  \frac{\mu (1-D(\tau-\bar\theta)) e^{\kappa\bar\theta} }{\lambda\kappa
    C(\tau-\bar\theta) + \frac{\kappa}{2} e^{\kappa
      \bar\theta}(1-e^{-\kappa \bar\theta})} > 0 \quad (\tau >
  \bar\theta).
\end{equation}
Moreover, direct computations reveal that 
\begin{equation} \label{eq:barphi1} 
  \phi_{\textrm{sell}}(\tau-\bar\theta,\bar\zeta(\tau)
  e^{-\kappa\bar\theta}) = 
  \frac{\mu}{\lambda^2} - \frac{\kappa}{2} \bar{\zeta}(\tau)
  \frac{1-e^{-\kappa
      \bar{\theta}}}{\lambda^2} \quad (\tau > \bar\theta)
\end{equation}
as well as
\begin{equation} \label{eq:barphi2} 
  \phi_{2}(\tau,\bar\zeta(\tau)) =
  \frac{\mu \kappa \tau - \mu (1+e^{-\kappa \tau})}{\kappa \lambda^2
    \tau + \frac{\kappa \eta}{2} (1+e^{-\kappa \tau})^2 -
    \lambda^2(1+e^{-\kappa \tau})} \quad (\underline\theta < \tau \leq
  \bar\theta ).
\end{equation}
It will turn out that
$\tau \mapsto (\tau,\bar\zeta(\tau),\bar\varphi(\tau))$ specifies a
curve which is embedded in the free boundary
$\partial \cR_{\mathrm{buy}}$ in the state space $\cS$. The next lemma
collects some useful properties concerning the maps $s_1, s_2,s_3$
introduced in~\eqref{def:s1}, \eqref{def:s2}, \eqref{def:s3},
respectively, as well as this curve. We also refer to the
graphical illustration in Figure~\ref{fig:domain} below in this
context.

\begin{Lemma} \label{lem:fbcurve}
$\phantom{}$
\vspace{-.5em}
  \begin{enumerate}
  \item We have $s_3(0)=0$ and
    $s_3(\underline\theta)=2\mu/\kappa$. Moreover, on the interval
    $(0,\underline\theta)$, the map $\tau \mapsto s_3(\tau)$ is
    strictly increasing. In particular, it holds that
    $s_3(\tau) < 2\mu/\kappa$ on $(0,\underline\theta)$.
  \item We have $s_1(0,\underline\theta)=0$ and
    $s_2(\underline\theta)=2\mu/\kappa$ as well as
    $s_1(0,\bar\theta)=s_2(\bar\theta)$. Moreover, on the interval
    $(\underline\theta,\bar\theta)$, the map
    $\tau \mapsto s_1(0,\tau)$ is strictly increasing and the map
    $\tau \mapsto s_2(\tau)$ is stricly decreasing. In particular, it
    holds that $s_1(0,\tau) < s_2(\tau)$ on
    $(\underline\theta,\bar\theta)$.
  \item The map $\tau \mapsto \bar\zeta(\tau)$, $\tau \geq 0$, is
    continuous, flat on $[0,\underline\theta)$, and strictly
    decreasing on $[\underline\theta , + \infty)$ with
    $\lim_{\tau \uparrow \infty} \bar\zeta(\tau) = 0$. In particular,
    we have $2\mu/\kappa \geq \bar\zeta(\tau) > 0$ for all
    $\tau \geq 0$.
  \item The map $\tau \mapsto \bar\varphi(\tau)$, $\tau \geq 0$, is
    continuous, flat on $[0,\underline\theta)$, and strictly
    increasing on $[\underline\theta , + \infty)$ with
    $\lim_{\tau \uparrow \infty} \bar\varphi(\tau) =
    \mu/\lambda^2$. In particular, we have
    $0 \leq \bar\varphi(\tau) < \mu/\lambda^2$ for all $\tau \geq 0$.
  \end{enumerate}
\end{Lemma}

\begin{proof}
  The claims follow from the fact that $\underline\theta$ and
  $\bar\theta$ satisfy equation~\eqref{eq:besslich2}
  and~\eqref{eq:besslich1}, respectively, as well as simple
  differentiation of the mappings with respect to~$\tau$ (recall also
  the representations of $\bar\varphi(\cdot)$ in~\eqref{eq:barphi1}
  and~\eqref{eq:barphi2}). Finally, observe that
  $\lim_{\tau \uparrow \infty} \bar\zeta(\tau) = 0$ can be deduced
  from~\eqref{eq:barzeta1} as well as
  $\lim_{\tau \uparrow \infty} \bar\varphi(\tau) = \mu/\lambda^2$
  from~\eqref{eq:barphi1}.
\end{proof}

Next, let us introduce for all $\tau \geq 0$,
$\zeta, \varphi \in \mathbb{R}$ and $0 \leq \theta \leq \tau$ the
mappings
\begin{equation} \label{def:zetahatbuy} 
  \begin{aligned}
    \hat{\zeta}^{\mathrm{buy}}(\tau,\zeta,\varphi,\theta) \set & \;
    \zeta \frac{\eta\beta^2}{2\lambda^2} \left(
      \frac{e^{-\beta(\tau-\theta)}}{\kappa+\beta}
      + \frac{e^{\beta(\tau-\theta)}}{\kappa-\beta} \right) \\
    & - \frac{\eta\beta}{\kappa} \left( \varphi-\frac{\mu}{\lambda^2}
    \right) \sinh\left(\beta(\tau-\theta)\right),
  \end{aligned}
\end{equation}
\begin{equation}
  \begin{aligned} \label{def:phihatbuy}
    \hat{\varphi}^{\mathrm{buy}}(\tau,\zeta,\varphi,\theta) \set 
    & \; 
    \left( \varphi - \frac{\mu}{\lambda^2} \right)
    \cosh(\beta(\tau-\theta)) \\
    & - \frac{\beta}{\kappa} \sinh(\beta(\tau-\theta)) \left(
      \varphi-\frac{\mu}{\lambda^2} + \frac{\kappa}{\lambda^2} \zeta
    \right) + \frac{\mu}{\lambda^2} 
  \end{aligned}
\end{equation}
as well as
\begin{equation}
  \begin{aligned} \label{def:zetahatsell}
    \hat{\zeta}^{\mathrm{sell}}(\tau,\zeta,\varphi,\theta) \set & \; 
    \zeta e^{-\kappa\theta} + \eta e^{-\kappa\theta} \bigg(
    \frac{\beta}{\kappa+\beta} c_+(\tau,\zeta,\varphi)
    (e^{(\kappa+\beta)\theta} - 1 )\bigg. \\
    &   \hspace{83pt} \bigg. + \frac{\beta}{\beta-\kappa} c_-(\tau,\zeta,\varphi)
    (e^{-(\beta-\kappa)\theta} - 1 )\bigg),
  \end{aligned}
\end{equation}
\begin{equation}
  \begin{aligned} \label{def:phihatsell}
    \hat{\varphi}^{\mathrm{sell}}(\tau,\zeta,\varphi,\theta) \set & \; 
    -c_+(\tau,\zeta,\varphi) e^{\beta\theta} - c_-(\tau,\zeta,\varphi)
    e^{-\beta\theta}
    + \frac{\mu}{\lambda^2}, 
  \end{aligned}
\end{equation}
where
\begin{equation}
  c_{\pm}(\tau,\zeta,\varphi) \set 
  \frac{\kappa (e^{\mp\beta\tau} \gamma_{\mp} (\eta\mu - \lambda^2
    (\zeta+\eta\varphi))
    +\eta\mu\gamma_{\pm})}{\lambda^2 \sqrt{\kappa\eta+\lambda^2} 
    (e^{\beta\tau} \gamma_+^2 - e^{-\beta\tau} \gamma_-^2)}. 
  \label{def:cpm}
\end{equation} 
In fact, simple computations reveal the identities
\begin{equation} \label{eq:propbuy}
  \hat{\zeta}^{\mathrm{buy}}(\tau,\zeta,\varphi,\tau) = \zeta \quad
  \text{and} \quad
  \hat{\varphi}^{\mathrm{buy}}(\tau,\zeta,\varphi,\tau) = \varphi
\end{equation}
for all $(\tau,\zeta,\varphi) \in \cS$ as well as
\begin{equation} \label{eq:propsell}
  \hat{\zeta}^{\mathrm{sell}}(\tau,\zeta,\phi_{\mathrm{sell}}(\tau,\zeta),0)
  = \zeta \quad \text{and} \quad
  \hat{\varphi}^{\mathrm{sell}}(\tau,\zeta,\phi_{\mathrm{sell}}(\tau,\zeta),0)
  = \phi_{\mathrm{sell}}(\tau,\zeta)
\end{equation}
for all $\tau \geq 0$ and $\zeta \in \mathbb{R}$.  Moreover, the
following lemma can be easily verified by elementary calculus which we
omit for the sake of brevity.

\begin{Lemma}[Monotonicity properties] \label{lem:hatfunctions} 
$\phantom{}$
\vspace{-.5em}
  \begin{enumerate}
  \item For any $\theta \geq 0$, the function
    $\tau \mapsto
    \hat\zeta^{\mathrm{buy}}(\tau,\bar\zeta(\theta),\bar\varphi(\theta),\theta)$,
    $\tau \geq \theta$, is continuous and strictly increasing with
    $\hat\zeta^{\mathrm{buy}}(\theta,\bar\zeta(\theta),\bar\varphi(\theta),\theta)
    =\bar\zeta(\theta)$.  Moreover, for any two
    $0 \leq \theta_1 < \theta_2$, the functions do not intersect on
    $[\theta_2,+\infty)$.
  \item For any $\tau \geq 0$, the function
    $z \mapsto \hat{\zeta}^{\mathrm{buy}}\left( \tau,
      \bar\zeta\left(\tau-z \right), \bar\varphi\left(\tau-z\right),
      \tau - z \right)$, $0 \leq z \leq \tau$, is continuous and
    strictly increasing.
  \item For any $\theta \geq 0$, the function
    $\tau \mapsto
    \hat\varphi^{\mathrm{buy}}(\tau,\bar\zeta(\theta),\bar\varphi(\theta),\theta)$,
    $\tau \geq \theta$, is continuous and strictly decreasing with
    $\hat\varphi^{\mathrm{buy}}(\theta,\bar\zeta(\theta),\bar\varphi(\theta),\theta)
    =\bar\varphi(\theta)$.
  \item For any $\tau \geq 0$, $\zeta \geq 0$, the function
    $t \mapsto
    \hat\varphi^{\mathrm{sell}}(\tau,\zeta,\phi_{\mathrm{sell}}(\tau,\zeta),t)$
    on $[0,\tau]$ is continuous and strictly decreasing with
    $\hat\varphi^{\mathrm{sell}}(\tau,\zeta,\phi_{\mathrm{sell}}(\tau,\zeta),0)
    = \phi_{\mathrm{sell}}(\tau,\zeta)$.
  \item For any $\tau \geq 0$, $\zeta > 0$, the function
    $t \mapsto
    \hat\varphi^{\mathrm{sell}}(\tau,-\zeta,\phi_{\mathrm{sell}}(\tau,-\zeta),t)$
    on $[0,\tau]$ is continuous and strictly increasing with
    $\hat\varphi^{\mathrm{sell}}(\tau,-\zeta,\phi_{\mathrm{sell}}(\tau,-\zeta),0)$
    $= \phi_{\mathrm{sell}}(\tau,-\zeta)$.
\qed
  \end{enumerate}
\end{Lemma}

As it will turn out below, for a given problem data
$(\tau,\zeta,\varphi)$ belonging to
$\partial \mathcal{R}_{\mathrm{buy}}$ or
$\partial \mathcal{R}_{\mathrm{sell}}$, the optimal share holdings
$\varphi^{\hat{X}^{\tau,\zeta,\varphi}}$ as well as the optimally
controlled spread dynamics $\zeta^{\hat{X}^{\tau,\zeta,\varphi}}$ of
the optimal policy $\hat{X}^{\tau,\zeta,\varphi}$ will be given in
terms of the mappings introduced in~\eqref{def:zetahatbuy}
to~\eqref{def:phihatsell}. Two further important ingredients are
provided by the following two lemmas.

\begin{Lemma}[Buying duration] \label{lem:taubuy} 
  For a given pair $(\tau,\zeta) \in [0,+\infty)^2$ such that
  $\bar\zeta(\tau) \leq \zeta \leq
  \hat{\zeta}^{\mathrm{buy}}(\tau,\bar\zeta(0),\bar\varphi(0),0)$, we
  define $\tau_{\mathrm{buy}}(\tau,\zeta)$ as the unique solution in
  $[0,\tau]$ to the equation
  \begin{equation} \label{def:taubuy} \zeta =
    \hat{\zeta}^{\mathrm{buy}}\left( \tau,
      \bar\zeta\left(\tau-\tau_{\mathrm{buy}}(\tau,\zeta)\right),
      \bar\varphi\left(\tau-\tau_{\mathrm{buy}}(\tau,\zeta)\right),
      \tau -\tau_{\mathrm{buy}}(\tau,\zeta) \right).
  \end{equation}
  In particular, it holds that 
  \begin{equation} \label{eq:taubuy:property} \tau_{\mathrm{buy}}
    \big( \tau,
    \hat{\zeta}^{\mathrm{buy}}(\tau,\bar\zeta(\theta),\bar\varphi(\theta),\theta)
    \big) = \tau - \theta \quad (0 \leq \theta \leq \tau)
  \end{equation}
  which implies
  $\tau_{\mathrm{buy}}(\tau,
  \hat{\zeta}^{\mathrm{buy}}(\tau,\bar\zeta(0),\bar\varphi(0),0)) =
  \tau$ and $\tau_{\mathrm{buy}}(\tau,\bar\zeta(\tau)) = 0$. We
  further set
  \begin{equation}
    \begin{aligned} \label{set:taubuy}
      \tau_{\mathrm{buy}}(\tau,\zeta) \set
      \begin{cases}
        \tau & \textrm{ for } \zeta >
        \hat{\zeta}^{\mathrm{buy}}(\tau,\bar\zeta(0),\bar\varphi(0),0), \\
        0 & \textrm{ for } 0 \leq \zeta < \bar\zeta(\tau),
      \end{cases}
    \end{aligned}
  \end{equation}
  so that $\tau_{\mathrm{buy}}(\cdot,\cdot)$ is defined for all
  $(\tau,\zeta)\in[0,\infty)^2$ with values in
  $[0,\tau]$.
\end{Lemma}

\begin{proof}
  For any $\tau \geq 0$ consider the mapping
  $z \mapsto F_\tau(z) \set \hat{\zeta}^{\mathrm{buy}}( \tau,
    \bar\zeta\left(\tau-z \right),$ $\bar\varphi\left(\tau-z\right),\tau
    -z)$ with $z \in [0,\tau]$. Then,
  $F_\tau(0) = \bar\zeta(\tau)$ due to~\eqref{eq:propbuy} as well as
  $F_\tau(\tau) = \hat{\zeta}^{\mathrm{buy}}\left( \tau,
    \bar\zeta\left(0 \right), \bar\varphi\left(0\right), 0
  \right)$. Moreover, it follows from Lemma~\ref{lem:hatfunctions} 2.)
  that $F_\tau(z)$ is continuous and strictly increasing on
  $[0,\tau]$. Consequently, for any
  $\bar\zeta(\tau) \leq \zeta \leq
  \hat{\zeta}^{\mathrm{buy}}(\tau,\bar\zeta(0),\bar\varphi(0),0)$
  there exists a unique
  $\tau_{\mathrm{buy}}(\tau,\zeta) \set z^*\in [0,\tau]$ such that
  $\zeta = F_\tau(z^*)$.
\end{proof}

\begin{Lemma}[Waiting duration] \label{lem:tauwait} 
   For a given pair $(\tau,\zeta) \in [0,\infty)^2$ such that either
  $\tau \geq \bar\theta$ and $0 \leq \zeta \leq \bar\zeta(\tau)$, or
  $\underline\theta \leq \tau < \bar\theta$ and
  $0 \leq \zeta \leq s_1(0,\tau)$, we define
  $\tau_{\mathrm{wait}}(\tau,\zeta)$ as the unique solution in
  $(0,\tau]$ to the equation
  \begin{equation} \label{def:tauwait} 
    \zeta = s_1\left(
      \tau-\tau_{\mathrm{wait}}(\tau,\zeta),
      \tau_{\mathrm{wait}}(\tau,\zeta) \right).
  \end{equation}
  In particular, in case $\tau \geq \bar\theta$ we have
  $\tau_{\mathrm{wait}}(\tau,\bar\zeta(\tau)) = \bar\theta$ and in
  case $\underline\theta \leq \tau < \bar\theta$ we have
  $\tau_{\mathrm{wait}}(\tau,s_1(0,\tau)) = \tau$. We further set
  \begin{equation}
    \begin{aligned} \label{set:tauwait}
      \tau_{\mathrm{wait}}(\tau,\zeta) \set
      \begin{cases}
        \bar\theta & \textrm{for } \tau \geq \bar\theta \textrm{ and
        } \\
& \bar\zeta(\tau) < \zeta \leq
        \hat{\zeta}^{\mathrm{buy}}(\tau,\bar\zeta(\bar\theta),\bar\varphi(\bar\theta),
        \bar\theta), \\
        \tau-\tau_{\mathrm{buy}}(\tau,\zeta) & \textrm{in all
          remaining cases},
      \end{cases}
    \end{aligned}
  \end{equation}
  so that $\tau_{\mathrm{wait}}(\cdot,\cdot)$ is defined for all
  $(\tau,\zeta)\in[0,\infty)^2$ with values in $[0,\tau]$.
\end{Lemma}

\begin{proof}
  Consider for any $\tau \geq \underline\theta$ arbitrary but fixed
  the continuous function $z \mapsto G_\tau(z) \set s_1(\tau - z, z)$
  with $z \in [0,\min\{\tau,\bar\theta\}]$. An elementary computation
  shows that $G_\tau(z)$ is strictly increasing on
  $[0, \min\{\tau,\bar\theta\}]$ with $G_\tau(0) = s_1(\tau,0) <
  0$. Moreover, in case $\tau \geq \bar\theta$ it holds that
  $G_\tau(\bar\theta) = \bar\zeta(\tau)$ due to the definition
  in~\eqref{def:barzeta}, and in case
  $\underline\theta \leq \tau < \bar\theta$ it holds that
  $G_\tau(\tau) = s_1(0,\tau)$. Consequently, when
  $\tau \geq \bar \theta$ equation~\eqref{def:tauwait} admits for
  every $0 \leq \zeta \leq \bar\zeta(\tau)$ a unique solution
  $\tau_{\mathrm{wait}}(\tau,\zeta) \in (0,\bar\theta]$. Similarly,
  when $\underline\theta \leq \tau < \bar\theta$
  equation~\eqref{def:tauwait} admits for every
  $0 \leq \zeta \leq s_1(0,\tau)$ a unique solution
  $\tau_{\mathrm{wait}}(\tau,\zeta) \in (0,\tau]$.
\end{proof}

We are now ready to introduce the second free boundary function
$(\tau,\zeta) \mapsto \phi_{\mathrm{buy}}(\tau,\zeta)$ on the domain
$[0,+\infty)^2$. For given $\tau \geq 0$,
$\zeta \geq 0$,  we distinguish the following cases:
\begin{enumerate}
\item For
  $\zeta >
  \hat{\zeta}^{\mathrm{buy}}(\tau,\bar\zeta(0),\bar\varphi(0),0) =
  \hat{\zeta}^{\mathrm{buy}}(\tau,2\mu/\kappa,0,0)$ we set
  \begin{equation} \label{def:phibuy:eq1}
    \phi_{\mathrm{buy}}(\tau,\zeta) \set
    \phi_{\mathrm{sell}}(\tau,-\zeta)
  \end{equation}
  with $\phi_{\mathrm{sell}}$ as given in~\eqref{def:phisell} together
  with Remark~\ref{rem:negphisell}.
\item For
  $\bar\zeta(\tau) < \zeta \leq
  \hat{\zeta}^{\mathrm{buy}}(\tau,\bar\zeta(0),\bar\varphi(0),0)$ we
  set
  \begin{equation} \label{def:phibuy:eq2}
    \begin{aligned}
      & \phi_{\mathrm{buy}}(\tau,\zeta) \\
& \set
      \hat{\varphi}^{\mathrm{buy}} \left( \tau,
        \bar\zeta(\tau-\tau_{\mathrm{buy}}(\tau,\zeta)),
        \bar\varphi(\tau-\tau_{\mathrm{buy}}(\tau,\zeta)), \tau
        -\tau_{\mathrm{buy}}(\tau,\zeta) \right)
    \end{aligned}
  \end{equation}
  with $\hat{\varphi}^{\mathrm{buy}}$ and $\tau_{\mathrm{buy}}(\tau,\zeta)$ as defined
  in~\eqref{def:phihatbuy} and~\eqref{def:taubuy}, respectively.
\item If $0 \leq \zeta \leq \bar\zeta(\tau)$ and
  \begin{enumerate}
  \item if, in addition, $\tau \geq \bar\theta$, we set
    \begin{equation} \label{def:phibuy:eq3a}
      \phi_{\mathrm{buy}}(\tau,\zeta) \set \phi_{\mathrm{sell}}\left(
        \tau-\tau_{\mathrm{wait}}(\tau,\zeta), \zeta e^{-\kappa
          \tau_{\mathrm{wait}}(\tau,\zeta)} \right),
    \end{equation}
    with $\tau_{\mathrm{wait}}(\tau,\zeta)$ as defined
    in~\eqref{def:tauwait};
  \item if, in addition, $\underline\theta \leq \tau < \bar\theta$, we set
    \begin{equation} \label{def:phibuy:eq3b}
      \begin{aligned}
        & \phi_{\mathrm{buy}}(\tau,\zeta) \\ & \set
        \begin{cases}
          \phi_2(\tau,\zeta) & \text{if } \zeta > s_1(0,\tau), \\
          \phi_{\mathrm{sell}} \big(
            \tau-\tau_{\mathrm{wait}}(\tau,\zeta), \zeta e^{-\kappa
              \tau_{\mathrm{wait}}(\tau,\zeta)} \big) &
          \text{if } 0 \leq \zeta \leq s_1(0,\tau),
        \end{cases}
      \end{aligned}
    \end{equation}
    with $\phi_2$ given in~\eqref{def:phi2} and
    $\tau_{\mathrm{wait}}(\tau,\zeta)$ defined
    in~\eqref{def:tauwait};
  \item if, in addition, $0 \leq \tau < \underline\theta$, we set
    \begin{equation} \label{def:phibuy:eq3c}
      \begin{aligned}
        \phi_{\mathrm{buy}}(\tau,\zeta) \set
        \begin{cases}
          0 & \text{if } \zeta > s_3(\tau), \\
          \phi_2( \tau,\zeta) & \text{if } 0 \leq \zeta \leq s_3(\tau),
        \end{cases}
      \end{aligned}
    \end{equation}
    with $\phi_2$ and $s_3$ given in~\eqref{def:phi2}
    and~\eqref{def:s3}, respectively.
  \end{enumerate}
\end{enumerate}

Notice that together with the properties of the functions
$s_1, s_2, s_3$ collected in Lemma~\ref{lem:fbcurve} 1.) and 2.), the
above cases from~\eqref{def:phibuy:eq1} to~\eqref{def:phibuy:eq3c}
fully determine a map
$(\tau,\zeta) \mapsto \phi_{\mathrm{buy}}(\tau,\zeta)$ on the domain
$[0,+\infty)^2$; cf. Figure~\ref{fig:domain} for the corresponding partition of the
$(\tau,\zeta)$-half-plane.

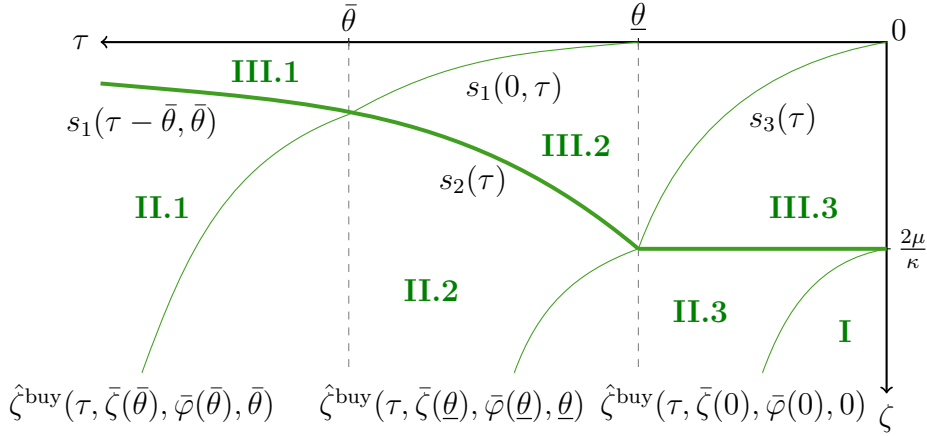
\begin{figure}[h]
\centering
\begin{tikzpicture}[scale=.55]
\draw [black, thick, <->] (-16,3) -- (3,3) -- (3,-5.5);
\node at (3.3,3.3) {0};
\node [left] at (-16,3) {$\tau$};
\node [below] at (3,-5.5) {$\zeta$}; 
\node[right] at (3,-2) {$\frac{2\mu}{\kappa}$};
\draw (2.9,-2) -- (3.1,-2); 
\node [above] at (-3,3) {$\underline\theta$}; 
\draw (-3,2.9) -- (-3,3.1); 
\draw [dashed, gray] (-3, 3) -- (-3,-5);
\node [above] at (-10,3) {$\bar\theta$}; 
\draw (-10,2.9) -- (-10,3.1); 
\draw [dashed, gray] (-10, 3) -- (-10,-5);
\draw [color={rgb:red,2;green,5;blue,1}, line width = 1.5] (-3,-2)
-- (3,-2);
\draw [color={rgb:red,2;green,5;blue,1}] (3,3) to [out=190,in=70] (-3,-2);
\node [above] at (.5,.5) {$s_3(\tau)$}; 
\draw [color={rgb:red,2;green,5;blue,1}, line width = 1.5] (-16,2) to [out=-5,in=140] (-3,-2);
\node [above] at (-7,-1) {$s_2(\tau)$}; 
\node [below] at (-15,1.75) {$s_1(\tau-\bar\theta,\bar\theta)$}; 
\draw [color={rgb:red,2;green,5;blue,1}] (-10,1.25) to [out=30,in=185] (-3,3);
\node [above] at (-6,1.25) {$s_1(0,\tau)$}; 
\draw [color={rgb:red,2;green,5;blue,1}] (0,-5) to [out=70,in=190] (3,-2);
\node [below] at (-.7,-5)
{$\hat{\zeta}^{\mathrm{buy}}(\tau,\bar\zeta(0),\bar\varphi(0),0)$};
\draw [color={rgb:red,2;green,5;blue,1}] (-6,-5) to [out=70,in=200] (-3,-2);
\node [below] at (-7.5,-5)
{$\hat{\zeta}^{\mathrm{buy}}(\tau,\bar\zeta(\underline\theta),\bar\varphi(\underline\theta),
  \underline\theta)$};
\draw [color={rgb:red,2;green,5;blue,1}] (-15,-5) to [out=70,in=200] (-10,1.25);
\node [below] at (-15,-5)
{$\hat{\zeta}^{\mathrm{buy}}(\tau,\bar\zeta(\bar\theta),\bar\varphi(\bar\theta),\bar\theta)$};
\node at (-14.5,-1) {\textcolor[rgb]{0,0.5,0}{\textbf{II.1}}};
\node at (-12,-3) {\textcolor[rgb]{0,0.5,0}{\textbf{}}};
\node at (-8,-3) {\textcolor[rgb]{0,0.5,0}{\textbf{II.2}}};    
\node at (-4.5,-4) {\textcolor[rgb]{0,0.5,0}{\textbf{}}};
\node at (-1.5,-3.5) {\textcolor[rgb]{0,0.5,0}{\textbf{II.3}}};
\node at (2,-4) {\textcolor[rgb]{0,0.5,0}{\textbf{I}}};
\node at (-12,2.3) {\textcolor[rgb]{0,0.5,0}{\textbf{III.1}}};
\node at (-9,2.5) {\textcolor[rgb]{0,0.5,0}{\textbf{}}};
\node at (-4.5,0.5) {\textcolor[rgb]{0,0.5,0}{\textbf{III.2}}};
\node at (-1.5,2) {\textcolor[rgb]{0,0.5,0}{\textbf{}}};
\node at (1,-1) {\textcolor[rgb]{0,0.5,0}{\textbf{III.3}}};
\end{tikzpicture}
\caption{An illustration of the mappings
  $s_1(\tau-\bar\theta,\bar\theta)$ on $[\bar\theta,+\infty)$,
  $s_1(0,\tau)$ and $s_2(\tau)$ on $[\underline\theta,\bar\theta]$,
  $s_3(\tau)$ on $[0,\underline\theta]$ as well as
  $\hat{\zeta}^{\mathrm{buy}}(\tau,\bar\zeta(0),\bar\varphi(0),0)$ on
  $[0,+\infty)$,
  $\hat{\zeta}^{\mathrm{buy}}(\tau,\bar\zeta(\underline\theta),
  \bar\varphi(\underline\theta), \underline\theta)$ on
  $[\underline\theta,+\infty)$ and
  $\hat{\zeta}^{\mathrm{buy}}(\tau,\bar\zeta(\bar\theta),\bar\varphi(\bar\theta),
  \bar\theta)$ on $[\bar\theta,+\infty)$ in the
  $(\tau,\zeta)$-half-plane as functions in $\tau$; cf. also
  Lemma~\ref{lem:fbcurve} and Lemma~\ref{lem:hatfunctions}. The thick
  green curve depicts the map $\tau \mapsto \bar\zeta(\tau)$ for
  $\tau \geq 0$; cf. the definition in~\eqref{def:barzeta}.}
\label{fig:domain}
\end{figure}

\begin{Lemma}
  The map $(\tau,\zeta) \mapsto \phi_{\mathrm{buy}}(\tau,\zeta)$
  defined from~\eqref{def:phibuy:eq1} to~\eqref{def:phibuy:eq3c} is
  continuous on $[0,+\infty)^2$. In particular, we
  have $\phi_{\mathrm{buy}}(\tau,\bar\zeta(\tau)) = \bar\varphi(\tau)$
  for all $\tau \geq 0$.
\end{Lemma}

\begin{proof}
  Appealing to the continuity of the functions $\phi_{\mathrm{sell}}$,
  $\hat{\varphi}^{\mathrm{buy}}$, and $\phi_2$, we merely need to
  check continuity of $\phi_{\mathrm{buy}}$ along the boundaries of
  the partition of $[0,+\infty)^2$ described by $s_1, s_2,
  s_3$. First, observe in~\eqref{def:phibuy:eq3b} with
  $\zeta = s_1(0,\tau)$ that
  $\tau_{\mathrm{wait}}(\tau,s_1(0,\tau)) = \tau$ due to
  Lemma~\ref{lem:tauwait} and thus
  $\phi_{\mathrm{sell}}(\tau-\tau,s_1(0,\tau) e^{-\kappa \tau}) =
  \phi_2(\tau,s_1(0,\tau))$ holds true by continuity of
  $\bar\varphi(\tau)$ in~\eqref{def:barphi} as argued in
  Lemma~\ref{lem:fbcurve} 4.). In~\eqref{def:phibuy:eq3c}, if
  $\zeta = s_3(\tau)$, we have $\phi_2(\tau,s_3(\tau)) = 0$ by
  definition of $\phi_2$ in~\eqref{def:phi2} and $s_3$
  in~\eqref{def:s3}.  Next, let
  $\zeta =
  \hat{\zeta}^{\mathrm{buy}}(\tau,\bar\zeta(0),\bar\varphi(0),0)$
  in~\eqref{def:phibuy:eq2}. Since
  $\tau_{\mathrm{buy}}(\tau,
  \hat{\zeta}^{\mathrm{buy}}(\tau,\bar\zeta(0),\bar\varphi(0),0)) =
  \tau$ due to Lemma~\ref{lem:taubuy}, a simple computation shows that
  $\hat{\varphi}^{\mathrm{buy}} \left( \tau, \bar\zeta(0),
    \bar\varphi(0), 0 \right) = \phi_{\mathrm{sell}}( \tau,
  -\hat{\zeta}^{\mathrm{buy}}(\tau, \bar\zeta(0), \bar\varphi(0),
  0))$, cf.~\eqref{def:phibuy:eq1}.  For $\zeta = \bar\zeta(\tau)$ and
  $\tau \geq \bar\theta$ we have
  $\tau_{\mathrm{wait}}(\tau,\bar\zeta(\tau)) = \bar\theta$ and
  $\tau_{\mathrm{buy}}(\tau, \bar\zeta(\tau)) = 0$ by virtue of
  Lemmas~\ref{lem:tauwait} and~\ref{lem:taubuy}. Consequently,
  in~\eqref{def:phibuy:eq3a} it holds that
  $\phi_{\mathrm{sell}}(\tau-\bar\theta,\bar\zeta(\tau)
  e^{-\kappa\bar\theta}) = \bar\varphi(\tau) =
  \hat{\varphi}^{\mathrm{buy}} \left( \tau, \bar\zeta(\tau),
    \bar\varphi(\tau), \tau \right)$ by definition of $\bar\varphi$
  in~\eqref{def:barphi} and property~\eqref{eq:propbuy}. Similarly,
  for $\zeta = \bar\zeta(\tau)$ and $\tau < \bar\theta$ we obtain once
  more due to $\tau_{\mathrm{buy}}(\tau, \bar\zeta(\tau)) = 0$ the
  identities
  $\phi_2(\tau,\bar\zeta(\tau)) = \bar\varphi(\tau) =
  \hat{\varphi}^{\mathrm{buy}}(\tau,\bar\zeta(\tau),\bar\varphi(\tau),\tau)$
  and
  $0 = \bar\varphi(\tau) =
  \hat{\varphi}^{\mathrm{buy}}(\tau,\bar\zeta(\tau),\bar\varphi(\tau),\tau)$
  (again by definition in~\eqref{def:barphi} and
  property~\eqref{eq:propbuy}). In particular, note that
  $\phi_{\mathrm{buy}}(\tau,\bar\zeta(\tau)) = \bar\varphi(\tau)$ for
  all $\tau \geq 0$ in~\eqref{def:phibuy:eq2}.
\end{proof}


\subsubsection{Proof of Theorem~\ref{thm:main} and
  Corollaries~\ref{cor:main1} and~\ref{cor:main2}}
\label{subsubsec:proofs:main}

We are now ready to prove our main Theorem~\ref{thm:main} together
with Corollaries~\ref{cor:main1} and~\ref{cor:main2}. The outline of
our reasoning is as follows: First, we show that
\begin{align}
  \{ (\tau,\zeta,\varphi)
  \in \cS : \varphi = \phi_{\mathrm{sell}}(\tau,\zeta) \} 
  & \subset \partial \mathcal{R}_{\mathrm{sell}}, \label{pf:sellbound} \\
  \{ (\tau,\zeta,\varphi) \in
  \cS : \varphi > \phi_{\mathrm{sell}}(\tau,\zeta) \} 
  & \subset \mathcal{R}_{\mathrm{sell}}, \label{pf:sellreg} \\
  \{ (\tau,\zeta,\varphi)
  \in \cS : \phi_{\mathrm{buy}}(\tau,\zeta) = \varphi \} 
  & \subset \partial \mathcal{R}_{\mathrm{buy}}, \label{pf:buybound} \\
  \{ (\tau,\zeta,\varphi) \in
  \cS : \phi_{\mathrm{buy}}(\tau,\zeta) >
  \varphi\} 
  & \subset \mathcal{R}_{\mathrm{buy}} \label{pf:buyreg}
\end{align}
hold true. Then we prove the inequality in~\eqref{eq:phisellphibuy},
i.e.,
$\phi_{\mathrm{sell}}(\tau,\zeta) > \phi_{\mathrm{buy}}(\tau,\zeta)$
on $[0,+\infty)^2$ and argue that
\begin{equation}
  \left\{ (\tau,\zeta,\varphi) \in \cS : \phi_{\mathrm{buy}}(\tau,\zeta) <
    \varphi < \phi_{\mathrm{sell}}(\tau,\zeta) \right\} \subset
  \mathcal{R}_{\mathrm{wait}}. \label{pf:waitreg}
\end{equation}
In fact, since for all $(\tau,\zeta) \in [0,+\infty)^2$ the two
surfaces $(\tau,\zeta, \phi_{\mathrm{buy}}(\tau,\zeta))$ and
$(\tau,\zeta, \phi_{\mathrm{sell}}(\tau,\zeta))$ separate the state
space $\cS$ into three disjoint regions, we can then
readily deduce that equality must hold in all relations
from~\eqref{pf:sellbound} to~\eqref{pf:waitreg} and that
$\partial \mathcal{R}_{\mathrm{wait}} = \partial
\mathcal{R}_{\mathrm{buy}} \cup \partial \mathcal{R}_{\mathrm{sell}}$
as claimed in Theorem~\ref{thm:main}.

\emph{Step 1:} We start with the boundary of the selling region
$\partial \mathcal{R}_{\mathrm{sell}}$ and the claim
in~\eqref{pf:sellbound}. Showing that this relation holds true comes
along with the verification of the claims in Corollary~\ref{cor:main1}
which describe the corresponding optimal strategies for triplets in
$\partial \mathcal{R}_{\mathrm{sell}}$. Therefore, let
$(\tau,\zeta,\varphi) \in \cS$ such that
$\varphi = \phi_{\mathrm{sell}}(\tau,\zeta)$ with
$\phi_{\mathrm{sell}}$ as introduced in~\eqref{def:phisell}. We have
to argue that $(\tau,\zeta,\phi_{\mathrm{sell}}(\tau,\zeta))$ belongs
to $\partial \mathcal{R}_{\mathrm{sell}}$ as defined
in~\eqref{def:sellbound}. To justify this, we claim that the
corresponding optimal strategy
$\hat{X}^{\tau,\zeta,\varphi}=(\hat{X}^{\tau,\zeta,\varphi,\uparrow},
\hat{X}^{\tau,\zeta,\varphi,\downarrow}) \in \cX^d$ associated to the
problem data
$(\tau,\zeta,\varphi)=(\tau,\zeta, \phi_{\mathrm{sell}}(\tau,\zeta))$
is given by
\begin{equation} \label{def:optsellstrat1}
  \hat{X}_t^{\tau,\zeta,\varphi,\uparrow} \equiv 0, \quad
  \hat{X}_t^{\tau,\zeta,\varphi,\downarrow} = \varphi -
  \hat{\varphi}^{\mathrm{sell}}(\tau,\zeta,\varphi,t) \quad (0 \leq t
  \leq \tau)
\end{equation}
with $\hat{\varphi}^{\mathrm{sell}}$ as defined
in~\eqref{def:phihatsell}. First, observe
that~\eqref{def:optsellstrat1} immediately yields
$\hat{X}_0^{\tau,\zeta,\varphi,\downarrow} =
\phi_{\mathrm{sell}}(\tau,\zeta) -
\hat{\varphi}^{\mathrm{sell}}(\tau,\zeta,
\phi_{\mathrm{sell}}(\tau,\zeta),0) = 0$ due
to~\eqref{eq:propsell}. Moreover, it follows from
Lemma~\ref{lem:hatfunctions} 4.) that
$\hat{X}^{\tau,\zeta,\varphi,\downarrow}$ in~\eqref{def:optsellstrat1}
is strictly increasing and thus
$\{ d\hat{X}^{\tau,\zeta,\varphi,\downarrow}>0\}=[0,\tau]$. Obviously,
the corresponding share holdings of strategy
$\hat{X}^{\tau,\zeta,\varphi}$ are given by
\begin{equation} 
  \varphi_t^{\hat{X}^{\tau,\zeta,\varphi}}  
  = \hat{\varphi}^{\mathrm{sell}}(\tau,\zeta,\varphi,t) \quad (0 \leq t
  \leq \tau). \label{def:optsellphi}
\end{equation}
Inserting~\eqref{def:optsellstrat1} into the spread dynamics
in~\eqref{eq:spreadsolution} yields, after some elementary
computations, the representation
\begin{equation}
  \zeta_t^{\hat{X}^{\tau,\zeta,\varphi}}  
  = \hat{\zeta}^{\mathrm{sell}}(\tau,\zeta,\varphi,t) \quad (0 \leq t
  \leq \tau) \label{def:optsellzeta}
\end{equation}
with $\hat{\zeta}^{\mathrm{sell}}$ as defined
in~\eqref{def:zetahatsell}. In particular, the identities
in~\eqref{eq:propsell} imply
$\varphi_0^{\hat{X}^{\tau,\zeta,\varphi}} = \varphi =
\phi_{\mathrm{sell}}(\tau,\zeta)$ and
$\zeta_0^{\hat{X}^{\tau,\zeta,\varphi}}=\zeta$ as desired. Given the
explicit expression of the share holdings in~\eqref{def:optsellphi},
it can be easily checked that the second order ODE in~\eqref{eq:ODE1}
with initial conditions~\eqref{eq:ODEinitcond1} is
satisfied. Moreover, using the representation of the corresponding
controlled spread dynamics in~\eqref{def:optsellzeta}, a
straightforward computation shows that the desired relation
in~\eqref{eq:slide1} also holds true. As a consequence, appealing to
Lemma~\ref{lem:phisell}, we can deduce that the final position in the
risky asset is strictly positiv, i.e.,
$\varphi_{\tau}^{\hat{X}^{\tau,\zeta,\varphi}} =
\phi_{\mathrm{sell}}(0,\zeta_{\tau}^{\hat{X}^{\tau,\zeta,\varphi}}) >
0$. Concerning the claimed optimality of the strategy
$\hat{X}^{\tau,\zeta,\varphi}=(\hat{X}^{\tau,\zeta,\varphi,\uparrow},
\hat{X}^{\tau,\zeta,\varphi,\downarrow})$ in~\eqref{def:optsellstrat1}
a simple but tedious computation (which we omit for the sake of brevity) yields that
$\hat{X}^{\tau,\zeta,\varphi}$ satisfies
$\nabla^{\downarrow}_t J_{\tau} (\hat{X}^{\tau,\zeta,\varphi}) = 0$
for all $t \in [0,\tau]$. 
Note that the subgradient does not depend on $\varrho$
here. Consequently, by virtue of the first order optimality conditions
in Proposition~\ref{prop:foc} together with Lemma~\ref{lem:nobuysell},
we obtain that $\hat{X}^{\tau,\zeta,\varphi}$
in~\eqref{def:optsellstrat1} is optimal. In particular, since
$\nabla^{\downarrow}_0 J_{\tau} (\hat{X}^{\tau,\zeta,\varphi})=0$ and
$\hat{X}_0^{\tau,\zeta,\varphi,\downarrow}=0$, we can conclude that
$(\tau,\zeta,\varphi) = (\tau,\zeta,\phi_{\mathrm{sell}}(\tau,\zeta))$
belongs to $\partial \mathcal{R}_{\mathrm{sell}}$ as defined
in~\eqref{def:sellbound} with Corollary~\ref{cor:main1} holding true
for these triplets.

\emph{Step 2:} Let us continue with the claim in~\eqref{pf:sellreg}
concerning the selling-region $\mathcal{R}_{\mathrm{sell}}$. We argue
that for any $(\tau,\zeta,\varphi) \in \cS$ with
$\varphi > \phi_{\mathrm{sell}}(\tau,\zeta)$ the corresponding optimal
strategy
$\hat{X}^{\tau,\zeta,\varphi}=(\hat{X}^{\tau,\zeta,\varphi,\uparrow},\hat{X}^{\tau,\zeta,\varphi,\downarrow})
\in \cX^d$ is given by
\begin{equation} \label{def:optsellstrat2}
  \hat{X}_t^{\tau,\zeta,\varphi,\uparrow} \equiv 0, \quad
  \hat{X}_t^{\tau,\zeta,\varphi,\downarrow} = x^\downarrow +
 \hat{X}_t^{\tau,\zeta+\eta x^{\downarrow},\varphi-x^{\downarrow},\downarrow} \quad (0 \leq t
  \leq \tau),
\end{equation}
where $x^{\downarrow}$ is defined as
\begin{equation} \label{def:xdown} x^\downarrow \set \frac{\varphi -
    \phi_{\mathrm{sell}}(\tau,\zeta)}{1 + \eta \frac{\kappa}{\lambda}
    C(\tau)} = \frac{\varphi - \frac{\mu}{\lambda^2} D(\tau) - \zeta
    \frac{\kappa}{\lambda} C(\tau)}{1 + \eta \frac{\kappa}{\lambda}
    C(\tau)}> 0.
\end{equation}
Indeed, note that~\eqref{def:xdown} implies
$\varphi - x^{\downarrow} = \phi_{\mathrm{sell}}(\tau,\zeta + \eta
x^{\downarrow})$ and thus we have
$(\tau,\zeta + \eta x^{\downarrow}, \varphi - x^{\downarrow}) \in
\partial \mathcal{R}_{\mathrm{sell}}$ due to Step 1 with corresponding
optimal strategy
$\hat{X}^{\tau,\zeta+\eta x^{\downarrow},\varphi-x^{\downarrow}}$ as
described in~\eqref{def:optsellstrat1} above. Recall that this implies
$\hat{X}_0^{\tau,\zeta+\eta
  x^\downarrow,\varphi-x^{\downarrow},\downarrow}=0$. Hence, by
construction in~\eqref{def:optsellstrat2}, it holds that
$\{ d\hat{X}^{\tau,\zeta,\varphi,\downarrow}>0\}=[0,\tau]$. Moreover,
appealing to the definition of the subgradients
in~\eqref{eq:buysubgradient} and~\eqref{eq:sellsubgradient}, we have
\begin{equation*}
  {}^\varrho\nabla_{t}^{\uparrow,\downarrow}
  J_\tau(\hat{X}^{\tau,\zeta,\varphi}) =
  {}^\varrho\nabla_{t}^{\uparrow,\downarrow}
  J_\tau(\hat{X}^{\tau,\zeta+\eta
    x^\downarrow,\varphi-x^{\downarrow}}) \quad (0 \leq t \leq \tau)
\end{equation*}
because
$\varphi^{\hat{X}^{\tau,\zeta,\varphi}}_t =
\varphi^{\hat{X}^{\tau,\zeta+\eta
    x^\downarrow,\varphi-x^\downarrow}}_t$ and
$\zeta^{\hat{X}^{\tau,\zeta,\varphi}}_t =
\zeta^{\hat{X}^{\tau,\zeta+\eta x^\downarrow,\varphi-x^\downarrow}}_t$
for all $t \in [0,\tau]$. But this allows us to deduce that the
strategy in~\eqref{def:optsellstrat2} is optimal by virtue of the
first order optimality conditions in Proposition~\ref{prop:foc} and
the fact that these are satisfied by the strategy
$\hat{X}^{\tau,\zeta+\eta x^\downarrow,\varphi-x^{\downarrow}}$ as
shown in Step~1. Specifically, we have
$\nabla_{t}^{\uparrow} J_\tau(\hat{X}^{\tau,\zeta,\varphi}) > 0$ and
$\nabla_{t}^{\downarrow} J_\tau(\hat{X}^{\tau,\zeta,\varphi}) = 0$ for
all $t \in [0,\tau]$ (observe that the subgradients do not depend on
$\varrho$ here as in Step~1). Together with
$\hat{X}_0^{\tau,\zeta,\varphi,\downarrow} = x^\downarrow > 0$
in~\eqref{def:optsellstrat2} we obtain that $(\tau,\zeta,\varphi)$
belongs to $\mathcal{R}_{\mathrm{sell}}$ as defined
in~\eqref{def:sellregion}.

\emph{Step 3:} Now, we address the boundary of the buying region
$\partial \mathcal{R}_{\mathrm{buy}}$ and the claim
in~\eqref{pf:buybound}.  Therefore, let $(\tau,\zeta,\varphi) \in \cS$
be such that $\varphi = \phi_{\mathrm{buy}}(\tau,\zeta)$ holds true
with $\phi_{\mathrm{buy}}$ as introduced in~\eqref{def:phibuy:eq1}
to~\eqref{def:phibuy:eq3c}. Since the definition of
$\phi_{\mathrm{buy}}$ rests upon a partition of the domain
$[0,+\infty)^2$, we have to consider each of these cases separately;
cf. also Figure~\ref{fig:domain}. We will verify this together with
the claims in Corollary~\ref{cor:main2} 1.), 2.), and 3.),
respectively, which describe the corresponding optimal strategies.

\underline{Case 1 (part I in Fig.~\ref{fig:domain}):} First, let
$\zeta \geq
\hat{\zeta}^{\mathrm{buy}}(\tau,\bar\zeta(0),\bar\varphi(0),0)$. In
this case, we have
$\varphi = \phi_{\mathrm{buy}}(\tau,\zeta) =
\phi_{\mathrm{sell}}(\tau,-\zeta)$ in view
of~\eqref{def:phibuy:eq1}. In order to show that
$(\tau,\zeta, \phi_{\mathrm{buy}}(\tau,\zeta))$ belongs to
$\partial \mathcal{R}_{\mathrm{buy}}$ as defined
in~\eqref{def:buybound}, we claim that the corresponding optimal
strategy
$\hat{X}^{\tau,\zeta,\varphi}=(\hat{X}^{\tau,\zeta,\varphi,\uparrow},
\hat{X}^{\tau,\zeta,\varphi,\downarrow}) \in \cX^d$ is given by
\begin{equation} \label{def:optbuystrat1}
  \hat{X}_t^{\tau,\zeta,\varphi,\uparrow} =
  \hat{\varphi}^{\mathrm{sell}}(\tau,-\zeta,\varphi,t) - \varphi,
  \quad \hat{X}_t^{\tau,\zeta,\varphi,\downarrow} \equiv 0 \quad (0
  \leq t \leq \tau),
\end{equation}
with associated share holdings and spread dynamics
$ \varphi_t^{\hat{X}^{\tau,\zeta,\varphi}} =
\hat{\varphi}^{\mathrm{sell}}(\tau,-\zeta,\varphi,t)$ and
$\zeta_t^{\hat{X}^{\tau,\zeta,\varphi}} = -
\hat{\zeta}^{\mathrm{sell}}(\tau,-\zeta,\varphi,t)$, respectively, for
all $t \in [0,\tau]$. In fact, very similar computations as in Step 1
above allow us to verify that the strategy $\hat{X}^{\tau,\zeta,\varphi}$
in~\eqref{def:optbuystrat1} is optimal and that all assertions stated
in Corollary~\ref{cor:main2} 1.) hold true for the triplet
$(\tau,\zeta,\phi_{\mathrm{sell}}(\tau,-\zeta))$. As in Step~1, an
elementary but lengthy computation reveals that
$\nabla^{\uparrow}_t J_{\tau} (\hat{X}^{\tau,\zeta,\varphi}) \equiv 0$
for all $t \in [0,\tau]$. Note that the subgradient does not depend on
$\varrho$ here because
$\varphi_{\tau}^{\hat{X}^{\tau,\zeta,\varphi}} =
\phi_{\mathrm{sell}}(0,-\zeta_{\tau}^{\hat{X}^{\tau,\zeta,\varphi}}) <
0$ based on~\eqref{eq:slide2} and the fact that
$\zeta_{\tau}^{\hat{X}^{\tau,\zeta,\varphi}} > 2\mu/\kappa$. Since
$\hat{X}_0^{\tau,\zeta,\varphi,\uparrow} = 0$
in~\eqref{def:optbuystrat1} due to~\eqref{eq:propsell}, we can
conclude that
$(\tau,\zeta,\phi_{\mathrm{sell}}(\tau,-\zeta))$ belongs to
$\partial \mathcal{R}_{\mathrm{buy}}$ as defined
in~\eqref{def:buybound}.

\underline{Case 2:} Next, let us consider the case $\bar\zeta(\tau) <
\zeta <
\hat{\zeta}^{\mathrm{buy}}(\tau,\bar\zeta(0),\bar\varphi(0),0)$ and
let $\tau_{\mathrm{buy}}(\tau,\zeta) \in
(0,\tau)$ as defined in Lemma~\ref{lem:taubuy},
equation~\eqref{def:taubuy}. To ease notation, we set
\begin{equation} \label{def:starstriplet} 
  \tau^* \set \tau
  -\tau_{\mathrm{buy}}(\tau,\zeta), \quad \zeta^* \set
  \bar\zeta(\tau^*), \quad \varphi^* \set \bar\varphi(\tau^*).
\end{equation}
In view of the definition of
$\phi_{\mathrm{buy}}$ in~\eqref{def:phibuy:eq2} we thus have
\begin{equation} \label{eq:proofphibuy} 
  \varphi =
  \phi_{\mathrm{buy}}(\tau,\zeta) = \hat{\varphi}^{\mathrm{buy}}
  (\tau, \zeta^*, \varphi^*, \tau^*).
\end{equation}
To show that $(\tau,\zeta,
\phi_{\mathrm{buy}}(\tau,\zeta))$ belongs to $\partial
\mathcal{R}_{\mathrm{buy}}$ as defined in~\eqref{def:buybound}, we
will explicitly state the corresponding optimal strategy
$\hat{X}^{\tau,\zeta,\varphi}=(\hat{X}^{\tau,\zeta,\varphi,\uparrow},
\hat{X}^{\tau,\zeta,\varphi,\downarrow})$ in
$\cX^d$. This will be carried out by distinguishing further sub-cases
with respect to the initial data $\tau$ and
$\zeta$ (cf.  Figure~\ref{fig:domain}).

\underline{Case 2.1 (part II.1 in Fig.~\ref{fig:domain}):} If
$\tau > \bar\theta$ and
$s_1(\tau-\bar\theta,\bar\theta) = \bar\zeta(\tau) < \zeta <
\hat{\zeta}^{\mathrm{buy}}(\tau,\bar\zeta(\bar\theta),
\bar\varphi(\bar\theta),\bar\theta)$, it follows from
Lemma~\ref{lem:hatfunctions} 1.) and 2.) that
$\tau_{\mathrm{buy}}(\tau,\zeta) < \tau - \bar\theta$ and thus
$\tau^* > \bar\theta$. This implies $\varphi^*>0$ due to
Lemma~\ref{lem:fbcurve} 4.). We claim that the corresponding optimal
strategy is given as follows: The cumulative purchases of the risky
asset are
\begin{equation} \label{def:optbuysellstratbuy}
  \hat{X}_t^{\tau,\zeta,\varphi,\uparrow} =
  \begin{cases}
    \hat\varphi^{\mathrm{buy}}(\tau-t,\zeta^*,\varphi^*,\tau^*) -
    \varphi
    & \text{if } 0 \leq t \leq \tau_{\mathrm{buy}}(\tau,\zeta), \\
    \varphi^* - \varphi & \text{if } \tau_{\mathrm{buy}}(\tau,\zeta) <
    t \leq \tau,
  \end{cases}
\end{equation}
with $\hat{\varphi}^{\mathrm{buy}}$ as defined
in~\eqref{def:phihatbuy}. Observe that
$\hat{X}_0^{\tau,\zeta,\varphi,\uparrow}=0$ due to
assumption~\eqref{eq:proofphibuy} as well as
$\{
d\hat{X}^{\tau,\zeta,\varphi,\uparrow}>0\}=[0,\tau_{\mathrm{buy}}(\tau,\zeta))$
by virtue of Lemma~\ref{lem:hatfunctions} 3.). In particular,
$\varphi^* > \varphi = \hat{\varphi}^{\mathrm{buy}} (\tau, \zeta^*,
\varphi^*, \tau^*)$. The cumulative sells of the risky asset are
\begin{equation} \label{def:optbuysellstratsell}
  \hat{X}_t^{\tau,\zeta,\varphi,\downarrow} =
  \begin{cases}
    0 & \text{if } 0 \leq t < \tau_{\mathrm{buy}}(\tau,\zeta) + \bar\theta, \\
    \hat{X}_{t-\tau_{\mathrm{buy}}(\tau,\zeta) -
      \bar\theta}^{\tau^*-\bar\theta,\zeta^*
      e^{-\kappa\bar\theta},\varphi^*,\downarrow} & \text{if }
    \tau_{\mathrm{buy}}(\tau,\zeta) + \bar\theta \leq t \leq \tau.
  \end{cases}
\end{equation}
Notice that
\begin{equation} \label{prop:optbuysellstrat}
  (\tau^*-\bar\theta,\zeta^* e^{-\kappa\bar\theta},\varphi^*) \in
  \partial \mathcal{R}_{\mathrm{sell}}
\end{equation}
due to Step 1 because
$\varphi^* = \bar\varphi(\tau^*)
=\phi_{\textrm{sell}}(\tau^*-\bar\theta,\zeta^*
e^{-\kappa\bar\theta})$ by the definition of $\bar\varphi$
in~\eqref{def:barphi} and the fact that $\tau^* > \bar\theta$. In
other words,
$\hat{X}_\cdot^{\tau^*-\bar\theta,\zeta^*
  e^{-\kappa\bar\theta},\varphi^*,\downarrow} = \varphi^* -
\hat{\varphi}^{\mathrm{sell}}(\tau^*-\bar\theta, \zeta^*
e^{-\kappa\bar\theta},\varphi^*,\cdot)$ denotes the optimal cumulative
sells on $[0,\tau^* - \bar\theta]$ as given
in~\eqref{def:optsellstrat1} in Step 1 for the triplet
$(\tau^*-\bar\theta,\zeta^* e^{-\kappa\bar\theta},\varphi^*)
\in \partial \mathcal{R}_{\mathrm{sell}}$. In particular, it holds
that
$\{ d\hat{X}^{\tau,\zeta,\varphi,\downarrow}>0\}
=[\tau_{\mathrm{buy}}(\tau,\zeta)+\bar\theta,\tau)$.  The associated
share holdings and spread dynamics of strategy
$\hat{X}^{\tau,\zeta,\varphi}=(\hat{X}^{\tau,\zeta,\varphi,\uparrow},
\hat{X}^{\tau,\zeta,\varphi,\downarrow})$ prescribed
in~\eqref{def:optbuysellstratbuy} and~\eqref{def:optbuysellstratsell}
can be easily computed and are given by
\begin{equation} \label{def:optbuysellphi}
  \varphi_t^{\hat{X}^{\tau,\zeta,\varphi}} =
  \begin{cases}
    \hat\varphi^{\mathrm{buy}}(\tau-t,\zeta^*,\varphi^*,\tau^*),
    & 0 \leq t \leq \tau_{\mathrm{buy}}(\tau,\zeta), \\
    \varphi^*, & \tau_{\mathrm{buy}}(\tau,\zeta) < t \leq
    \tau_{\mathrm{buy}}(\tau,\zeta) + \bar\theta, \\
    \hat{\varphi}^{\mathrm{sell}}\big(\tau^*-\bar\theta,\zeta^*
    e^{-\kappa\bar\theta},\varphi^*, & \\
    \hspace{27pt} t-\tau_{\mathrm{buy}}(\tau,\zeta) - \bar\theta
    \big), & \tau_{\mathrm{buy}}(\tau,\zeta) + \bar\theta < t \leq
    \tau,
  \end{cases}
\end{equation}
and
\begin{equation} \label{def:optbuysellzeta}
  \zeta_t^{\hat{X}^{\tau,\zeta,\varphi}} =
  \begin{cases}
    \hat\zeta^{\mathrm{buy}}(\tau-t,\zeta^*,\varphi^*,\tau^*),
    & 0 \leq t \leq \tau_{\mathrm{buy}}(\tau,\zeta), \\
    \zeta^* e^{-\kappa(t-\tau_{\mathrm{buy}}(\tau,\zeta))}, &
    \tau_{\mathrm{buy}}(\tau,\zeta) < t \leq
    \tau_{\mathrm{buy}}(\tau,\zeta) + \bar\theta, \\
    \hat{\zeta}^{\mathrm{sell}}(\tau^*-\bar\theta,\zeta^*
    e^{-\kappa\bar\theta},\varphi^*, & \\
    \hspace{27pt} t-\tau_{\mathrm{buy}}(\tau,\zeta) - \bar\theta ), &
    \tau_{\mathrm{buy}}(\tau,\zeta) + \bar\theta < t \leq \tau.
  \end{cases}
\end{equation}
Observe that
$\varphi_{\tau_{\mathrm{buy}}(\tau,\zeta)}^{\hat{X}^{\tau,\zeta,\varphi}}
= \varphi^*$,
$\zeta_{\tau_{\mathrm{buy}}(\tau,\zeta)}^{\hat{X}^{\tau,\zeta,\varphi}}
= \zeta^*$,
$\varphi_{\tau_{\mathrm{buy}}(\tau,\zeta)+\bar\theta}^{\hat{X}^{\tau,\zeta,\varphi}}
= \varphi^*$, and
$\zeta_{\tau_{\mathrm{buy}}(\tau,\zeta)
  +\bar\theta}^{\hat{X}^{\tau,\zeta,\varphi}} = \zeta^*
e^{-\kappa\bar\theta}$ by virtue
of~\eqref{eq:propbuy},~\eqref{eq:propsell}. Hence,
recalling~\eqref{prop:optbuysellstrat}, it holds that
\begin{equation} \label{optbuysell:sellbound} 
  \left( \tau-
    \tau_{\mathrm{buy}}(\tau,\zeta) - \bar\theta,
    \zeta_{\tau_{\mathrm{buy}}(\tau,\zeta)
      +\bar\theta}^{\hat{X}^{\tau,\zeta,\varphi}} ,
    \varphi_{\tau_{\mathrm{buy}}(\tau,\zeta)+\bar\theta}^{\hat{X}^{\tau,\zeta,\varphi}}
  \right) \in
  \partial \mathcal{R}_{\mathrm{sell}}.
\end{equation}
In other words, referring to~\eqref{def:tausell}
and~\eqref{eq:waitandsell} in Corollary~\ref{cor:main2}, we have
$\tau_{\mathrm{sell}}(\tau,\zeta) =
\tau-\tau_{\mathrm{buy}}(\tau,\zeta)-\bar\theta = \tau^* - \bar\theta
> 0$ with $\tau_{\mathrm{wait}}(\tau,\zeta) = \bar\theta$ (see also
the definition in~\eqref{set:tauwait}). Next, it can be easily checked
that the second order ODE in~\eqref{eq:ODE2} with desired terminal
conditions~\eqref{eq:ODEfinalcond} is satisfied by
$\varphi^{\hat{X}^{\tau,\zeta,\varphi}}$ on
$(0, \tau_{\mathrm{buy}}(\tau,\zeta))$ as stated
in~\eqref{def:optbuysellphi}. Moreover, the relation
in~\eqref{eq:slide3} also holds true. Indeed, for all
$t \in [0,\tau^{\mathrm{buy}}(\tau,\zeta)]$ it holds that
$\zeta_t^{\hat{X}^{\tau,\zeta,\varphi}}
=\hat\zeta^{\mathrm{buy}}(\tau-t,\zeta^*,\varphi^*,\tau^*) \in
[\bar\zeta(\tau-t),
\hat\zeta^{\mathrm{buy}}(\tau-t,\bar\zeta(0),\bar\varphi(0),0)]$ as
well as
$\tau^{\mathrm{buy}}(\tau-t,
\hat\zeta^{\mathrm{buy}}(\tau-t,\zeta^*,\varphi^*,\tau^*)) =
\tau-t-\tau^*$ due to Lemma~\ref{lem:hatfunctions} 1.)
and~\eqref{eq:taubuy:property}, respectively. Thus, by the definition
of $\phi_{\mathrm{buy}}$ in~\eqref{def:phibuy:eq2} we obtain
$\phi_{\mathrm{buy}}(\tau-t, \zeta_t^{\hat{X}^{\tau,\zeta,\varphi}}) =
\phi_{\mathrm{buy}}(\tau-t,
\hat\zeta^{\mathrm{buy}}(\tau-t,\zeta^*,\varphi^*,\tau^*))
=\hat{\varphi}^{\mathrm{buy}}(\tau-t,\bar\zeta(\tau^*),\bar\varphi(\tau^*),\tau^*)
= \varphi_t^{\hat{X}^{\tau,\zeta,\varphi}}$ for all
$t \in (0, \tau_{\mathrm{buy}}(\tau,\zeta))$ as desired. It is left to
argue that the strategy $\hat{X}^{\tau,\zeta,\varphi}$ specified
in~\eqref{def:optbuysellstratbuy} and~\eqref{def:optbuysellstratsell}
satisfies the first order optimality conditions in
Proposition~\ref{prop:foc} and is thus optimal. Due to the dynamic
programming principle from Remark~\ref{rem:dynamicprog} this can be
done via a backward reasoning in time. First of all, optimality of the
strategy on the time interval
$[\tau_{\mathrm{buy}}(\tau,\zeta)+\bar\theta,\tau]$ follows by
construction of $\hat{X}^{\tau,\zeta,\varphi}$
from~\eqref{optbuysell:sellbound} and Step 1. Next, we have to check
the sell- and buy-subgradients on
$[\tau_{\mathrm{buy}}(\tau,\zeta),\tau_{\mathrm{buy}}(\tau,\zeta)+\bar\theta]$. Observe
that, again by construction of $\hat{X}^{\tau,\zeta,\varphi}$ on this
interval and due to the fact that
\begin{equation} \label{opt:buysell:sellgradient}
  \nabla_{\tau_{\mathrm{buy}}(\tau,\zeta)+\bar\theta}^{\downarrow}J_\tau(
  \hat{X}^{\tau,\zeta,\varphi}) =
  \nabla_{0}^{\downarrow}J_{\tau^*-\bar\theta} (
  \hat{X}^{\tau^*-\bar\theta,\zeta^*e^{-\kappa\bar\theta},\varphi^*} )
  = 0,
\end{equation}
we obtain with Lemma~\ref{lem:waitgradients} 1.) for all
$t \in
[\tau_{\mathrm{buy}}(\tau,\zeta),\tau_{\mathrm{buy}}(\tau,\zeta)+\bar\theta)$
the expressions
\begin{equation} \label{eq:mastergrad1}
  \begin{aligned}
    & \nabla_{t}^{\uparrow,\downarrow}
    J_\tau(\hat{X}^{\tau,\zeta,\varphi}) =
    \nabla_{t-\tau_{\mathrm{buy}}(\tau,\zeta)}^{\uparrow,\downarrow}
    J_{\tau^*}(\hat{X}^{\tau^*,\zeta^*,\varphi^*}) \\
    & =
    \nabla_{0}^{\uparrow,\downarrow}J_{\tau-t}(\hat{X}^{\tau-t,\zeta^*e^{-\kappa
        (t-\tau_{\mathrm{buy}}(\tau,\zeta))},\varphi^*})
    \\
    & = g^{\uparrow,\downarrow}
    (\tau_{\mathrm{buy}}(\tau,\zeta)+\bar\theta-t;
    \tau^*-\bar\theta,\zeta^* e^{-\kappa\bar\theta},\varphi^*) \\
    & = \pm(\lambda^2 \varphi^* - \mu)
    (\tau_{\mathrm{buy}}(\tau,\zeta)+\bar\theta-t) + \frac{1}{2}
    \zeta^* e^{-\kappa\tau^*}
    (e^{\kappa(\tau_{\mathrm{buy}}(\tau,\zeta)+\bar\theta-t)}\pm 1) \\
    & \hspace{13.5pt} + \frac{1}{\kappa}
    (e^{-\kappa(\tau_{\mathrm{buy}}(\tau,\zeta)+\bar\theta-t)} \pm 1)
    ( -\lambda^2\varphi^* + \mu + \frac{1}{2} \kappa \zeta^*
    e^{-\kappa\bar\theta}),
  \end{aligned}
\end{equation}
where we used the fact that
\begin{equation*} 
  \int_{[0,\tau^*-\bar\theta]} e^{-\kappa
    u} \, d\hat{X}^{\tau^*-\bar\theta,\zeta^* e^{-\kappa\bar\theta},\varphi^*,\downarrow}_u =
  \frac{2}{\eta\kappa} ( -\lambda^2\varphi^* + \mu +
  \frac{1}{2} \kappa \zeta^* e^{-\kappa\bar\theta}).
\end{equation*}
Notice that~\eqref{opt:buysell:sellgradient} implies
$g^{\downarrow}( 0 ; \tau^*-\bar\theta,\zeta^*
e^{-\kappa\bar\theta},\varphi^*)=0$. Moreover, it can be easily
checked that
$\frac{\partial}{\partial \theta} g^{\downarrow}(\theta ;
\tau^*-\bar\theta,\zeta^* e^{-\kappa\bar\theta},\varphi^*)
\vert_{\theta = 0} = 0$. Consequently, due to strict convexity of
$s\mapsto g^{\downarrow}(\bar\theta - s; \tau^*-\bar\theta,\zeta^*
e^{-\kappa\bar\theta},\varphi^*)$ on $[0,\bar\theta]$, we can deduce
that $\nabla_{t}^{\downarrow}J_\tau(\hat{X}^{\tau,\zeta,\varphi}) > 0$
on
$[\tau_{\mathrm{buy}}(\tau,\zeta),
\tau_{\mathrm{buy}}(\tau,\zeta)+\bar\theta)$. Similarly, concerning
the buy-subgradient, \eqref{opt:buysell:sellgradient} implies
$g^{\uparrow}(0; \tau^*-\bar\theta,\zeta^*
e^{-\kappa\bar\theta},\varphi^*) > 0$ due to
Lemma~\ref{lem:nobuysell}. In addition, simple algebraic manipulations
show that the identity $\varphi^* = \bar\varphi(\tau^*)$ (by using the
representation in~\eqref{eq:barphi1}) actually implies that
$g^{\uparrow}(\bar\theta; \tau^*-\bar\theta,\zeta^*
e^{-\kappa\bar\theta},$ $\varphi^*)=0$ and
$\frac{\partial}{\partial \theta} g^{\uparrow}( \theta ;
\tau^*-\bar\theta,\zeta^* e^{-\kappa\bar\theta},\varphi^*)
\vert_{\theta = \bar\theta} = 0$.  Hence, utilizing the fact that
$s \mapsto g^{\uparrow}(\bar\theta - s; \tau^*-\bar\theta,\zeta^*
e^{-\kappa\bar\theta},\varphi^*)$ is strictly convex on
$[0,\bar\theta]$, we can deduce that
$\nabla_{t}^{\uparrow}J_\tau(\hat{X}^{\tau,\zeta,\varphi}) > 0$ for
all
$t \in (\tau_{\mathrm{buy}}(\tau,\zeta),
\tau_{\mathrm{buy}}(\tau,\zeta)+\bar\theta]$. To complete the
verification of optimality of strategy $\hat{X}^{\tau,\zeta,\varphi}$,
we need to check that
$\nabla_{t}^{\uparrow}J_\tau(\hat{X}^{\tau,\zeta,\varphi}) = 0$ for
all $t \in [0, \tau_{\mathrm{buy}}(\tau,\zeta)]$. Indeed, once more
simple but tedious algebraic manipulations show that this holds
true. 
To sum up, it follows from the
first order optimality conditions in Proposition~\ref{prop:foc} that
$\hat{X}^{\tau,\zeta,\varphi}$ in~\eqref{def:optbuysellstratbuy}
and~\eqref{def:optbuysellstratsell} is optimal. Hence, we can conclude
that $(\tau,\zeta,\varphi)$ with
$\varphi = \phi_{\mathrm{buy}}(\tau,\zeta) =
\hat{\varphi}^{\mathrm{buy}} (\tau, \zeta^*,$ $\varphi^*, \tau^*)$
in~\eqref{eq:proofphibuy} belongs to
$\partial \mathcal{R}_{\mathrm{buy}}$ as defined
in~\eqref{def:buybound} with Corollary~\ref{cor:main2} 2.) holding
true for these triplets.

\underline{Case 2.2 (part II.2 in Fig.~\ref{fig:domain}):} Let us next
consider one of the two cases where either $\tau \geq \bar\theta$ and
$\hat{\zeta}^{\mathrm{buy}}(\tau,\bar\zeta(\bar\theta),
\bar\varphi(\bar\theta),\bar\theta) \leq \zeta <
\hat{\zeta}^{\mathrm{buy}}(\tau,\bar\zeta(\underline\theta),
\bar\varphi(\underline\theta),\underline\theta)$ or
$\underline\theta < \tau < \bar\theta$ and
$s_2(\tau) \leq \zeta <
\hat{\zeta}^{\mathrm{buy}}(\tau,\bar\zeta(\underline\theta),
\bar\varphi(\underline\theta),\underline\theta)$. Recall that we are
still given $\tau^*, \zeta^*, \varphi^*$ from~\eqref{def:starstriplet}
as well as the identity in~\eqref{eq:proofphibuy}. Notice, though,
that $\tau^* \in (\underline\theta,\bar\theta]$ in view of
Lemma~\ref{lem:hatfunctions} 1.) and 2.). In each of both considered
cases, we claim that the optimal strategy
$\hat{X}^{\tau,\zeta,\varphi}=(\hat{X}^{\tau,\zeta,\varphi,\uparrow},
\hat{X}^{\tau,\zeta,\varphi,\downarrow})$ is given as follows: The
cumulative purchases of the risky asset are still prescribed as in
\eqref{def:optbuysellstratbuy} above with
$\hat{X}_0^{\tau,\zeta,\varphi,\uparrow}=0$ and
$\{ d\hat{X}^{\tau,\zeta,\varphi,\downarrow}>0\} =
[0,\tau_{\mathrm{buy}}(\tau,\zeta))$. In contrast, the cumulative
sells of the risky asset are now given by
$\hat{X}_t^{\tau,\zeta,\varphi,\downarrow} \equiv 0$ on $[0,\tau]$. As
a consequence, compared to~\eqref{def:optbuysellphi}
and~\eqref{def:optbuysellzeta}, the corresponding induced share
holdings and spread dynamics simplify to
\begin{equation} \label{def:optbuyholdphi}
  \varphi_t^{\hat{X}^{\tau,\zeta,\varphi}} =
  \begin{cases}
    \hat\varphi^{\mathrm{buy}}(\tau-t,\zeta^*,\varphi^*,\tau^*),
    & 0 \leq t \leq \tau_{\mathrm{buy}}(\tau,\zeta), \\
    \varphi^*, & \tau_{\mathrm{buy}}(\tau,\zeta) < t \leq \tau,
  \end{cases}
\end{equation}
and
\begin{equation} \label{def:optbuyholdzeta}
  \zeta_t^{\hat{X}^{\tau,\zeta,\varphi}} =
  \begin{cases}
    \hat\zeta^{\mathrm{buy}}(\tau-t,\zeta^*,\varphi^*,\tau^*),
    & 0 \leq t \leq \tau_{\mathrm{buy}}(\tau,\zeta), \\
    \zeta^* e^{-\kappa(t-\tau_{\mathrm{buy}}(\tau,\zeta))}, &
    \tau_{\mathrm{buy}}(\tau,\zeta) < t \leq \tau.
  \end{cases}
\end{equation}
Notice that
$\varphi_{\tau_{\mathrm{buy}}(\tau,\zeta)}^{\hat{X}^{\tau,\zeta,\varphi}}
= \varphi^* = \bar\varphi(\tau^*)> 0$ (cf. Lemma~\ref{lem:fbcurve}
4.)) and
$\zeta_{\tau_{\mathrm{buy}}(\tau,\zeta)}^{\hat{X}^{\tau,\zeta,\varphi}}
= \zeta^*$ by virtue of~\eqref{eq:propbuy}. Moreover, following the
definition in~\eqref{set:tauwait}, we have
$\tau_{\mathrm{wait}}(\tau,\zeta) =
\tau-\tau_{\mathrm{buy}}(\tau,\zeta) > 0$ in the current
setting. Hence, $\tau_{\mathrm{sell}}(\tau,\zeta) = 0$
in~\eqref{def:tausell} in Corollary~\ref{cor:main2} which is in line
with the fact that $\hat{X}^{\tau,\zeta,\varphi,\downarrow} \equiv 0$
on $[0,\tau]$. All other assertions in Corollary~\ref{cor:main2} 2.)
can be easily checked as in Step 2.1. Next, very similar arguments as
in Step 2.1 above allow us to verify via the first order conditions in
Proposition~\ref{prop:foc} that the strategy
$\hat{X}^{\tau,\zeta,\varphi}=(\hat{X}^{\tau,\zeta,\varphi,\uparrow},0)$
with $\hat{X}^{\tau,\zeta,\varphi,\uparrow}$ given
in~\eqref{def:optbuysellstratbuy} is optimal. First, we check the
sell- and buy-subgradients on
$[\tau_{\mathrm{buy}}(\tau,\zeta),\tau]$. Due to the construction of
$\hat{X}^{\tau,\zeta,\varphi}$, we can again refer to
Lemma~\ref{lem:waitgradients} 1.) (which is applicable here in light
of our convention in Remark~\ref{rem:tauzero} 1.)) and obtain for all
$t \in [\tau_{\mathrm{buy}}(\tau,\zeta),\tau]$ the expressions
\begin{equation} \label{eq:mastergrad2}
  \begin{aligned}
    &
    \nabla_{t}^{\uparrow,\downarrow}J_\tau(\hat{X}^{\tau,\zeta,\varphi})
    = \nabla_{t-\tau_{\mathrm{buy}}(\tau,\zeta)}^{\uparrow,\downarrow}
    J_{\tau^*}(\hat{X}^{\tau^*,\zeta^*,\varphi^*}) \\
    & =
    \nabla_{0}^{\uparrow,\downarrow}J_{\tau-t}(\hat{X}^{\tau-t,\zeta^*e^{-\kappa
        (t-\tau_{\mathrm{buy}}(\tau,\zeta))},\varphi^*}) =
    g^{\uparrow,\downarrow}( \tau -t;
    0,\zeta^* e^{-\kappa\tau^*},\varphi^*) \\
    & = \pm(\lambda^2 \varphi^* - \mu) (\tau - t) + \frac{1}{2}
    \zeta^* e^{-\kappa\tau^*} (e^{\kappa(\tau-t)}\pm 1) +
    \frac{\eta}{2} \varphi^* (e^{-\kappa(\tau-t)} \pm 1)
  \end{aligned}
\end{equation}
with
$\nabla_{\tau}^{\downarrow}J_\tau(\hat{X}^{\tau,\zeta,\varphi}) =
\nabla_{0}^{\downarrow}J_0(\hat{X}^{0,\zeta^*
  e^{-\kappa\tau},\varphi^*})=g^{\downarrow}(0; 0,\zeta^*
e^{-\kappa\tau^*},\varphi^*)=0$. Using the monotonicity properties
from Lemma~\ref{lem:fbcurve} 3.)  and 4.), it holds that
\begin{equation*}
  \varphi^* \leq \bar\varphi(\bar\theta) = 
  \phi_{\textrm{sell}}(0,\bar\zeta(\bar\theta)
  e^{-\kappa\bar\theta}) = \frac{2\mu + \kappa \bar\zeta(\bar\theta)
    e^{-\kappa\bar\theta}}{2\lambda^2 + \kappa\eta} < \frac{2\mu + \kappa \zeta^*
    e^{-\kappa\tau^*}}{2\lambda^2 + \kappa\eta},
\end{equation*}
which implies
$\frac{\partial}{\partial t} g^{\downarrow}(\tau -t; 0,\zeta^*
e^{-\kappa\tau^*},\varphi^*) < 0$ and hence
$\nabla_{t}^{\downarrow}J_\tau(\hat{X}^{\tau,\zeta,\varphi}) > 0$ for
$t \in [\tau_{\mathrm{buy}}(\tau,\zeta),\tau)$. Concerning the
buy-subgradient, we have
$\nabla_{\tau}^{\uparrow}J_\tau(\hat{X}^{\tau,\zeta,\varphi}) =
g^{\uparrow}(0; 0,\zeta^* e^{-\kappa\tau^*},\varphi^*) > 0$. In
addition, using the fact that $\varphi^* = \bar\varphi(\tau^*)$ as
in~\eqref{eq:barphi2} and $\zeta^* = \bar\zeta(\tau^*)=s_2(\tau^*)$ as
in~\eqref{def:s2}, one can verify that
$g^{\uparrow}( \tau^*; 0,\zeta^* e^{-\kappa\tau^*},\varphi^*) = 0$ as
well as
$\frac{\partial}{\partial \theta} g^{\uparrow} (\theta ; 0,\zeta^*
e^{-\kappa\tau^*},\varphi^*)\vert_{\theta = \tau^*} = 0$.  But this
implies
$\nabla_{t}^{\uparrow}J_\tau(\hat{X}^{\tau,\zeta,\varphi}) > 0$ for
all $t \in (\tau_{\mathrm{buy}}(\tau,\zeta),\tau]$ because
$t \mapsto g^{\uparrow}( \tau^* - t; 0,\zeta^*
e^{-\kappa\tau^*},\varphi^*)$ is strictly convex on $[0,\tau^*]$. To
complete the verification of optimality of strategy
$\hat{X}^{\tau,\zeta,\varphi}$, one sees as in Step 2.1 that
$\nabla_{t}^{\uparrow}J_\tau(\hat{X}^{\tau,\zeta,\varphi}) = 0$ for
all $t \in [0, \tau_{\mathrm{buy}}(\tau,\zeta)]$. Hence, we can
conclude that
$(\tau,\zeta,\varphi) = (\tau,\zeta,\phi_{\mathrm{buy}}(\tau,\zeta))$
belongs to $\partial \mathcal{R}_{\mathrm{buy}}$ as defined
in~\eqref{def:buybound}.

\underline{Case 2.3 (part II.3 in Fig.~\ref{fig:domain}):} Consider
next one of the two cases where either
$\underline\theta \leq \tau < \bar\theta$ and
$\hat{\zeta}^{\mathrm{buy}}(\tau,\bar\zeta(\underline\theta),\bar\varphi(\underline\theta),\underline\theta)
\leq \zeta <
\hat{\zeta}^{\mathrm{buy}}(\tau,\bar\zeta(0),\bar\varphi(0),0)$, or
$0 < \tau < \underline\theta$ and
$2\mu/\kappa \leq \zeta <
\hat{\zeta}^{\mathrm{buy}}(\tau,\bar\zeta(0), \bar\varphi(0),0)$. Due
to Lemma~\ref{lem:hatfunctions} 1.) and 2.), we now have
$\tau^* = \tau -\tau_{\mathrm{buy}}(\tau,\zeta) \in
(0,\underline\theta]$ which implies
$\zeta^* = \bar\zeta(\tau^*) = 2\mu/\kappa$ and
$\varphi^* = \varphi(\tau^*) = 0$ in~\eqref{def:starstriplet} (recall
the definitions in~\eqref{def:barzeta} and~\eqref{def:barphi}). In
each of these cases, we claim that the optimal strategy
$\hat{X}^{\tau,\zeta,\varphi}=(\hat{X}^{\tau,\zeta,\varphi,\uparrow},
\hat{X}^{\tau,\zeta,\varphi,\downarrow})$ is prescribed as in Case 2.2
with controlled dynamics~\eqref{def:optbuyholdphi}
and~\eqref{def:optbuyholdzeta}. As a consequence, all assertions in
Corollary~\ref{cor:main2} 2.) still hold true in the current setting
and we again have $\tau_{\mathrm{sell}}(\tau,\zeta) = 0$
in~\eqref{def:tausell}. Optimality can once more be verified via the
first order conditions in Proposition~\ref{prop:foc} with similar
arguments as in Steps 2.1 and 2.2. Notice, though, that
$\varphi_{\tau}^{\hat{X}^{\tau,\zeta,\varphi}} = \varphi^* = 0$, that
is, the first order conditions need to be checked with a proper choice
of subgradients depending on $\varrho$. Therefore, we set
$\varrho^*\set e^{\kappa \tau^*} (\kappa\tau^*-1)$. Observe that
$\varrho^* \in (-1,1]$ since $\tau^* \in (0,\underline\theta]$ (recall
that $\underline\theta$ satisfies~\eqref{eq:besslich2}). Then, it
follows by construction of $\hat{X}^{\tau,\zeta,\varphi}$ and
Lemma~\ref{lem:waitgradients} 2.) that the buy- and sell-subgradients
on $[\tau_{\mathrm{buy}}(\tau,\zeta),\tau]$ are given by
\begin{equation} \label{eq:mastergrad3}
  \begin{aligned}
    &
    {}^{\varrho^*}\nabla_{t}^{\uparrow,\downarrow}J_\tau(\hat{X}^{\tau,\zeta,\varphi})
    =
    {}^{\varrho^*}\nabla_{t-\tau_{\mathrm{buy}}(\tau,\zeta)}^{\uparrow,\downarrow}
    J_{\tau^*}(\hat{X}^{\tau^*,\zeta^*,\varphi^*}) \\
    & = {}^{\varrho^*}\nabla_{0}^{\uparrow,\downarrow}
    J_{\tau-t}(\hat{X}^{\tau-t,\zeta^*e^{-\kappa
        (t-\tau_{\mathrm{buy}}(\tau,\zeta))},\varphi^*}) =
    h^{\uparrow,\downarrow}( \tau -t;
    \zeta^* e^{-\kappa\tau^*},\varrho^*) \\
    & = \mp \mu (\tau - t) + \frac{1}{2} \zeta^* e^{-\kappa\tau^*}
    (e^{\kappa(\tau-t)} \pm \varrho^*).
  \end{aligned}
\end{equation}
Obviously,
${}^{\varrho^*}\nabla_{t}^{\downarrow}J_\tau(\hat{X}^{\tau,\zeta,\varphi})
\geq 0$ on $[\tau_{\mathrm{buy}}(\tau,\zeta),\tau]$. Moreover, it
holds that
${}^{\varrho^*}\nabla_{\tau_{\mathrm{buy}}(\tau,\zeta)}^{\uparrow}
J_{\tau}(\hat{X}^{\tau,\zeta,\varphi})$
$= h^{\uparrow}(\tau^*; \zeta^* e^{-\kappa\tau^*},\varrho^*) = 0$ and
$\frac{\partial}{\partial\theta} h^{\uparrow}(\theta; \zeta^*
e^{-\kappa\tau^*},\varrho^*)\vert_{\theta=\tau^*} = 0$, which implies
${}^{\varrho^*}\nabla_{t}^{\downarrow}J_\tau(\hat{X}^{\tau,\zeta,\varphi})
> 0$ on $(\tau_{\mathrm{buy}}(\tau,\zeta),\tau]$ due to strict
convexity of the mapping
$t \mapsto h^{\uparrow}( \tau -t; \zeta^*
e^{-\kappa\tau^*},\varrho^*)$ on
$[\tau_{\mathrm{buy}}(\tau,\zeta),\tau]$. Concerning the interval
$[0, \tau_{\mathrm{buy}}(\tau,\zeta)]$, one can check as in Step 2.1
and 2.2 that the buy-gradient vanishes. Hence,
$\hat{X}^{\tau,\zeta,\varphi}$ is optimal and we can conclude that
$(\tau,\zeta,\varphi) = (\tau,\zeta,\phi_{\mathrm{buy}}(\tau,\zeta))$
belongs to $\partial \mathcal{R}_{\mathrm{buy}}$ as defined
in~\eqref{def:buybound}.

\underline{Case 3:} In order to finalize Step 3 concerning the
boundary of the buying region $\partial \mathcal{R}_{\mathrm{buy}}$
and the claim in~\eqref{pf:buybound}, we have to address the case
$0 \leq \zeta \leq \bar\zeta(\tau)$. This will be proved together with the
assertion in Corollary~\ref{cor:main2}~3.). Regarding the definition
of $\phi_{\mathrm{buy}}$ in~\eqref{def:phibuy:eq3a},
\eqref{def:phibuy:eq3b}, and~\eqref{def:phibuy:eq3c}, we have to carry
out once more a refined analysis.

\underline{Case 3.1 (part III.1 in Fig.~\ref{fig:domain}):} Let either
$\tau \geq \bar\theta$ and
$0 \leq \zeta \leq \bar\zeta(\tau) = s_1(\tau-\bar\theta,\bar\theta)$,
or $\underline\theta < \tau \leq \bar\theta$ and
$0 \leq \zeta < s_1(0,\tau)$. In view of the definitions
in~\eqref{def:phibuy:eq3a} and~\eqref{def:phibuy:eq3b}, we have
\begin{equation} \label{phi:optwaitsellstrat}
  \varphi = \phi_{\mathrm{buy}}(\tau,\zeta) =
  \phi_{\mathrm{sell}}(\tau-\tau_{\mathrm{wait}}(\tau,\zeta),\zeta
  e^{-\kappa \tau_{\mathrm{wait}}(\tau,\zeta)}) > 0
\end{equation}
with $\tau_{\mathrm{wait}}(\tau,\zeta) \in (0,\tau)$ as defined
in~\eqref{def:tauwait}. In particular, recall that this implies
$\zeta = s_1 ( \tau-\tau_{\mathrm{wait}}(\tau,\zeta),
\tau_{\mathrm{wait}}(\tau,\zeta))$. In both cases, we claim that the
optimal strategy
$\hat{X}^{\tau,\zeta,\varphi}=(\hat{X}^{\tau,\zeta,\varphi,\uparrow},
\hat{X}^{\tau,\zeta,\varphi,\downarrow})$ is given by
\begin{equation} \label{def:optwaitsellstratsell}
  \begin{aligned}
    \hat{X}_t^{\tau,\zeta,\varphi,\uparrow} & = 0 \qquad (0 \leq t \leq \tau), \\
    \hat{X}_t^{\tau,\zeta,\varphi,\downarrow} & =
    \begin{cases}
      0 & \text{if } 0 \leq t < \tau_{\mathrm{wait}}(\tau,\zeta), \\
      \hat{X}_{t-\tau_{\mathrm{wait}}(\tau,\zeta)}^{\tau-\tau_{\mathrm{wait}}(\tau,\zeta),\zeta
        e^{-\kappa
          \tau_{\mathrm{wait}}(\tau,\zeta)},\varphi,\downarrow} &
      \text{if } \tau_{\mathrm{wait}}(\tau,\zeta) \leq t \leq \tau.
    \end{cases}
  \end{aligned}
\end{equation}
Note that~\eqref{phi:optwaitsellstrat} immediately yields
\begin{equation} \label{prop:optwaitsellstrat}
  (\tau-\tau_{\mathrm{wait}}(\tau,\zeta),\zeta e^{-\kappa
    \tau_{\mathrm{wait}}(\tau,\zeta)},\varphi) \in
  \partial \mathcal{R}_{\mathrm{sell}}
\end{equation}
due to Step 1.  That is,
$\hat{X}_{\cdot}^{\tau-\tau_{\mathrm{wait}}(\tau,\zeta),\zeta
  e^{-\kappa \tau_{\mathrm{wait}}(\tau,\zeta)},\varphi,\downarrow}$
denotes the optimal cumulative sells on
$[0, \tau-\tau_{\mathrm{wait}}(\tau,\zeta)]$ as given
in~\eqref{def:optsellstrat1} for the triplet
in~\eqref{prop:optwaitsellstrat}. Hence, the associated share holdings
and spread dynamics for strategy $\hat{X}^{\tau,\zeta,\varphi}$ are
given by
\begin{equation} \label{def:optwaitsellphi}
  \varphi_t^{\hat{X}^{\tau,\zeta,\varphi}} =
  \begin{cases}
    \varphi,
    & 0 \leq t < \tau_{\mathrm{wait}}(\tau,\zeta), \\
    \hat{\varphi}^{\mathrm{sell}}\big(\tau-\tau_{\mathrm{wait}}(\tau,\zeta),\zeta
    e^{-\kappa \tau_{\mathrm{wait}}(\tau,\zeta)}, & \\
    \hspace{26pt} \varphi, t-\tau_{\mathrm{wait}}(\tau,\zeta) \big), &
    \tau_{\mathrm{wait}}(\tau,\zeta) \leq t \leq \tau,
  \end{cases}
\end{equation}
and
\begin{equation} \label{def:optwaitsellzeta}
  \zeta_t^{\hat{X}^{\tau,\zeta,\varphi}} =
  \begin{cases}
    \zeta e^{-\kappa t}, &
    0 \leq t < \tau_{\mathrm{wait}}(\tau,\zeta), \\
    \hat{\zeta}^{\mathrm{sell}}(\tau-\tau_{\mathrm{wait}}(\tau,\zeta),\zeta
    e^{-\kappa \tau_{\mathrm{wait}}(\tau,\zeta)}, & \\
    \hspace{23pt} \varphi, t-\tau_{\mathrm{wait}}(\tau,\zeta)), &
    \tau_{\mathrm{wait}}(\tau,\zeta) \leq t \leq \tau.
  \end{cases}
\end{equation}
Observe that~\eqref{phi:optwaitsellstrat} also implies
$\varphi_{\tau_{\mathrm{wait}}(\tau,\zeta)}^{\hat{X}^{\tau,\zeta,\varphi}}
= \varphi$ and
$\zeta_{\tau_{\mathrm{wait}}(\tau,\zeta)}^{\hat{X}^{\tau,\zeta,\varphi}}
= \zeta e^{-\kappa \tau_{\mathrm{wait}}(\tau,\zeta)}$ by virtue
of~\eqref{eq:propbuy}, \eqref{eq:propsell}. Moreover, due to the
definition of $\tau_{\mathrm{buy}}$ in~\eqref{set:taubuy}, we have
$\tau_{\mathrm{buy}}(\tau,\zeta) = 0$ in the current setup. Thus,
refering to~\eqref{def:tausell} and~\eqref{eq:waitandsell} in
Corollary~\ref{cor:main2}, we obtain
$\tau_{\mathrm{sell}}(\tau,\zeta) =
\tau-\tau_{\mathrm{wait}}(\tau,\zeta) > 0$, which is in line
with~\eqref{prop:optwaitsellstrat},~\eqref{def:optwaitsellphi}
and~\eqref{def:optwaitsellzeta} above. Next, optimality of strategy
$\hat{X}^{\tau,\zeta,\varphi}$ on
$[\tau_{\mathrm{wait}}(\tau,\zeta), \tau]$ follows by Step
1. Moreover, since
$\nabla_{\tau_{\mathrm{wait}}(\tau,\zeta)}^{\downarrow}
J_\tau(\hat{X}^{\tau,\zeta,\varphi})=0$, we obtain analogously
to~\eqref{eq:mastergrad1} for the sell- and buy-subgradients on
$[0,\tau_{\mathrm{wait}}(\tau,\zeta)]$ the expressions
\begin{equation*} 
  \begin{aligned}
    & \nabla_{t}^{\uparrow,\downarrow}
    J_\tau(\hat{X}^{\tau,\zeta,\varphi}) = g^{\uparrow,\downarrow}
    (\tau_{\mathrm{wait}}(\tau,\zeta)-t;
    \tau-\tau_{\mathrm{wait}}(\tau,\zeta),\zeta e^{-\kappa
      \tau_{\mathrm{wait}}(\tau,\zeta)}, \varphi) \\
    & = \pm(\lambda^2 \varphi - \mu)
    (\tau_{\mathrm{wait}}(\tau,\zeta)-t) + \frac{1}{2} \zeta
    e^{-\kappa \tau_{\mathrm{wait}}(\tau,\zeta)}
    (e^{\kappa(\tau_{\mathrm{wait}}(\tau,\zeta)-t)}\pm 1) \\
    & \hspace{13.5pt} + \frac{1}{\kappa}
    (e^{-\kappa(\tau_{\mathrm{wait}}(\tau,\zeta)-t)} \pm 1) (
    -\lambda^2\varphi + \mu + \frac{1}{2} \kappa \zeta e^{-\kappa
      \tau_{\mathrm{wait}}(\tau,\zeta)}).
  \end{aligned}
\end{equation*}
In fact, by similar convexity arguments as in Step 2.1 we have
$\nabla_{t}^{\downarrow} J_\tau(\hat{X}^{\tau,\zeta,\varphi}) >0$ on
the interval $[0, \tau_{\mathrm{wait}}(\tau,\zeta))$ as well as
$\nabla_{t}^{\uparrow} J_\tau(\hat{X}^{\tau,\zeta,\varphi}) > 0$ on
$(0, \tau_{\mathrm{wait}}(\tau,\zeta)]$. Indeed, since
$\zeta = s_1 ( \tau-\tau_{\mathrm{wait}}(\tau,\zeta),
\tau_{\mathrm{wait}}(\tau,\zeta))$ and
$\varphi =
\phi_{\mathrm{sell}}(\tau-\tau_{\mathrm{wait}}(\tau,\zeta),$ $\zeta
e^{-\kappa \tau_{\mathrm{wait}}(\tau,\zeta)})$, one can compute
$\nabla_{0}^{\uparrow} J_\tau(\hat{X}^{\tau,\zeta,\varphi}) =
g^{\uparrow} (\tau_{\mathrm{wait}}(\tau,\zeta);
\tau-\tau_{\mathrm{wait}}(\tau,\zeta),$ $\zeta e^{-\kappa
  \tau_{\mathrm{wait}}(\tau,\zeta)},\varphi) = 0$ and
$\frac{\partial}{\partial \theta} g^{\uparrow} (\theta;
\tau-\tau_{\mathrm{wait}}(\tau,\zeta),$
$\zeta e^{-\kappa \tau_{\mathrm{wait}}(\tau,\zeta)},\varphi)
\vert_{\theta=\tau_{\mathrm{wait}}(\tau,\zeta)} < 0$. Consequently, by
virtue of the first order conditions in Proposition~\ref{prop:foc}, it
follows that $\hat{X}^{\tau,\zeta,\varphi}$ is optimal. In particular,
we can conclude that $(\tau,\zeta,\varphi)$ with
$\varphi = \phi_{\mathrm{buy}}(\tau,\zeta)$ given
in~\eqref{phi:optwaitsellstrat} belongs to
$\partial \mathcal{R}_{\mathrm{buy}}$ as defined
in~\eqref{def:buybound} with Corollary~\ref{cor:main2} 3.) holding
true for these triplets.

\underline{Case 3.2 (part III.2 in Fig.~\ref{fig:domain}):} In case
$\underline\theta \leq \tau < \bar\theta$ and
$s_1(0,\tau) < \zeta \leq \bar\zeta(\tau) = s_2(\tau)$, or
$0 \leq \tau < \underline\theta$ and $0 \leq \zeta < s_3(\tau)$, we
now have
\begin{equation} \label{phi:optwaitstrat1} 
  \varphi =
  \phi_{\mathrm{buy}}(\tau,\zeta) = \phi_{2}(\tau,\zeta) =
  \frac{\mu\tau -\frac{1}{2} \zeta (1+e^{-\kappa\tau})}{\lambda^2 \tau
    + \frac{1}{2}\eta(1+e^{-\kappa\tau})} > 0
\end{equation}
due to the definitions in~\eqref{def:phibuy:eq3b},
\eqref{def:phibuy:eq3c},~\eqref{def:phi2} and the monotonicity
properties of $\phi_2$. In both above cases, we claim that the optimal
strategy $\hat{X}^{\tau,\zeta,\varphi}$ is given by
$\hat{X}^{\tau,\zeta,\varphi,\uparrow}_t =
\hat{X}^{\tau,\zeta,\varphi,\downarrow}_t = 0$ for all
$t \in [0,\tau]$. Hence, the corresponding dynamics for the share
holdings and the spread simplify to
$\varphi_t^{\hat{X}^{\tau,\zeta,\varphi}} = \varphi$ and
$\zeta_t^{\hat{X}^{\tau,\zeta,\varphi}} = \zeta e^{-\kappa t}$,
$t \in [0,\tau]$. Notice that $\tau_{\mathrm{buy}}(\tau,\zeta) = 0$
and $\tau_{\mathrm{wait}}(\tau,\zeta) = \tau$ due to the definitions
in~\eqref{set:taubuy} and~\eqref{set:tauwait} which yields
$\tau_{\mathrm{sell}}(\tau,\zeta) = 0$
in~\eqref{def:tausell}. Concerning the proof of optimality via
Proposition~\ref{prop:foc}, we obtain for the buy- and
sell-subgradients on $[0,\tau]$ similar to~\eqref{eq:mastergrad2} the
representations
\begin{equation*}
  \begin{aligned}
    &
    \nabla_{t}^{\uparrow,\downarrow}J_\tau(\hat{X}^{\tau,\zeta,\varphi})
    = g^{\uparrow,\downarrow}( \tau -t;
    0,\zeta e^{-\kappa\tau},\varphi) \\
    & = \pm(\lambda^2 \varphi - \mu) (\tau - t) + \frac{1}{2} \zeta
    e^{-\kappa\tau} (e^{\kappa(\tau-t)}\pm 1) + \frac{\eta}{2} \varphi
    (e^{-\kappa(\tau-t)} \pm 1).
  \end{aligned}
\end{equation*} 
By utilizing the identity in~\eqref{phi:optwaitstrat1} and similar
convexity arguments as in Step 2.2, it holds that
$\nabla_{t}^{\downarrow}J_\tau(\hat{X}^{\tau,\zeta,\varphi}) > 0$ on
$[0,\tau)$ as well as
$\nabla_{t}^{\uparrow}J_\tau(\hat{X}^{\tau,\zeta,\varphi}) > 0$ on
$(0,\tau]$ with
$\nabla_{0}^{\uparrow}J_\tau(\hat{X}^{\tau,\zeta,\varphi}) =
0$. Therefore, we can conclude that $(\tau,\zeta,\varphi)$ with
$\varphi = \phi_{\mathrm{buy}}(\tau,\zeta)$ given
in~\eqref{phi:optwaitstrat1} belongs to
$\partial \mathcal{R}_{\mathrm{buy}}$ as defined
in~\eqref{def:buybound} with Corollary~\ref{cor:main2} 3.) holding
true for these triplets.

\underline{Case 3.3 (part III.3 in Fig.~\ref{fig:domain}):} Finally,
in case $0 < \tau < \underline\theta$ and
$s_3(\tau) \leq \zeta \leq \bar\zeta(\tau) = 2\mu/\kappa$, we have
$\varphi = \phi_{\mathrm{buy}}(\tau,\zeta) = 0$ due
to~\eqref{def:phibuy:eq3c}. As in Case 3.2 above, we claim that the
optimal strategy $\hat{X}^{\tau,\zeta,\varphi}$ is again given by
$\hat{X}^{\tau,\zeta,\varphi,\uparrow}_t =
\hat{X}^{\tau,\zeta,\varphi,\downarrow}_t = 0$ for all
$t \in [0,\tau]$ with $\tau_{\mathrm{sell}}(\tau,\zeta) = 0$
in~\eqref{def:tausell}. Optimality can be checked via
Proposition~\ref{prop:foc} similar to Step 2.3 above. Indeed, since
$\varphi_{\tau}^{\hat{X}^{\tau,\zeta,\varphi}} = \varphi = 0$, we set
$\varrho^* \set e^{\kappa \tau} ( 2\mu\tau/\zeta -1)$. Notice that
$\varrho^* \in [-1,1]$ in the current setup. Next, analog
to~\eqref{eq:mastergrad3}, we obtain for the buy- and
sell-subgradients on $[0,\tau]$ the representations
${}^{\varrho^*}\nabla_{t}^{\uparrow,\downarrow}J_\tau(\hat{X}^{\tau,\zeta,0})
= h^{\uparrow,\downarrow}( \tau -t; \zeta e^{-\kappa\tau},\varrho^*) =
\mp \mu (\tau - t) + \frac{1}{2} \zeta e^{-\kappa\tau}
(e^{\kappa(\tau-t)} \pm \varrho^*)$.  Obviously, it holds that
${}^{\varrho^*}\nabla_{t}^{\downarrow}J_\tau(\hat{X}^{\tau,\zeta,0})
\geq 0$ on $[0,\tau]$. Moreover, we have
${}^{\varrho^*}\nabla_{0}^{\uparrow}J_{\tau}(\hat{X}^{\tau,\zeta,0}) =
h^{\uparrow}(\tau; \zeta e^{-\kappa\tau},\varrho^*) = 0$ and
$\frac{\partial}{\partial t} h^{\uparrow}(\tau - t; \zeta
e^{-\kappa\tau},\varrho^*) > 0$. But this implies
${}^{\varrho^*}\nabla_{t}^{\downarrow}J_\tau(\hat{X}^{\tau,\zeta,\varphi})
> 0$ on $(0,\tau]$. As a consequence, we obtain that
$\hat{X}^{\tau,\zeta,\varphi}$ is optimal and that
$(\tau,\zeta, \phi_{\mathrm{buy}}(\tau,\zeta)) = (\tau,\zeta,0)$
belongs to $\partial \mathcal{R}_{\mathrm{buy}}$ as defined
in~\eqref{def:buybound} with Corollary~\ref{cor:main2} 3.) holding
true for these triplets. This finishes Step 3 and the proof of the
claim in~\eqref{pf:buybound}.

\emph{Step 4:} Concerning the claim in~\eqref{pf:buyreg} for the
buying-region $\mathcal{R}_{\mathrm{buy}}$, the reasoning follows
along the same lines as in Step 2 for the selling-region
$\mathcal{R}_{\mathrm{sell}}$. That is, for any
$(\tau,\zeta,\varphi) \in \cS$ with
$\varphi < \phi_{\mathrm{buy}}(\tau,\zeta)$ the corresponding optimal
strategy
$\hat{X}^{\tau,\zeta,\varphi}=(\hat{X}^{\tau,\zeta,\varphi,\uparrow},
\hat{X}^{\tau,\zeta,\varphi,\downarrow}) \in \cX^d$ is in fact given
by
\begin{equation} \label{def:optbuystrat2}
  \hat{X}_t^{\tau,\zeta,\varphi,\uparrow} = x^{\uparrow} +
  \hat{X}_t^{\tau,\zeta+\eta
    x^{\uparrow},\varphi+x^{\uparrow},\uparrow}, \quad
  \hat{X}_t^{\tau,\zeta,\varphi,\downarrow} =
  \hat{X}_t^{\tau,\zeta+\eta
    x^{\uparrow},\varphi+x^{\uparrow},\downarrow}
\end{equation}
for all $t \in [0,\tau]$, where $x^{\uparrow}>0$ denotes the unique
solution to the equation
\begin{equation} \label{def:xup} 
  \varphi + x^{\uparrow} =
  \phi_{\mathrm{buy}}(\tau, \zeta + \eta x^{\uparrow}).
\end{equation}
Notice that~\eqref{def:xup} implies
$(\tau,\zeta + \eta x^{\uparrow}, \varphi + x^{\uparrow}) \in
\partial \mathcal{R}_{\mathrm{buy}}$ by virtue of Step~3. Therefore,
$\hat{X}^{\tau,\zeta+\eta x^{\uparrow},\varphi+x^{\uparrow}}$ denotes
the corresponding optimal strategy as prescribed in one of the
different cases presented in Step 3 above. Optimality of the strategy
in~\eqref{def:optbuystrat2} then follows as in Step 2 by virtue of the
first order optimality conditions in Proposition~\ref{prop:foc} and
the fact that they are satisfied by
$\hat{X}^{\tau,\zeta+\eta x^{\uparrow},\varphi+x^{\uparrow}}$. In
particular, it holds that
$\nabla_{0}^{\uparrow} J_\tau(\hat{X}^{\tau,\zeta,\varphi}) = 0$ and
$\hat{X}_0^{\tau,\zeta,\varphi,\uparrow} = x^\uparrow > 0$ which
implies that $(\tau,\zeta,\varphi)$ belongs to
$\mathcal{R}_{\mathrm{buy}}$ as defined in~\eqref{def:buyregion}.

\emph{Step 5:} We now argue that inequality~\eqref{eq:phisellphibuy}
holds true, i.e.,
$\phi_{\mathrm{sell}}(\tau,\zeta) > \phi_{\mathrm{buy}}(\tau,\zeta)$
on $[0,+\infty)^2$. Observe that this actually follows from the fact
that $\phi_{\mathrm{sell}}(\tau,\zeta) > 0$ on $[0,+\infty)^2$ (recall
Lemma~\ref{lem:phisell}) but, e.g.,
$\phi_{\mathrm{buy}}(\tau,\bar\zeta(\tau)) = 0$ for all
$\tau \in [0,\underline\theta]$ together with~\eqref{pf:sellbound}
and~\eqref{pf:buybound} as well as
$\partial \mathcal{R}_{\mathrm{buy}} \cap \partial
\mathcal{R}_{\mathrm{sell}} = \varnothing$
(cf. Lemma~\ref{lem:nobuysell}).

\emph{Step 6:} It is left to prove~\eqref{pf:waitreg}. We will only
sketch the argument. For this, let
$(\tau,\zeta,\varphi) \in \cS$ be such that
$\phi_{\mathrm{buy}}(\tau,\zeta) < \varphi <
\phi_{\mathrm{sell}}(\tau,\zeta)$. It is easy to observe that the
continuous mapping
$t \mapsto \phi_{\mathrm{sell}}(\tau-t,\zeta e^{-\kappa t})$ is
decreasing on $[0,\tau]$. In addition, one can also check that the
continuous mapping
$t \mapsto \phi_{\mathrm{buy}}(\tau-t,\zeta e^{-\kappa t})$ is
increasing for those $t \in [0,\tau]$ such that
$\zeta e^{-\kappa t} \geq \bar\zeta(\tau - t)$, that is, when
$\phi_{\mathrm{buy}}$ is either given as in~\eqref{def:phibuy:eq1}
or~\eqref{def:phibuy:eq2}. Otherwise, if $\zeta < \bar\zeta(\tau)$, it
holds that the mapping
$t \mapsto \phi_{\mathrm{buy}}(\tau-t,\zeta e^{-\kappa t})$ is
non-increasing on $[0,\tau]$. This is the case when
$\phi_{\mathrm{buy}}$ is given as in~\eqref{def:phibuy:eq3a},
\eqref{def:phibuy:eq3b} or~\eqref{def:phibuy:eq3c}. Now, the following
cases can arise.

\underline{Case 6.1:} Let $\zeta \geq \bar\zeta(\tau)$. In case there
exists a smallest $t^* \in [0,\tau)$ such that either
$\varphi = \phi_{\mathrm{sell}}(\tau-t^*,\zeta e^{-\kappa t^*})$ or
$\varphi = \phi_{\mathrm{buy}}(\tau-t^*,\zeta e^{-\kappa t^*})$ holds
true, we claim that the corresponding optimal strategy satisfies
$\hat{X}_t^{\tau,\zeta,\varphi,\uparrow} =
\hat{X}_t^{\tau,\zeta,\varphi,\downarrow} = 0$ on $[0,t^*]$ and is
then given by
$\hat{X}_{t-t^*}^{\tau-t^*,\zeta e^{-\kappa t^*},\varphi}$ on
$[t^*,\tau]$ as characterized in Step 1 or 3 above (i.e.,
Corollary~\ref{cor:main1} or Corollary~\ref{cor:main2}). Otherwise, we
obtain that
$\hat{X}_t^{\tau,\zeta,\varphi,\uparrow} =
\hat{X}_t^{\tau,\zeta,\varphi,\downarrow} = 0$ on $[0,\tau]$. Indeed,
by exploiting similar convexity arguments as above one can deduce that
$\nabla_{t}^{\uparrow,\downarrow} J_\tau(\hat{X}^{\tau,\zeta,\varphi})
> 0$ on $[0,t^*)$ and $[0,\tau)$, respectively. This implies optimality
of $\hat{X}^{\tau,\zeta,\varphi}$ via the dynamic programming
principle from Remark~\ref{rem:dynamicprog} and the first order
conditions from Proposition~\ref{prop:foc}. Moreover, if
$\varphi = \phi_{\mathrm{buy}}(\tau-t^*,\zeta e^{-\kappa t^*})$ it
must necessarily hold that
$\zeta e^{-\kappa t^*} \geq \bar\zeta(\tau-t^*)$ (i.e.,
$\phi_{\mathrm{buy}}$ is either given by~\eqref{def:phibuy:eq1}
or~\eqref{def:phibuy:eq2}) due to the monotonicity properties of
$\phi_{\mathrm{buy}}$ mentioned above.

\underline{Case 6.2:} Let $\zeta < \bar\zeta(\tau)$. If
$\varphi > \phi_{\mathrm{sell}}(0,\zeta e^{-\kappa\tau})$, there
exists a smallest $t^* \in [0,\tau)$ such that
$\varphi = \phi_{\mathrm{sell}}(\tau-t^*,\zeta e^{-\kappa t^*})$ holds
true. Analogously to Case 6.1, one can verify that the corresponding
optimal strategy satisfies
$\hat{X}_t^{\tau,\zeta,\varphi,\uparrow} =
\hat{X}_t^{\tau,\zeta,\varphi,\downarrow} = 0$ on $[0,t^*]$ and is
then given by
$\hat{X}_{t-t^*}^{\tau-t^*,\zeta e^{-\kappa t^*},\varphi}$ on
$[t^*,\tau]$ as characterized in Step 1. Otherwise, we have
$\hat{X}_t^{\tau,\zeta,\varphi,\uparrow} =
\hat{X}_t^{\tau,\zeta,\varphi,\downarrow} = 0$ on $[0,\tau]$.

In both, Case 6.1 and Case 6.2, we obtain that
$(\tau,\zeta,\varphi) \in \mathcal{R}_{\mathrm{wait}}$ as defined
in~\eqref{def:waitregion}. This finishes the proof of
Theorem~\ref{thm:main}, Corollary~\ref{cor:main1} and~\ref{cor:main2}.
\qed

\medskip


The following lemma summarizes some simple results which are used in the
proofs of Theorem~\ref{thm:main} and Corollary~\ref{cor:main2}.

\begin{Lemma} \label{lem:waitgradients} Let
  $(\tau,\zeta,\varphi) \in \cS$, $\tau \geq 0$, $\zeta > 0$, with
  corresponding optimal strategy
  $\hat{X}^{\tau,\zeta,\varphi} =
  (\hat{X}^{\tau,\zeta,\varphi,\uparrow},\hat{X}^{\tau,\zeta,\varphi,\downarrow})
  \in \cX^d$. For any $\theta > 0$ consider the problem data
  $(\tau+\theta,\zeta e^{\kappa \theta},\varphi) \in \cS$ and the
  strategy
  \begin{equation} \label{lem:waitgradients:cand}
    X_t^{\tau+\theta,\zeta e^{\kappa \theta},\varphi}\set
    \hat{X}_{t-\theta}^{\tau,\zeta,\varphi} 1_{[\theta, \tau +
      \theta]}(t) \quad (0 \leq t \leq \tau+\theta)
  \end{equation}
  in $\cX^d$ such that
  $\varphi^{X^{\tau+\theta,\zeta e^{\kappa \theta},\varphi}}_{0-} =
  \varphi$,
  $\zeta^{X^{\tau+\theta,\zeta e^{\kappa \theta},\varphi}}_{0-} =
  \zeta e^{\kappa \theta}$.
  \begin{enumerate} 
  \item Assume that
    ${}^\varrho\nabla^{\downarrow}_0J_{\tau}(\hat{X}^{\tau,\zeta,\varphi}) =
    0$. Then we have
    \begin{align} \label{lem:waitgradients:grad1} 
      \begin{aligned}
        & g^{\uparrow,\downarrow}(\theta;\tau,\zeta,\varphi) \set
        {}^\varrho\nabla^{\uparrow,\downarrow}_0
        J_{\tau+\theta}(X^{\tau+\theta,\zeta e^{\kappa
            \theta},\varphi}) \\
        & = \pm (\alpha \sigma^2 \varphi - \mu ) \theta + \frac{1}{2}
        \zeta (e^{\kappa \theta} \pm 1) + \frac{1}{2} \eta \vert
        \varphi^{\hat{X}^{\tau,\zeta,\varphi}}_{\tau} \vert
        (e^{-\kappa (\tau + \theta)} \pm e^{-\kappa \tau})
        \\
        &\;\;\;\; + \frac{1}{2} \eta (e^{-\kappa \theta} \pm 1)
        \int_{[0,\tau]} e^{-\kappa u}
        (d\hat{X}^{\tau,\zeta,\varphi,\uparrow}_u +
        d\hat{X}^{\tau,\zeta,\varphi,\downarrow}_u).
      \end{aligned}
    \end{align}
    The maps
    $\theta \mapsto
    g^{\uparrow,\downarrow}(\theta;\tau,\zeta,\varphi)$ are continuous
    and strictly convex on $(0,+\infty)$.
  \item Assume that $\tau = \varphi = 0$. Then we have
    \begin{align} \label{lem:waitgradients:grad2}
      \begin{aligned}
        h^{\uparrow,\downarrow}(\theta;\zeta,\varrho) \set
        {}^\varrho\nabla^{\uparrow,\downarrow}_0
        J_{\theta}(X^{\theta,\zeta e^{\kappa \theta},0}) = \mp \mu
        \theta + \frac{1}{2} \zeta \left( e^{\kappa \theta} \pm
          \varrho \right).
      \end{aligned}
    \end{align}
    The maps
    $\theta \mapsto h^{\uparrow,\downarrow}(\theta;\zeta,\varrho)$ are
    continuous and strictly convex on $(0,+\infty)$.
  \end{enumerate}
\end{Lemma}

\begin{proof}
  \emph{1.):} We only compute the mapping
  $\theta \mapsto g^{\uparrow}(\theta;\tau,\zeta,\varphi)$
  in~\eqref{lem:waitgradients:grad1}. The computation of
  $g^{\downarrow}$ is very similar and thus omitted. Hence, let
  $(\tau,\zeta,\varphi) \in \cS$ with associated optimal strategy
  $\hat{X}^{\tau,\zeta,\varphi}$. We have to compute the
  buy-subgradient of strategy
  $X^{\tau+\theta,\zeta e^{\kappa \theta},\varphi}$ in
  \eqref{lem:waitgradients:cand} at 0, i.e.,
  $\nabla^{\uparrow}_0 J_{\tau+\theta}(X^{\tau+\theta,\zeta e^{\kappa
      \theta},\varphi})=g^{\uparrow}(\theta;\tau,\zeta,\varphi) $.
  For notational convenience, we will henceforth write $X$ for the
  strategy $X^{\tau+\theta,\zeta e^{\kappa \theta},\varphi}$ and
  denote by $\varphi^{X}$, $\zeta^{X}$ the corresponding stock
  holdings and spread dynamics on $[0,\tau+\theta]$. By definition of
  the buy-subgradient in~\eqref{eq:buysubgradient} we obtain
  \begin{equation} 
    \begin{aligned} \label{lem:waitgradients:peq1}
      {}^{\varrho}\nabla^{\uparrow}_0
      J_{\tau+\theta}(X)& 
      = \int_{\theta}^{\tau+\theta} \kappa e^{-\kappa t} \zeta_t^{X}
      dt + \int_{\theta}^{\tau+\theta} (\alpha\sigma^2 \varphi^{X}_t -
      \mu) dt
      \\
      &\;\;\; + \kappa \int_0^{\theta} e^{-\kappa t} \zeta_t^X dt +
      \theta ( \alpha\sigma^2 - \mu \varphi)  \\
      & \;\;\; + \frac{1}{2}
      (\eta \vert \varphi^{X}_{\tau+\theta} \vert +
      \zeta^{X}_{\tau + \theta}) e^{-\kappa (\tau + \theta)} +
      \frac{\eta}{2} \varphi^{X}_{\tau+\theta} + \\
      & \;\;\; + \frac{1}{2} \sign_{\varrho}(\varphi^{X}_{\tau+\theta}) \zeta^{X}_{\tau +
        \theta}.
    \end{aligned}
  \end{equation}
  In addition, it holds that
  $0 =
  {}^{\varrho}\nabla^{\downarrow}_0J_{\tau}(\hat{X}^{\tau,\zeta,\varphi})
  = {}^{\varrho}\nabla^{\downarrow}_{\theta} J_{\tau+\theta}(X)$ which
  gives us the identity
  \begin{equation} \label{lem:waitgradients:peq2}
    \begin{aligned}
      \int_{\theta}^{\tau+\theta} (\alpha\sigma^2 \varphi^{X}_t - \mu)
      dt = & \int_{\theta}^{\tau+\theta} \kappa e^{-\kappa(t-\theta)}
      \zeta^{X}_t dt + \frac{1}{2} (\eta \vert
      \varphi^{X}_{\tau+\theta} \vert +
      \zeta^{X}_{\tau + \theta}) e^{-\kappa \tau} \\
      & - \frac{\eta}{2} \varphi^{X}_{\tau+\theta} - \frac{1}{2}
      \sign_{\varrho}(\varphi^{X}_{\tau+\theta}) \zeta^{X}_{\tau + \theta}.
    \end{aligned}
  \end{equation}
  Inserting~\eqref{lem:waitgradients:peq2} back
  into~\eqref{lem:waitgradients:peq1} and using the fact that
  $\zeta^X_t = \zeta e^{\kappa (\theta - t)}$ on $[0,\theta]$ yields
  \begin{equation} \label{lem:waitgradients:peq3}
    \begin{aligned}
      {}^{\varrho}\nabla^{\uparrow}_0 J_{\tau+\theta}(X) 
      & = \kappa (1+e^{\kappa \theta}) \int_{\theta}^{\tau+\theta}
      \zeta^{X}_t e^{-\kappa t} dt \\
& \;\;\; -\frac{1}{2} \zeta (e^{-\kappa
        \theta} - e ^{\kappa \theta}) + \theta
      (\alpha\sigma^2 \varphi - \mu)  \\
      & \;\;\; + \frac{1}{2} \eta \vert \varphi^{X}_{\tau+\theta}
      \vert (e^{-\kappa \tau} + e^{-\kappa(\tau+\theta)}) +
      \frac{1}{2} \zeta^{X}_{\tau + \theta} (e^{-\kappa \tau} +
      e^{-\kappa(\tau+\theta)}).
    \end{aligned}
  \end{equation}
  Next, applying the spread dynamics
  \begin{equation} \label{lem:waitgradients:peq4} \zeta^{X}_t = \zeta
    e^{-\kappa (t-\theta) } + e^{-\kappa (t-\theta) }
    \int_{[\theta,t]} \eta e^{\kappa (s-\theta)} (dX^{\uparrow}_s +
    dX^{\downarrow}_s) \quad (\theta \leq t \leq \tau+\theta)
  \end{equation} 
  and Fubini's Theorem, we finally obtain
  in~\eqref{lem:waitgradients:peq3} the representation
  \begin{align*}
    g^{\uparrow}(\theta;\tau,\zeta,\varphi) = 
    & \; (\alpha \sigma^2 \varphi -\mu) \theta + \frac{1}{2} \zeta
      (e^{\kappa \theta} + 1) + \frac{1}{2} \eta \vert \varphi^X_{\tau + \theta} \vert
      (e^{-\kappa (\tau + \theta)} + e^{-\kappa \tau})
    \\
    & + \frac{1}{2} \eta (1 + e^{-\kappa \theta}) 
      \int_{[\theta, \tau + \theta]} e^{\kappa (\theta - u)} (dX^{\uparrow}_u +
      dX^{\downarrow}_u). 
  \end{align*}
  Observing that
  $ \varphi^X_{\tau + \theta} =
  \varphi^{\hat{X}^{\tau,\zeta,\varphi}}_{\tau}$ and
  \begin{equation*}
    \int_{[\theta, \tau + \theta]} e^{\kappa (\theta - u)} (dX^{\uparrow}_u +
    dX^{\downarrow}_u) = \int_{[0,\tau]} e^{-\kappa u} (d\hat{X}^{\tau,\zeta,\varphi,\uparrow}_u +
    d\hat{X}^{\tau,\zeta,\varphi,\downarrow}_u)
  \end{equation*}
  yields the desired result in
  \eqref{lem:waitgradients:grad1}. Obviously, the map $g^{\uparrow}$
  is continuous in $\theta$. Moreover, it can be easily verified that
  the second dervative of $g^{\uparrow}$ with respect to $\theta$ is
  strictly positive which implies that
  $\theta \mapsto g^{\uparrow}(\theta;\tau,\zeta,\varphi)$ is strictly
  convex.

  \emph{2.)} Let $(0,\zeta,0) \in \cS$ with associated optimal
  strategy $\hat{X}^{0,\zeta,0} = (0,0)$ (recall also
  Remark~\ref{rem:tauzero}, 2.)). Using the definition
  in~\eqref{eq:buysubgradient} and~\eqref{eq:sellsubgradient}, the
  buy- and sell-subgradient of strategy
  $X^{\theta,\zeta e^{\kappa \theta},0}=(0,0)$ on $[0,\theta]$
  in~\eqref{lem:waitgradients:cand} can be readily computed as claimed
  in~\eqref{lem:waitgradients:grad2}. Strict convexity of the mappings
  $h^{\uparrow,\downarrow}$ follows as in 1.).
\end{proof}

\bibliographystyle{plainnat} \bibliography{finance}

\end{document}